%% file: main.tex
\documentclass[a4paper,11pt]{article}
\usepackage{amsmath,amssymb,amsthm}
\usepackage{thmtools}
\usepackage{thm-restate}
\usepackage[margin=2cm]{geometry}
\usepackage[utf8]{inputenc}
\usepackage[framemethod=tikz]{mdframed}
\usepackage{amsmath}
\usepackage[colorlinks=true,linkcolor=red,citecolor=blue]{hyperref}
\usepackage{amssymb}
\usepackage{tikz}
\usetikzlibrary{topaths,calc}
\usepackage{thm-restate,thmtools}
\usepackage{mathtools}
\usepackage{amsthm}
\usepackage{dsfont}
\usepackage{color}
\usepackage{amsmath}
\usepackage{amsfonts}
\usepackage{amssymb}
\usepackage{bbm}
\usepackage{relsize}
\usepackage{enumitem}
\usepackage[procnumbered, linesnumbered,algoruled]{algorithm2e}
\usepackage{array}

\usepackage[nameinlink]{cleveref}
\usepackage{lineno}

\newcommand{\ceil}[1]{{\left\lceil#1  \right\rceil}}

\newcommand{\cover}{\textnormal{\texttt{cover}}}
\newcommand{\LP}{\textnormal{\texttt{LP}}}

\newcommand{\vv}{\mathbf{v}}
\newcommand{\cF}{{\mathcal{F}}}
\newcommand{\cA}{{\mathcal{A}}}
\newcommand{\cB}{{\mathcal{B}}}
\newcommand{\cG}{{\mathcal{G}}}
\newcommand{\cI}{{\mathcal{I}}}

\newcommand{\cD}{{\mathcal{D}}}

\newcommand{\cR}{{\mathcal{R}}}

\newcommand{\cM}{{\mathcal{M}}}

\newcommand{\tol}{\textnormal{\texttt{tol}}}

\newcommand{\supp}{{\textnormal{supp}}}
\newcommand{\cP}{{\mathcal{P}}}
\newcommand{\cC}{{\mathcal{C}}}
\newcommand{\cS}{{\mathcal{S}}}
\newcommand{\cL}{{\mathcal{L}}}
\newcommand{\cK}{{\mathcal{K}}}
\newcommand{\bef}{{\bar{f}}}
\newcommand{\bd}{{\bar{d}}}
\newcommand{\bq}{{\bar{q}}}
\newcommand{\bmu}{{\bar{\mu}}}
\newcommand{\cU}{{\mathcal{U}}}
\newcommand{\bgam}{{\bar{\gamma}}}
\newcommand{\blam}{{\bar{\lambda}}}
\newcommand{\bx}{{\bar{x}}}
\newcommand{\bw}{{\bar{w}}}
\newcommand{\by}{{\bar{y}}}
\newcommand{\bz}{{\bar{z}}}
\newcommand{\bv}{{\bar{v}}}
\newcommand{\bu}{{\bar{u}}}
\newcommand{\bp}{{\bar{p}}}

\newcommand{\bbe}{{\bar{\beta}}}

\newcommand{\type}{{\textnormal{\texttt{type}}}}
\newcommand{\OPT}{\textnormal{OPT}}
\newcommand{\tw}{\tilde{\w}}

\newcommand{\eps}{{\varepsilon}}

\newcommand{\E}{{\mathbb{E}}}
\newcommand{\floor}[1]{\left\lfloor #1 \right\rfloor}

\newcommand{\abs}[1]{{\left| #1 \right|}}

\newcommand{\ALP}{\textnormal{Assign-LP}}

\newcommand{\w}{\mathbf{w}}
\renewcommand{\v}{\mathbf{v}}

\newtheorem{lemma}{Lemma}[section]

\newtheorem{obs}[lemma]{Observation}

\newtheorem{definition}[lemma]{Definition}
\newtheorem{theorem}[lemma]{Theorem}

\newtheorem{observation}[lemma]{Observation}
\newtheorem{claim}[lemma]{Claim}

\def \II   {{\mathcal I}}
\newcommand{\one}{\mathbbm{1}}
	\crefname{claim}{claim}{claims}
		\crefname{cor}{corollary}{corollaries}
\begin{document}

\begin{titlepage}

		\title{
			An EPTAS for Cardinality Constrained Multiple Knapsack via Iterative Randomized Rounding
		}
		
		\author{Ilan Doron-Arad\thanks{Computer Science Department, 
				Technion, Haifa, Israel. \texttt{idoron-arad@cs.technion.ac.il}}
			\and
			Ariel Kulik\thanks{Computer Science Department, 
				Technion, Haifa, Israel. \texttt{kulik@cs.technion.ac.il}}
			\and
			Hadas Shachnai\thanks{Computer Science Department, 
				Technion, Haifa, Israel. \texttt{hadas@cs.technion.ac.il}}
		}	

		\maketitle
		\thispagestyle{empty}
		
		\begin{abstract}
			\input{abstract}

		\end{abstract}
		\end{titlepage}

\input{intro}

	\input{overview}

\input{organization}

		\section{Preliminaries}
		\label{sec:defs}
		\input{defs}

		\section{The Algorithm}
		\label{sec:algorithm}
		\input{algorithm}

		\section{Linear Structure}
		\label{sec:structure}
		\input{structure}

\input{reduction}

\input{fewBins}

		\bibliographystyle{splncs04}
	
		\bibliography{bibfile}
\appendix

\section{Applications}
\label{sec:appendix}

\input{app_greedy}

	\end{document}

%% file: abstract.tex
In [Math. Oper. Res., 2011], Fleischer et al.
introduced a powerful technique  for solving the generic class of {\em separable assignment problems} (SAP), in which a set of items of given values and weights needs to be packed into a set of bins subject to separable assignment constraints, so as to maximize the total value.
The approach of Fleischer at al. relies on solving a configuration LP and sampling a configuration for each bin independently based on the LP solution. 
While there is a SAP variant for which this approach yields the best possible approximation ratio, for various special cases, there are discrepancies between the approximation ratios obtained using the above approach and  the state-of-the-art approximations. This raises the following natural question: Can we do better by {\em iteratively} solving the configuration LP and sampling a few bins at a time?

To assess the potential gain from iterative randomized rounding, we consider 
as a case study one interesting SAP variant, namely, 
{\sc Uniform Cardinality Constrained Multiple Knapsack}, for which we answer this question affirmatively.
 The input is a set of items, each has a value and a weight, and a set of uniform capacity bins. The goal is to assign a subset of the items of maximum total value to the bins such that $(i)$ the capacity of any bin is not exceeded, and $(ii)$ the number of items assigned to each bin satisfies a given {\em cardinality} constraint.
While the technique of Fleischer et al. yields a $\left(1-\frac{1}{e}\right)$-approximation for the problem, we show that iterative randomized rounding leads to {\em efficient polynomial time approximation scheme (EPTAS)}, thus essentially resolving the complexity status of the problem. Our analysis of iterative randomized rounding
can be useful for solving other SAP variants.

%% file: intro.tex
\section{Introduction}
We consider problems in the class of maximizing assignment problems with packing constraints, also known as {\sc separable assignment problems (SAP)}. 
A general problem in this class
is defined by a set of bins and a set of items to be packed in the bins. There is a value $v_{ij}$ associated with assigning item $i$ to bin $j$. We are also given a separate packing constraint for each bin $j$. The goal is to find an assignment of a subset of the items to the bins which maximizes the total value accrued. This class includes several well studied problems such as the {\sc generalized assignment problem (GAP)}.

In~\cite{fleischer2011tight}, Fleischer at al. introduced a powerful technique
for solving SAP and its variants. The technique relies on first solving a  configuration linear programming (configuration-LP) relaxation of the problem. Subsequently, 
configurations (i.e., feasible subsets of items for a single bin) are sampled independently according to a distribution specified by the LP solution to obtain an integral solution for the given instance. For many SAP variants, such as GAP, the approximation guarantee of the resulting algorithm is $\left(1-\frac{1}{e}\right)$.

Intuitively, we can do better using the following iterative randomized rounding approach: 
iteratively solve a configuration LP relaxation of the problem for the remaining items and bins and sample a {\em few} configurations based on the distribution specified by the LP solution, until all bins are used.
We note that if the LP is solved only once and  all of the bins are packed based on the solution, then we have exactly the algorithm of Fleischer et al.~\cite{fleischer2011tight}. 
This raises the following question:

\medskip
\begin{mdframed}[backgroundcolor=white!2]
	\centering
Can iterative randomized rounding improve the approximation ratio of \cite{fleischer2011tight}? 
\end{mdframed}
\medskip

As the authors of~\cite{fleischer2011tight} show, 
under standard complexity assumptions, there is a SAP variant for which their approximation ratio of $\left(1-\frac{1}{e}\right)$ is the best possible. 
However, for various special cases (such as  {\sc multiple knapsack} and GAP), there are discrepancies between the approximation guarantee obtained using the algorithm of \cite{fleischer2011tight} and  the state-of-the-art approximations. This suggests the iterative approach can potentially lead to improved approximation for some variants (compared to \cite{fleischer2011tight}). 
To assess the potential gain from iterative randomized rounding,
 we focus as a case study on one interesting SAP variant, namely,
 {\sc Uniform Cardinality Constrained Multiple Knapsack (CMK)}. Specifically, we show that iterative randomized rounding  
 is superior to the technique of \cite{fleischer2011tight} and use it  to essentially resolve the complexity status of this problem.

An input for CMK consists of a set of items, each has a value and a weight, and a set of uniform capacity bins.
The goal is to assign a subset of the items of maximum total value to the bins such that $(i)$ the total weight of items in each bin does not exceed its capacity, and $(ii)$ the number of items assigned to each bin satisfies a given cardinality constraint.\footnote{See a more formal definition in \Cref{sec:overview}.} 
CMK has real-world applications in cloud computing, as well as in manufacturing systems and radio networks 
(see \Cref{sec:appendix}).

\subsection{Related Work}

Iterative randomized rounding of configuration-LPs has been recently used for {\em vector bin packing}~\cite{KMS23}. In this problem, the goal is to pack a set of items, each given by a $d$-dimensional size vector for some $d >1$, in a minimum number of $d$-dimensional bins, where a subset of items fits in a bin if it adheres to the capacity constraints in all dimensions. We are not aware of an application of iterative randomized rounding of configuration-LPs for {\em maximization} problems.

The {\sc multiple knapsack with uniform capacities} (UMK) problem  is the special case of CMK with no cardinality constraint, or equivalently, where the cardinality constraint is larger than the total number of items. In the more general {\sc multiple knapsack} (MK) problem, the capacity of the bins may be arbitrary. In terms of approximation algorithms, UMK and MK are well understood. A {\em polynomial time approximation scheme} (PTAS) for UMK was given by Kellerer \cite{kellerer1999polynomial}. Later, Chekuri and Khanna \cite{chekuri2005polynomial} developed the first PTAS for MK and ruled out the existence of a {\em fully} PTAS (FPTAS), already for UMK with only two bins. Jansen designed more involved {\em efficient} PTAS (EPTAS) for MK~\cite{jansen2010parameterized,jansen2012fast}, thus resolving the complexity status of the problem.\footnote{We give formal definitions of approximation schemes in \Cref{sec:defs}.}

For CMK, a randomized  $\left(1-\frac{1}{e}\right)$-approximation follows from the above mentioned results of Fleischer et al.~\cite{fleischer2011tight} for SAP. More specifically,
the authors present a randomized algorithm for SAP whose approximation guarantee is
$\left(\left(1-\frac{1}{e}\right) \cdot \beta \right)$, where $\beta$ is the best approximation ratio for the single bin subproblem.\footnote{The paper~\cite{fleischer2011tight} shows that the existence of an FPTAS for the single bin subproblem, as in the case of $0/1$-knapsack with cardinality constraint, implies a $\left(1-\frac{1}{e}\right)$-approximation for the corresponding variant of SAP.}
 A slightly more efficient approximation ratio for CMK follows from a recent result of Cohen et al. \cite{CKS23} who give a randomized $\left(1-\frac{\ln (2)}{2}-\eps\right)\approx 0.653$-approximation for {\em uniform $2$-dimensional vector multiple knapsack}. In this problem, the cardinality constraint of CMK is generalized to a second knapsack constraint. %

We note that a PTAS for CMK can be obtained using ideas of~\cite{chekuri2005polynomial}. We outline the main steps.
First, item values are discretized into $O (\log n)$ value classes, where $n$ is the number of items. Then, enumeration is used to roughly determine the number of items taken from each value class. Clearly, one should take from each value class the items with smallest weights, leading to a reduced problem of packing  sufficient items from each value class in the $m$ given (uniform) bins. Packing these items in $(1+\eps) \cdot m+O(1)$ uniform bins can be done using an asymptotic FPTAS for bin packing with cardinality constraint~\cite{EL10}. Finally, the algorithm keeps the $m$ bins with highest values. We note that the running time of the enumeration step is very high. This leaves open the question whether CMK admits an EPTAS.

\subsection{Our Results}

Our main contribution is in showing that iterative randomized rounding can substantially improve the approximation guarantee of the configuration-LP rounding approach of \cite{fleischer2011tight}. The analysis is based on concentration bounds; thus,
our iterative algorithm is applied to
a slightly restricted subclass of CMK instances in which the number of bins is large w.r.t. the given error parameter. Recall that even for two bins the problem does not admit an FPTAS \cite{chekuri2005polynomial}. Hence, we do not expect the iterative approach to work for a small number of bins. More specifically, given an error parameter $\eps \in (0,0.1)$ we say that a CMK instance $\cI$ is $\eps$-{\em simple} if  (i) $m> \exp(\exp(\eps^{-30}))$, and (ii) $\eps \cdot m\in \mathbb{N}$, where $m$ is the number of bins and $\OPT(\cI)$ is the optimum value of $\cI$.\footnote{In our discussion of $\eps$-simple instances, we did not attempt to optimize the constants.} 
For clarity, we first state our algorithmic result for the subclass of $\eps$-simple CMK instances
(see \Cref{sec:algorithm} for more details). 

\begin{theorem}
	\label{thm:simple}
For every $\eps \in (0,0.1)$ and an $\eps$-simple \textnormal{CMK} instance $\cI$, iterative randomized rounding (see \Cref{alg:randomized_rounding})  returns a $(1-\eps)$-approximation for $\cI$ in time $ \left( \frac{\abs{\cI}}{\eps} \right)^{O(1)}$, where $\abs{\cI}$ is the encoding size of $\cI$. 
\end{theorem}
An in depth look into the algorithm of Fleischer et al. \cite{fleischer2011tight} reveals that the approximation ratio of their algorithm on $\eps$-simple instance is no better than $\left(1-\frac{1}{e}\right)$, indicating the improved ratio in \Cref{thm:simple} stems from the use of the iterative approach. 

We give a simple reduction showing that our algorithm for $\eps$-simple instances yields a randomized EPTAS for general CMK instances. 
This essentially resolves the complexity status of CMK, since an FPTAS is ruled out \cite{chekuri2005polynomial}.

\begin{theorem}
\label{thm:main}
There is a randomized \textnormal{EPTAS} for \textnormal{CMK}. 
\end{theorem}  

For the proof of \Cref{thm:main} see \Cref{sec:algorithm}.

%% file: overview.tex
\subsection{Technical Overview}
\label{sec:overview}

In the following we outline our algorithmic approach and its analysis. For clarity, we focus in this section on high-level ideas and omit some technical details. We start with a more formal definition of CMK. An instance of CMK consists of a set of items  $I$, a weight function $\w:I\rightarrow [0,1]$, a value function $\v:I\rightarrow \mathbb{R}_{\geq 0}$, a number of bins $m\in \mathbb{N}_{>0}$, and a cardinality constraint $k\in \mathbb{N}_{>0}$. 
A {\em solution} is a tuple $(C_1,\ldots,C_m)$  such that for all $j \in \{1,\ldots, m\}$ it holds that $C_j\subseteq I$,  $|C_j| \leq k$ and $\w(C_j) = \sum_{i \in C_j} \w(i) \leq 1$.  
The {\em value} of the solution $(C_1,\ldots,C_m)$ is $\sum_{i \in S} \v(i)$ where $S= \bigcup_{j=1}^{m} C_j$. The goal is to find a solution of maximum value.  Note that we allow the sets $C_1,\ldots, C_m$ to intersect, but if an item appears in multiple sets its value is counted only once.

\paragraph{The Algorithm.} Our algorithm applies an iterative randomized rounding approach based on a configuration-LP. 
The use of such linear program dates back to the work of Karmarkar and Karp  on bin packing \cite{karmarkar1982efficient}, and such linear programs are commonly used in approximation algorithms for resource allocation problems (e.g., \cite{BCS09,BEK16,fleischer2011tight,jansen2010parameterized,KMS23,CKS23}).

We use a {\em configuration polytope} $P(\ell)\subseteq [0,1]^{I}$, where $\by\in P(\ell)$ can be intuitively interpreted as ``there is a way to fractionally pack the items into $\ell$ bins such that each item $i\in I$ is packed $\by_i$ times''. The algorithm takes as an input a value $\eps>0$ which serves as a discretization factor and determines the approximation ratio. Our iterative approach uses~$\frac{1}{\eps}$ iterations, and each iteration packs $\eps\cdot m$ of the remaining  bins.  Therefore, at the beginning of the  $j$-th iterations, $(j-1)\cdot\eps \cdot m$ bins were packed  (in previous iterations) and $(1-(j-1)\cdot \eps)\cdot m$ bins are still empty. We use $S_j$ to denote the set of items that were not packed by the end of the $j$-th iteration (and thus are still available for packing). The main steps of the algorithm are  as follows.

\begin{enumerate}
	\item Initialize $S_0\leftarrow I$ to be all the items.
	\item For $j$ from $1$ to $\frac{1}{\eps}$ do:
	\begin{enumerate}
		\item 
		\label{item:intro_lp}
		Solve the linear program
		\begin{equation}
			\label{eq:overview_lp}
		\begin{aligned}
			&\textnormal{ max }~~&&\sum_{i\in I} \by_i \cdot\v(i) \\ 
			&\textnormal{ s.t. } && \by \in P\left( m\cdot (1-(j-1)\cdot \eps )\right) \\
			&&& \by_i = 0 ~~~~&\forall i\in I\setminus S_{j-1}\\
		\end{aligned}
		\end{equation}
		 That is, we want to obtain a maximum value using $m\cdot (1-(j-1)\cdot \eps)$ bins and only items in $S_{j-1}$. Let $\by^j$ be the solution found.

		\item 
		\label{item:intro_rounding}Apply randomized sampling (rounding) to pack items into  $\eps\cdot m$ bins according to the solution $\by^j$; update $S_j$ to be $S_{j-1}$ minus all the items packed in the current iteration. 
	\end{enumerate}
	\item Return the collection of $m$ packed bins.
\end{enumerate}
The linear program in \eqref{eq:overview_lp} can be approximated efficiently but cannot be solved exactly in polynomial time. For the purpose of this technical overview, we assume it can be solved exactly.
In \Cref{item:intro_rounding} we use
the randomized rounding technique for configuration LPs of Fleischer et al. \cite{fleischer2011tight}. The same randomized rounding technique is commonly used by other algorithms (e.g., \cite{BCS09,BEK16,KMS23,BK14}). We give the full details on the sampling process in \Cref{sec:algorithm}. While sampling of configurations based on configuration-LP has been used for several decades, its iterative application has been proposed only recently~\cite{KMS23}. Furthermore, the iterative approach has not been used previously for maximization problems.

\paragraph{High Level Analysis.} %
 Let $Q_j$ be the set of items packed in the $j$-th iteration (that is, $Q_j = S_{j}\setminus S_{j-1}$). 
In the $j$-th iteration, the linear program uses $m\cdot ( 1-(j-1)\cdot \eps)$ bins and attains value of $\sum_{i\in I }\by^j_i \cdot \v(i)$, with an average value of $\frac{\sum_{i\in I }\by^j_i \cdot \v(i)}{m\cdot ( 1-(j-1)\cdot \eps)} $ per bin. As the number of bins sampled in each iteration is {\em small}, it can be shown that   the average value per bin in the packing generated by the randomized rounding is roughly the same as the average value in the fractional solution. That is,
$$
\E\left[ \frac{ \sum_{i\in Q_j} \v(i)} {\eps \cdot m} \right]~\approx ~\E\left[\frac{ \sum_{i\in I } \by^j_i \cdot \v(i)}{m\cdot ( 1-(j-1)\cdot \eps)}\right],
$$
or equivalently,
$$
\E\left[\sum_{i\in Q_j} \v(i)\right]~\approx ~\E\left[ \eps \cdot \frac{ \sum_{i\in I } \by^j_i \cdot \v(i)}{1-(j-1)\cdot \eps} \right]. 
$$
Observe the left hand term is the expected value attained from items packed in the $j$-th iteration.  
In each iteration of the algorithm the distribution by which the bins are sampled is updated, so the algorithm  does not pack items already packed in previous iterations (by the constraints  $\by_i=0$ for $i\in I\setminus S_{j-1}$ in~\eqref{eq:overview_lp}).
 Thus, we have that $Q_1,\ldots, Q_{\eps^{-1}}$ are disjoint. 
It follows that the value of the solution returned by the algorithm is
\begin{equation}
	\label{eq:overview_sol_profit}
	\E\left[\v(I\setminus S_{\eps^{-1}})\right] \,= \,  \E\left[
	\sum_{j=1}^{\eps^{-1}} \sum_{i\in Q_j} \v(i) \right]\,\approx \,\E\left[\eps\cdot  \sum_{j=1}^{\eps^{-1}} \frac{ \sum_{i\in I } \by^j_i \cdot \v(i)}{1-(j-1)\cdot \eps}\right].
\end{equation}

Ideally, we would like the average value per bin to be (at least) $\frac{\OPT}{m}$ in each of the solutions $\by^j$,  where $\OPT$ is the value of the optimal solution of the instance. 
That is, the average value per bin in each of the iterations remains the average value per bin in the optimum. As $\by^j$ conceptually uses $m\cdot (1-(j-1)\eps)$ bins, this implies that
\begin{equation}
	\label{eq:overview_frac_value}
	\E\left[\sum_{i\in I } \by^j_i \cdot \v(i)\right] \, \gtrsim  \,  \frac{\OPT}{m}\cdot m\cdot (1-(j-1)\cdot \eps) \, =\, \OPT\cdot  \left(1-(j-1)\cdot  \eps\right),
\end{equation}
for every $j \in [\eps^{-1}]$. 
If we assume \eqref{eq:overview_frac_value} holds and plug it into \eqref{eq:overview_sol_profit}, we get that the value of the solution returned by the algorithm is
$$
\E\left[\v(I\setminus S_{\eps^{-1}})\right] \,
\approx\, \E\left[\eps\cdot \sum_{j=1}^{\eps^{-1}} \frac{ \sum_{i\in I } \by^j_i \cdot \v(i)}{1-(j-1)\cdot \eps} \right]\, \gtrsim\,  \eps \cdot \sum_{j=1}^{\eps^{-1}} \frac{\OPT \cdot (1-(j-1)\cdot \eps)}{1-(j-1)\cdot \eps} \,= \, \OPT.
$$
That is, the algorithm returns a solution of expected value close to $\OPT$, assuming \eqref{eq:overview_frac_value} holds. This leaves us with the goal of showing that \eqref{eq:overview_frac_value} holds with high probability.

\paragraph{Linear Structures and Equation~\eqref{eq:overview_frac_value}.}
To show that \eqref{eq:overview_frac_value} holds we define a random vector $\bgam^j\in [0,1]^I$ for every $j\in [\eps^{-1}]$. 
We use $\bgam^j$ to lower bound the value of the configuration-LP. We show that (i) the expected value of $\bgam^j$, $\E\left[\sum_{i\in I} \bgam^j_i\cdot \v(i)\right]$,  is $(1-(j-1)\eps)\cdot \OPT$, and (ii) with high probability $\bgam^j\in P((1+\delta )\cdot m_j)$, where $m_{j} = (1-(j-1)\cdot \eps)\cdot m$ is the number of remaining bins at the beginning of the $j$-th iteration and $\delta>0$ is small.
Once properties (i) and (ii) are shown, it follows that that $\frac{\bgam^j}{1+\delta}$ is a solution of high value for the linear program in the $j$-th iteration, and \eqref{eq:overview_frac_value} immediately follows as the algorithm finds an optimal solution in every iteration.
Property (i) is shown using a simple calculation of the expected value of the vector $\bgam^j$. Showing property (ii) is more challenging.

The polytope $P(\ell)$ can be represented via a finite set of linear constraints $\cS\subseteq \mathbb{R}^I_{\geq 0}$ by $P(\ell ) = \{\by \in [0,1]^{I} \,|\,\forall \bu \in \cS:\, \bu\cdot \by \leq \ell\}$ (the set $\cS$ is the same for every $\ell$). While $\cS$ is finite, its size is non-polynomial in the input instance.  A naive approach to show that $\bgam^j \in P( (1+\delta)m_j)$ is to  consider each constraint $\bu \in \cS$ separately, and  apply  concentration bounds to show  $\by\cdot \bu \lesssim  m_j$ with high probability. Subsequently, the union bound can be used to lower bound the probability that  $\by\cdot \bu \lesssim  m_j$ for every $\bu \in \cS$ simultaneously.  However, 
due to the large number of vectors in $\cS$, a direct application of the union bound does not lead to such useful lower bound.

We use a {\em linear structure} to overcome the above challenge. The linear structure provides an approximate representation of the configuration polytope using a small number of constraints (that is, the number of constraints only depends on $\eps$). 
As the number of constraints is reduced, we can now apply the above logic successfully $-$ use a concentration bound to show that each constraint of the linear structure holds independently with high probability, and then apply the union bound to show that  all the constraints hold simultaneously with high probability.  By the properties of the linear structure, once we show that all   constraints  hold, we are guaranteed  that $\bgam^j \in P((1+\delta)\cdot m_j)$, as stated in~(ii).

The concept of linear structure was introduced in \cite{KMS23}. It is essentially a non-constructive version of the {\em subset oblivious} algorithms used by the Round\&Approx framework of \cite{BCS09}. We construct the linear structure for CMK based on ideas from 
\cite{KMS23,BEK16}. The structure leverages the relatively simple structure of the cardinality constraint.

\paragraph{Technical Contribution.} In this paper we use for the first time iterative randomized rounding of configuration-LP to solve a maximization problem. 
As such, the paper provides the basic foundations required for the analysis of iterative randomized rounding for maximization problems. Iterative randomized rounding of a configuration-LP has been recently used in~\cite{KMS23} for solving bin packing problems. While in some parts the analysis only  requires simple adaptations of ideas from \cite{KMS23}, in other parts the adaptation is more challenging.

These challenges arise mainly due to the fact that in bin packing all the remaining items must be fully packed by the configuration-LP; however, in maximization problems the remaining items may be partially selected or not selected at all by the configuration-LP. 
Thus, the probability of an item to be packed after $j$ iterations may take different values for different items. In contrast, this probability is the same for all items in the case of bin packing.
Similarly, while in the case of bin packing all items must be packed by the configuraiton-LP in every iteration, in maximization problem there is a degree of freedom in the selection of items to be packed. This, in turn, led to a different approach for the use of the linear structure.

We note that while the paper~\cite{CKS23} deals with a generalization of CMK and uses several similar concepts (configuration LP, sampling, subset oblivious algorithms), the algorithm in \cite{CKS23} does not use an iterative approach. It relies on two separate stages: the first stage uses randomized rounding of a configuration-LP that is solved once,  and the second stage uses a combinatorial algorithm. We believe the analysis of the iterative randomized rounding algorithm presented in this paper can be used to show the technique yields an improved approximation for the Uniform $2$-dimensional Vector Multiple Knapsack ($2$d-UMK) problem considered in \cite{CKS23}. A main challenge is that 
our analysis relies on a robust linear structure, which is unlikely to exist for $2$d-UMK (as this would lead to a PTAS, contradicting the hardness results in~\cite{CKS23}). This can potentially be tackled by defining a suitable analog of the linear structure for $2$d-UMK and adapting the analysis to this new structure.

%% file: organization.tex
\subsection{Organization}
In \Cref{sec:defs} we give some definitions and notation. \Cref{sec:algorithm} presents our main algorithm and an outline of its analysis. In \Cref{sec:analysisFull} we give the detailed analysis (proofs of Lemmas~\ref{lem:Upper_Bound_V}, \ref{lem:gamma_value_lowerbound}, \ref{lem:value_in_terms_of_gamma} and \ref{lem:main_analysis}). \Cref{sec:structure} is devoted to the proof of \Cref{lem:structure}, and \Cref{sec:reduction} gives the proof of \Cref{lem:reduction}.

%% file: defs.tex
We start with some definitions and notation.
Let $\OPT(I)$ be the value of an optimal solution for an instance $I$ of a maximization problem~$\Pi$. For $\alpha\in (0,1]$,   a solution~$x$ for the instance $I$ is  an $\alpha$-approximate solution  if its value is at least $\alpha \cdot \OPT(I)$.
For $\alpha \in (0,1]$, we say that $\cA$ is an $\alpha$-approximation algorithm
for $\Pi$ if for any instance $I$ of $\Pi$, $\cA$ outputs an $\alpha$-approximate solution for $I$. An algorithm $\cA$ is a {\em randomized} $\alpha$-approximation for $\Pi$ if for any instance $I$ of $\Pi$ it always returns a solution for $I$, 
and the solution is an $\alpha$-approximate solution with probability at least~$\frac{1}{2}$. A {\em polynomial-time approximation scheme} (PTAS) for a maximization problem $\Pi$ is a family of algorithms $(\cA_{\eps})_{\eps>0}$ such that for any $\eps>0$, $\cA_{\eps}$ is a polynomial-time $(1-\eps)$-approximation algorithm for $\Pi$. As the running time of a PTAS may be impractically high,
two restrictive classes of PTAS have been proposed in the literature: 
$(\cA_{\eps})_{\eps>0}$ is an {\em efficient} PTAS (EPTAS) if the running time of $\cA_{\eps}$ 
is of the form $f\left(\frac{1}{\eps} \right) \cdot n^{O(1)}$, where $f$ is an arbitrary function, and $n$ is the bit-length encoding size of the input instance;
$(\cA_{\eps})_{\eps>0}$ is a {\em fully} PTAS (FPTAS) if the running time of $\cA_{\eps}$ is bounded by $\left(\frac{n}{\eps}\right)^{O(1)}$. 
Given a boolean expression $\cal D$, we define $\one_{\cal D} \in \{0,1\}$ such that $\one_{\cal D} =1$ if $\cal D$ is true and $\one_{\cal D} =0$ otherwise.

We give an alternative definition of our problem that will be used in the technical sections.  
An instance of CMK
is a tuple $\cI= (I,\w,\v,m,k)$, where $I$ is a set of items, $\w:I\rightarrow [0,1]$ is the weight function, $\v:I\rightarrow \mathbb{R}_{\geq 0}$ is the value function, $m\in \mathbb{N}_{>0}$ is the number of bins, and $k\in \mathbb{N}_{>0}$ is the cardinality constraint. 
A {\em configuration}  of the instance $\cI$ is $C\subseteq I$ such that $|C|\leq k$ and $\w(C)=\sum_{i\in C} \w(i)\leq 1$. 
Let  $\cC_{\cI}$ be the set of all configurations of $\cI$, and $\cC_{\cI}(i) =\{C\in \cC~|~i\in C\}$ the set of all configurations which contain $i\in I$. When clear from the context, we simply use $\cC = \cC_{\cI}$ and $\cC(i) = \cC_{\cI}(i)$.

A {\em solution} of $\cI$ is a tuple of $m$ configurations $S = (C_1,\ldots,C_m)\in \cC^m$. 
The value of the solution $S = (C_1,\ldots,C_m)$ is $\v(S) = \v\left(\bigcup_{b\in [m]} C_b \right)$ (generally, for any set $B\subseteq A $ and a function $f:A \rightarrow \mathbb{R}_{\geq 0}$, we use $f(B) = \sum_{b \in B} f(b)$). The objective is to find a solution of maximum value. Let $\OPT(\cI)$ be 
the optimal solution value for the instance $\cI$,
and $|\cI|$ the encoding size of $\cI$. W.l.o.g., we consider a tuple with fewer than $m$ configurations to be a solution. In this case, for some $r \leq m$, the tuple $(C_1,\ldots, C_{r})\in \cC^r$ is equivalent to the solution $(C_1,\ldots, C_{r},\emptyset,\ldots, \emptyset)\in \cC^m$.

Our main algorithm, given in \Cref{sec:algorithm}, is applied to
a restricted subclass of {\em simple} instances. %
We now give a more formal definition for this subclass of instances. 
\begin{definition}
	\label{def:simple}
	Let $\eps \in (0, 0.1)$,
	We say that a CMK instance $\cI= (I,\w,\v,m,k)$ is {\em $\eps$-simple} if the following conditions hold.
	\begin{itemize}

		\item $m> \exp(\exp(\eps^{-30}))$ %
		\item $\eps \cdot m\in \mathbb{N}$. 
	\end{itemize}
\end{definition}

We give a reduction showing that our algorithm for $\eps$-simple instances yields a randomized EPTAS for general CMK instances.\footnote{In our discussion of $\eps$-simple instances, we did not attempt to optimize the constants.} This is formalized in the next lemma (we give the proof in \Cref{sec:reduction}).
\begin{lemma}
\label{lem:reduction}
Given $\eps \in (0, 0.1)$ such that $\eps^{-\frac{1}{2}}\in \mathbb{N}$, let $\cA$ be a randomized algorithm  which	returns a $(1-\eps)$-approximate solution for any $\eps$-simple CMK instance $\cI$ in time $\left(\frac{|\cI|}{\eps}\right)^{O(1)}$. Then, there is a randomized \textnormal{EPTAS} for CMK.
\end{lemma}
\Cref{thm:main} follows from \Cref{thm:simple} and \Cref{lem:reduction}.

%% file: algorithm.tex
In this section we formally present our iterative randomized rounding algorithm for $\eps$-simple CMK instances. The algorithm relies on a linear programming (LP) relaxation of CMK that we formalize through the notion of fractional solutions.

A {\em fractional solution} for an instance $\cI= (I,\w,\v,m,k)$ is a vector $\bx\in \mathbb{R}_{\geq 0}^{\cC}$; the value $\bx_C$ represents a fractional selection of the configuration $C$ for the solution. The {\em coverage} of $\bx$ is the vector $\cover(\bx) \in \mathbb{R}_{\geq 0}^I$ defined by
$$
\forall i\in I:~~~\cover_i(\bx)\,=\, \left( \cover(\bx)\right)_i\,=\, \sum_{C\in \cC(i)} \bx_C.$$
The vector $\bx$ is {\em feasible} if $\cover(\bx) \in[0,1]^I$. The size of $\bx$ is $\|\bx\| = \sum_{C\in \cC} \bx_C$ (throughout this paper, for every vector $\bz\in \mathbb{R}^n$ we use $\|\bz\| = \sum_{i=1}^n\abs{\bz_i}$).
The {\em value} of $\by \in [0,1]^I$ is $\v(\by ) = \sum_{i\in I} \by_i \cdot \v(i)$.
 The {\em value} of $\bx$ is  the value of the cover of $\bx$, that is, $\v(\bx) =\v(\cover(\bx))$. For $\ell \in \mathbb{N}_{>0}$, let $[\ell] = \{ 1, \ldots , \ell \}$.
 
  A solution $S=(C_1,\ldots, C_m)$ for $\cI$, where $C_1,\ldots, C_m$ are disjoint and non-empty, can be encoded as a feasible fractional solution $\bx\in \{0,1\}^{C}$ defined by $\bx_{C_b} = 1$ for every $b\in [m]$, and $\bx_C=0$ for every other configuration. It is easy to verify that $\|\bx\|= m$, $\cover_i(\bx)=1$ for every $i\in S$, $\cover_i(\bx)=0$ for every $i\in I\setminus S$,  and $\v(\bx) = \v(S)$.
 
We use fractional solutions to define a linear program (LP). 
Let $K$ be a set and  $\bgam \in \mathbb{R}^K$. The {\em support} of $\bgam$ is $\supp(\bgam) = \left\{i\in K~|~\bgam_i \neq 0\right\}$. Let $\cI= (I,\w,\v,m,k)$ be a  CMK instance. 
For every set $S\subseteq I$ of remaining  items  and $\ell\in \mathbb{N}$ remaining bins, we define the configuration LP of $S$ and $\ell$ by
$$
\LP(S,\ell):~~~
\begin{aligned}
	&\textnormal{ max }~~&&\v(\bx) \\
	&\textnormal{ s.t. } && \textnormal{ $\bx$ is a feasible fractional solution for $\cI$} \\
	&&& \supp(\bx) \subseteq 2^{S}\\
	&&&\|\bx\|= \ell
	\end{aligned}
$$
That is, in $\LP(S,\ell)$ exactly $\ell$ configurations are selected\footnote{Note that $\bx_{\emptyset}$ may be greater than $1$.}, and these configurations contain only items in $S$. We can formally define the  configuration polytope $P(\ell)$ discussed in \Cref{sec:overview} via  fractional solutions by \begin{equation}
\label{eq:configuration_polytope}
P(\ell) = \left\{ \cover(\bx)\,|\, \textnormal{$\bx$ is a feasible fractional solution for $\II$ and $\|\bx\|\leq \ell$}\right\}.
\end{equation}
It can be shown that $\LP(S_j, (1-(j-1)\cdot \eps)\cdot m)$ is equivalent to the linear program in  \eqref{eq:overview_lp}.

A generalization of $\LP(S,\ell)$ for the {\em separable assignment problem} (SAP) was considered in \cite{fleischer2011tight}. Given $p_i \geq 0$ for every $i \in I$, the paper~\cite{fleischer2011tight}  shows that linear programs such as $\LP(S,\ell)$ admit an FPTAS whenever the {\em single bin problem} $-$ of finding $C\in \cC$ such that $\sum_{i\in I} p_i$ is maximized $-$ admits an FPTAS. 
As the single bin
case of CMK has an FPTAS~\cite{caprara2000approximation}, we get the following.
\begin{lemma}
	\label{lem:lp_fptas}
	There is an algorithm which given a CMK instance $\cI=(I,\w,\v,m,k)$, $S\subseteq I$, $\ell\in \mathbb{N}$  and $\eps>0$, finds a $(1-\eps)$-approximate solution for $\LP(S,\ell)$ in time $\left(\frac{|\cI|}{\eps}\right)^{O(1)}$.
\end{lemma}

Given  a fractional solution  $\bx$ such that $\|\bx\|\neq 0$,
we say that a random configuration $R\in \cC$ is {\em distributed by $\bx$}, and write $R\sim \bx$, if $\Pr(R=C) = \frac{\bx_C}{\|\bx\|}$ for all $C\in \cC$. 

The pseudocode of our algorithm for CMK is  given in 
\Cref{alg:randomized_rounding}. 
In each iteration $1 \leq j \leq \eps^{-1}$, 
the algorithm uses 
the solution $\bx^j$ for  $\LP(S_{j-1},m_j)$
to sample $\eps \cdot m$ configurations, where $S_{j-1}$ is the set of items remaining after iteration $(j-1)$,
and $m_j$ is the number of remaining (unassigned) bins.

\begin{algorithm}[h]
	\caption{\textsf{Iterative Randomized Rounding}}
	\label{alg:randomized_rounding}
	\SetKwInOut{Input}{input}
	\SetKwInOut{Output}{output}
	\SetKwInOut{Configuration}{configuration}
	
	\Input{Error parameter $\eps \in (0, 0.1)$, $\eps^{-\frac{1}{2}} \in \mathbb{N}$, and an $\eps$-simple CMK instance $\cI=(I,\w,\v,m,k)$}
	\Output{A solution for the instance}
	
	Initialize  $S_0\leftarrow I$
	
	\For{$j=1,\ldots,\eps^{-1}$ \label{randomized_rounding:loop}}{
		\label{randomized_rounding:solve}
		Find a $(1-\eps)$-approximate solution $\bx^j$  for $\LP(S_{j-1}, m_j)$, where $m_j = m\left(1-(j-1)\cdot \eps \right)$. 
		
		Sample independently $q = \eps \cdot m$ configurations $R^j_1,\ldots ,R^j_{q} \sim \bx^j$. 
		
		Update $S_j = S_{j-1} \setminus\left(\bigcup_{b=1}^{q} R^j_b \right)$.
		
	}
	Return the solution $\left(R^j_b\right)_{1\leq j \leq \eps^{-1},~1\leq b \leq q}$
	
\end{algorithm}

Consider the execution of \Cref{alg:randomized_rounding} with the input $\cI=(I,\w,\v,m,k)$ and $\eps\in (0,0.1)$ such that $\eps^{-\frac{1}{2}}\in \mathbb{N}$. 
The notations we use below, such as $\bx^j$, $S_j$, and $R^j_b$, refer to the 
variables used throughout the execution of the algorithm.
Clearly, \Cref{alg:randomized_rounding} returns a solution for~$\cI$. Furthermore, by \Cref{lem:lp_fptas}, the running time of the algorithm is polynomial in $\cI$ and $\eps^{-1}$.
 Let $V= \v\left(\bigcup_{j=1}^{\eps^{-1}} \bigcup_{b=1}^{q} R^j_b\right) = \v\left(I\setminus S_{\eps^{-1}}\right)$ be the value of the returned solution.

 \paragraph{Main Lemmas.} 
 In the following, we describe the main lemmas we prove  in order to lower bound the value of~$V$. The proofs of the lemmas are given in \Cref{sec:analysisFull,sec:structure}.
  
A simple calculation shows that the expected value of $\v(R^j_b)$, given all the samples up to (and including) iteration $(j-1)$, is $\frac{\v(\bx^j) }{m\cdot (1-(j-1)\cdot \eps)}$. To compute the expected value of 
$\v\left( \bigcup_{b=1}^{q}R^j_b\right)$, we need to take into consideration 
events in which an item $i\in I$ appears in several configurations among $R^j_1,\ldots, R^j_q$.
In \Cref{sec:upper_bound_v} we show that, since only a small number of configurations are sampled in each iteration (in comparison to the overall remaining number of bins), such events have small effect on the expected value (with the exception of the last  $\eps^{-\frac{1}{2}}$ iterations). This observation is coupled with a  concentration bound to prove the next lemma.

\begin{restatable}{lemma}{upperboundv}
	\label{lem:Upper_Bound_V}
	It holds that $$\E[V]\,=\,\E\left[\v(I\setminus S_{\eps^{-1}})\right] \,\geq \,\left( \eps-\eps^{\frac{3}{2}} \right)\cdot \sum_{j=1}^{\eps^{-1} - \eps^{-\frac{1}{2}}} \cdot \frac{\E[ \v(\bx^j)]}{1-(j-1)\eps}.$$
\end{restatable}
\Cref{lem:Upper_Bound_V} is the formal statement of
 \eqref{eq:overview_sol_profit}. 
\Cref{lem:Upper_Bound_V}
essentially reduces the problem of deriving a lower bound for $V$ to obtaining a lower bound on $\v(\bx^j)$.

To obtain a lower bound for $\v(\bx^j)$ we use the following steps. We define random vectors $\bgam^j \in [0,1]^I$ for every $j\in [\eps^{-1}]$ such that $\v(\bgam^j)$ is  high, and there is $\bz^j$ such that $\cover(\bz^j) =\bgam^{j-1}$ and $\|\bz^j\| \approx m_j$. We scale down $\bz^j$ to obtain a solution for $\LP(S_{j-1},m_j)$ of value $\approx \v(\bgam^{j-1})$, and consequently get a lower bound for $\v(\bx^j)$. 
We use a {\em linear structure} defined below, to show the existence of $\bz^j$. We further use auxiliary random vectors $\blam^j$ to define $\bgam^j$.

Let $(\Omega,\cF,\Pr)$ be the probability space defined by the execution of the algorithm.
Define the  $\sigma$-algebras $\cF_0 =\{\emptyset, \Omega\}$ and $\cF_j=
\sigma\left(\{ R^{j'}_b~|~1\leq j'\leq j,~1\leq b\leq q\}\right)$. That is, $\cF_j$ describes events which only depend on the outcomes of the random sampling up to (and including) the $j$-th iteration  
of the algorithm. We follow the standard definition of conditional probabilities and expectations given $\sigma$-algebras (see, e.g., \cite{CT97}).

Fix an optimal solution $(C^*_1,\ldots, C^*_m)$ for the instance and let $S^*=\bigcup_{j=1}^{m} C^*_j$ be the set of items in this solution.  Also, given a set $S\subseteq I$  denote by $\one_S$ the vector $\bz \in \{0,1\}^I$ satisfying $\bz_i =1$ for $i \in S$, and $\bz_i =0$ otherwise. 

 We define $\bgam^j$ and $\blam^j$ inductively using $S^*$.  
Define $\bgam^0 = \one_{S^*}$, that is $\bgam^0_i=1$ for every $i\in S^*$ and $\bgam^0_i=0$ for every $i\in I\setminus S^*$. For every $j\in [\eps^{-1}-1]$ define  	$\blam^{j}\in \mathbb{R}_{\geq 0} ^I$  by
\begin{equation}
	\label{eq:blam_def}
\blam^{j}_i =\frac{1-j\cdot \eps}{1-(j-1)\eps}\cdot \frac{1}{\Pr(i\in S_j~|~\cF_{j-1})} \cdot \bgam^{j-1}_i
\end{equation}
for all $i\in S_{j-1}$ and $\blam^{j}_i = 0$ for $i \notin S_{j-1}$.   Intuitively, the expression   $\Pr(i\in S_j\, |\, \cF_{j-1})$ in \eqref{eq:blam_def}  is the probability that item $i$ will still be available for packing after the $j$-th iteration, where the probability is calculated at the end of the iteration $j-1$. 
 Also, for every $j\in [\eps^{-1}-1]$ define~$\bgam^{j}\in \mathbb{R}_{\geq 0} ^I$  by 
 \begin{equation}
 	\label{eq:bgam_def}
 	\bgam^j_i = \one_{i\in S_{j}} \cdot \blam^j_i~~~~~\forall i\in I.
 	\end{equation}
 Observe that $\blam^j$ is $\cF_{j-1}$-measurable random variable whereas $\bgam^j$ is $\cF_j$-measurable. Intuitively, this means that the value of $\blam^j$ is known by the end of the $(j-1)$-th iteration, while the value of $\bgam^j$ is only known known by the end of the $j$-th iteration.

The lower bound on $\v(\bgam^{j-1})$ relies on a simple calculation of expectations.  By induction it can be shown that $\E[\bgam^{j-1}_i] = (1-(j-1)\eps)\cdot \one_{i\in S^*}$. Therefore, we can prove the following. 
\begin{restatable}{lemma}{gammaval}
	\label{lem:gamma_value_lowerbound}
For every $j\in [\eps^{-1}]$ it holds that 
$$\E[\v(\bgam^{j-1})] = (1-(j-1)\eps)\cdot \v(S^*) =(1-(j-1)\eps)\cdot\OPT(\cI)$$
\end{restatable}
We give the proof of \Cref{lem:gamma_value_lowerbound} in \Cref{sec:gamma_val}. 

Our next challenge is to show that there is a solution for $\LP(S_{j-1},m_j)$ whose cover is roughly $\bgam^{j-1}$, which can be alternatively stated as $\bgam^{j-1}\in P(\ell)$ where $\ell \approx m_j$, and $P(\ell)$ is as defined in~\eqref{eq:configuration_polytope}.  
To this end, we introduce a {\em linear structure} for  CMK. The main idea in linear structures is that they allow us to determine that $\bgam^j \in P(\ell)$ by checking if $\bgam^j$ satisfies a small number of linear inequalities.  

Given a vector $\bu \in \mathbb{R}_{\geq 0}^I$ which defines an inequality in the linear structure, we use concentration bounds to show that $\bgam^j\cdot \bu \leq \E[ \bgam^j\cdot \bu]+\xi$, where $\xi$ is an error terms. 
The concentration bounds we use only provide useful  guarantees  if the error term $\xi$ is of order of $\tol(\bu) = \max\left\{ \sum_{i\in C} \bu_i \,|\,C\in \cC\right\} $. We refer to the value $\tol(\bu)$ as the {\em tolerance} of $\bu$.  
We consequently  require the linear structure to be robust to additive errors of order of the tolerance. 
Also, we say that $S\subseteq I$ {\em can be packed into $\ell\in\mathbb{N}$ bins} if there are  $\ell$ configurations $C_1,\ldots, C_{\ell}\in \cC$ such that $\bigcup_{b=1}^{\ell} C_b = S$.

\begin{definition}[Linear Structure]
	\label{def:structure} Let $(I,\w,\v,m,k)$ be a CMK instance and $\delta>0$ a parameter. 
	Also, consider a subset $S\subseteq I$ such that $S$ can be packed in $\ell\in \mathbb{N}$ bins. A {\em $\delta$-linear structure} of $S$ is a set of vectors $\cL\subseteq \mathbb{R}^{I}_{\geq 0}$ 
	which satisfy the following property. 
	\begin{itemize}
		\item 
		Let $\by \in \left([0,1]\cap\mathbb{Q}\right)^I$, $0<\alpha<1$ and $t>0$, such that 
		\begin{enumerate}
			\item  $\supp(\by) \subseteq S$
			\item $\forall \bu \in \cL:~~~ \bu \cdot \by \leq \alpha \cdot \bu \cdot \one_{S} +t\cdot \tol(\bu)$ 
		\end{enumerate} 
		Then, there is a fractional solution  $\bx$ whose cover is $\by$ and $\|\bx\|\leq    \alpha \cdot \ell + 20\delta \ell  +(t+1)\cdot \exp(\delta^{-5})$. 
	\end{itemize} The {\em size} of the structure $\cL$ is $|\cL|$.
\end{definition}
Alternatively,  a $\delta$-linear structure guarantees for $S$ that for every $\by \in [0,1]^I$ with rational entries,  $0<\alpha <1$ and $t>0$, if $\supp(\by)\subseteq S$ and $\by$ satisfies $\abs{\cL}$ linear inequalities, then $\by\in P(\alpha \cdot \ell + 20\delta \ell + (t+1)\cdot \exp(\delta^{-5}))$.

In \Cref{sec:structure} we prove the next result.
\begin{lemma}
	\label{lem:structure}
	Given $\delta>0$, 
	let $I=(I,\w,\v,m,k)$ be a CMK instance, and $S\subseteq I$ a subset which can be packed into $\ell>\exp(\delta^{-5})$ bins. Then there is a $\delta$-linear structure $\cL$ of $S$ of size at most $\exp\left(\delta^{-4}\right)$. 
\end{lemma}
The above lemma is an adaptation of a construction of~\cite{KMS23} used to solve the vector bin packing problem, in which there are additional requirements for the packing of $S$.  Our adaptation leverages the relative simplicity of a cardinality constraint to omit these additional requirements.

We use \Cref{lem:structure} to show the existence of an $\eps^2$-linear structure of $S^*$, where $S^*$ is the set of items in an optimal solution.  We use the linear structure to show the existence of a fractional solution $\bz^j$ such that $\cover(\bz^j)=\bgam^{j-1}$ and $\|\bz^j\|\approx (1-(j-1)\eps)m$ for every $j\in [\eps^{1}]$. A simple scaling is then used to construct a solution for $\LP(S_{j-1},m_j)$ and establish the following lower bound on $\v(\bx^j)$.
	\begin{restatable}{lemma}{xingamma}
	\label{lem:value_in_terms_of_gamma}
	With probability at least $1-\exp(-\eps^{-20})$, it holds that 
	$$
	\forall j\in [\eps^{-1}]:~~~~
	\v(\bx^j) \,\geq \, (1-\eps)\cdot\left( 1- \frac{30\cdot \eps^2}{1-(j-1)\eps}\right)\cdot  \v(\bgam^{j-1}).
	$$
\end{restatable}
We give the proof of \Cref{lem:value_in_terms_of_gamma} in \Cref{sec:in_gamma}. Together, \Cref{lem:value_in_terms_of_gamma} and \Cref{lem:gamma_value_lowerbound} essentially give the formal proof of \eqref{eq:overview_frac_value}. 

Finally, using \Cref{lem:Upper_Bound_V,lem:value_in_terms_of_gamma,lem:gamma_value_lowerbound}, we obtain the next result, whose proof is given in \Cref{sec:combine}.
\begin{restatable}{lemma}{approximation}
	\label{lem:main_analysis}
$\E[V]\,\geq \, (1-300\cdot \eps^{\frac{1}{2}})\cdot \OPT(\II)$.
\end{restatable}
\Cref{thm:simple} follows directly from \Cref{lem:main_analysis}.

\section{The Analysis}
\label{sec:analysisFull}
Consider an execution of \Cref{alg:randomized_rounding} with the input $\cI = (I,\w,\v,m,k)$ and $\eps>0$. We use the notation and definitions as given in \Cref{sec:algorithm}. Also,  let $\by^j = \cover(\bx^j)$ be the coverage of $\bx^j$. Observe $\bx^j$ and $\by^j$ are $\cF_{j-1}$-measurable. That is, their values are determined by the outcomes of the samples up to (and including) the $j-1$ iteration. 
As in \Cref{sec:algorithm} we let  $(C^*_1,\ldots, C^*_m)$ be an optimal solution for the instance $\cI$. We define $S^* =\bigcup_{b=1}^{m} C^*_b$ and $\OPT= \vv(S^*)=\OPT(\cI)$.

\subsection{Concentration Bounds}
\label{sec:concentration_bounds}

Before we give the proofs of \Cref{lem:Upper_Bound_V,lem:value_in_terms_of_gamma,lem:gamma_value_lowerbound,lem:main_analysis}, we need to introduce some concentration bounds for {\em self-bounding functions}.
\begin{definition}
	\label{def:self_bounding}
	A non-negative function $f:\mathcal{X}^n\rightarrow \mathbb{R}_{\geq0}$ is called self-bounding if there exist $n$ functions $f_1,\ldots,f_n:\mathcal{X}^{n-1}\rightarrow \mathcal{R}$ such that for all $x=(x_1,\ldots,x_n)\in\mathcal{X}^n$, $$\begin{aligned}&0\leq f(x)-f_t(x^{(t)}) \leq 1, \textnormal{~~~~~~~and~~}\\ &\sum_{t=1}^{n}\left(f(x)-f_t(x^{(t)})\right) \leq f(x),
	\end{aligned}
	$$ where $x^{(t)}=(x_1,\ldots,x_{t-1},x_{t+1},\ldots,x_n)\in\mathcal{X}^{n-1}$ is obtained by dropping the $t$-th component of $x$. 
\end{definition}

We rely on the following concentration bound due to Boucheron, Lugosi and Massart~\cite{BCS09}.
\begin{lemma}
	\label{lem:selfBoundLem}
	Let $f:\mathcal{X}^n \rightarrow \mathbb{R}_{\geq0}$ be a self-bounding function and let $X_1,\ldots,X_n\in \mathcal{X}$ be independent random variables. Define $Z=f(X_1,\ldots,X_n)$. Then the following holds:
	\begin{enumerate}
		\item $\Pr\left(Z\geq \E[Z]+t\right)\leq \exp\left(-\frac{t^2}{2\cdot\E[Z]+\frac{t}{3}}\right)$, for every $t\geq0$.
		\item $\Pr\left(Z\leq \E[Z]-t\right)\leq \exp\left(-\frac{t^2}{2\cdot\E[Z]}\right)$, for every $t>0$.
	\end{enumerate}
\end{lemma}

The setting considered in \cite{BCS09} can be trivially extended to a setting in which the random variable are conditionally independent on a $\sigma$-algebra $\cG$ (see \cite{CT97} for the definition of conditional independence) and  the function $f$ itself is a $\cG$-measurable random function.
This is formally stated in the next lemma.
\begin{lemma}
	\label{lem:generalized_concentration}
	Let $(\Omega, \cF, \Pr)$ be a finite probability space and let $\cG\subseteq \cF$ be a $\sigma$-algebra. Let $D$ be a finite set of self-bounding function from $\chi^\ell$ to $\mathbb{R}_{\geq 0}$ and let $f\in D$ be a $\cG$-measurable random function. Also, let $ X_1,\ldots, X_{\ell} \in \chi$ be random variables which are conditionally independent given $\cG$.  Define $Z=f(X_1,\ldots,X_n)$. Then the following holds:
	\begin{enumerate}
		\item $\Pr\left(  Z\geq \E[Z\,|\,\cG]+t\,|\,\cG\right)\leq \exp\left(-\frac{t^2}{2\cdot \E[Z\,|\,\cG]+\frac{t}{3}}\right)$, for every $t\geq0$.
		\item $\Pr\left(Z\leq \E[Z\,|\,\cG]-t\,|\,\cG\right)\leq \exp\left(-\frac{t^2}{2\cdot\E[Z\,|\,\cG]}\right)$ for every $t\geq 0$.
	\end{enumerate}
\end{lemma}
The generalization in \Cref{lem:generalized_concentration} is required since the variables $R^j_1,\ldots, R^j_q$ are dependent for $q>1$ while being conditionally independent given the variables $R^{j'}_{b}$ for every $j'<j$ and $b\in [q]$.
The following construction for self-bounding function was shown in \cite{CKS23}. 
\begin{lemma}
	\label{lem:boundFunc}
	Let $\mathcal{I}=(I,\w,\v,m,k)$ be a CMK instance, and $h:I\rightarrow \mathbb{R}_{\geq0}$. For some $\ell \in \mathbb{N}_{>0}$ define $f:\mathcal{C}^{\ell}\to\mathbb{R}_{\geq 0}$ by $f(C_1,\ldots,C_\ell)=\frac{h(\bigcup_{i\in [\ell]}C_i)}{\eta}$  where
	$\eta \geq \max_{C\in\mathcal{C}}h(C)$.
	Then $f$ is  self-bounding.
\end{lemma}

\subsection{The proof of \Cref{lem:Upper_Bound_V}}
\label{sec:upper_bound_v}
The first step towards the proof of \Cref{lem:Upper_Bound_V} is to show a lower bound on the probability of an item to appear in one of the sampled configurations $R^j_1,\ldots, R^j_q$ in terms of $\by^j_i$.
\begin{lemma}
	\label{lem:selection_item_prob}
	For every $i\in I$ and $j\in \left[\eps^{-1}\right]$ it holds that 
	$\Pr\left( i\in S_{j-1}\setminus S_j ~\middle|~\cF_{j-1} \right) \geq 1-\exp\left(- \eps \cdot \frac{\by^j_i}{1-(j-1)\eps}\right)$.
\end{lemma}

\begin{proof}
	By a simple calculation,
	\begin{equation}
		\label{eq:11112}
		\begin{aligned}
		\Pr\left( i\in S_{j-1}\setminus S_j ~\middle|~\cF_{j-1} \right) ={} & 	\Pr\left( i\in \bigcup_{b=1}^{q} R^j_b ~\middle|~\cF_{j-1} \right) \\
		={} & 1- \Pr\left( i\notin \bigcup_{b=1}^{q} R^j_b ~\middle|~\cF_{j-1} \right)\\
		={} & 1- \prod_{b=1}^{q} \Pr\left( i\notin R^j_b~\middle|~\cF_{j-1} \right)\\
			={} & 1- \prod_{b=1}^{q} \left(1-\Pr\left( i\in R^j_b~\middle|~\cF_{j-1} \right) \right)\\
				={} & 1- \prod_{b=1}^{q} \left(1-\Pr\left(R^j_b \in \cC(i)~\middle|~\cF_{j-1} \right) \right).
		\end{aligned}
	\end{equation} 
	The third equality holds as $R^j_1,\ldots, R^j_q$ are conditionally independent given $\cF_{j-1}$.
	Therefore, by \eqref{eq:11112} and since the configurations are distributed by $\bx^j$ we have 
	\begin{equation*}
	\label{eq:1}
	\begin{aligned}
		\Pr\left( i\in S_{j-1}\setminus S_j ~\middle|~\cF_{j-1} \right) ={} 
		& 1- \prod_{b=1}^{q} \left(1-\Pr\left(R^j_b \in \cC(i)~\middle|~\cF_{j-1} \right) \right)\\
		={}& 1- \left(1-\frac{\sum_{C \in \cC(i)} \bx^j_C}{\|\bx^j\|}\right)^{q}\\
		={} & 1- \left(1-\frac{\by^j_i}{m \cdot \left(1-(j-1)\cdot \eps \right)}\right)^{\eps \cdot m}\\
			={} & 1-  \left(  \left(1-\frac{\by^j_i}{m \cdot \left(1-(j-1)\cdot \eps \right)}\right)^{\frac{m \cdot \left(1-(j-1)\cdot \eps \right)}{\by^j_i}} \right)^{\frac{\eps \cdot \by^j_i}{\left(1-(j-1)\cdot \eps \right)}}\\
				\geq{} & 1-  \left( e^{-1} \right)^{\frac{\eps \cdot \by^j_i}{\left(1-(j-1)\cdot \eps \right)}}\\
				={} & 1- \exp\left(-\frac{\eps \cdot \by^j_i}{\left(1-(j-1)\cdot \eps \right)}\right).\\
	\end{aligned}
\end{equation*} The inequality holds since $(1-\frac{1}{x})^x \leq \frac{1}{e}$ for all $x \geq 1$. 
\end{proof}

The next lemma uses	\Cref{lem:selection_item_prob} to lower bound the total value of sampled configurations in the $j$-th iteration. 
\begin{lemma}
	\label{lem:value_to_LP}
	For all $j\in \left[\eps^{-1}-\eps^{-\frac{1}{2}}\right]$ it holds that 
	$$
	\E\left[\v(S_{j-1}\setminus S_j ) ~|~\cF_{j-1} \right] \geq \v(\bx^j) \cdot \left(\eps-\eps^{\frac{3}{2}} \right) \frac{1}{1-(j-1)\eps}.
	$$
\end{lemma}
\begin{proof} By \Cref{lem:selection_item_prob} we get
	\begin{equation}
		\label{eq:ex_j}
		\begin{aligned}
			\E\left[\v(S_{j-1}\setminus S_j ) ~|~\cF_{j-1} \right] = {} & 	\sum_{i \in I} \v(i) \cdot \Pr\left( i\in S_{j-1}\setminus S_j ~\middle|~\cF_{j-1} \right) \\
			\geq {} &  	\sum_{i \in I} \v(i) \cdot \left( 1-\exp\left(- \eps \cdot \frac{\by^j_i}{1-(j-1)\eps}\right) \right)\\
			\geq{} & \sum_{i \in I} \v(i) \cdot \left( \eps \cdot \frac{\by^j_i}{1-(j-1)\eps}-\left(\eps \cdot \frac{\by^j_i}{1-(j-1)\eps}\right)^2\right) \\
				={} & \sum_{i \in I} \v(i) \cdot \left( \eps \cdot \frac{\by^j_i}{1-(j-1)\eps} \cdot \left( 1-\eps \cdot \frac{\by^j_i}{1-(j-1)\eps}\right) \right).\\
		\end{aligned}
	\end{equation} 
	The second inequality follows from $1-\exp(-x) \geq x-x^2$ for all $x\geq 0$. By \eqref{eq:ex_j} we have	\begin{equation*}
	\begin{aligned}
		\E\left[\v(S_{j-1}\setminus S_j ) ~|~\cF_{j-1} \right] 
							\geq{} & \sum_{i \in I} \v(i) \cdot \left( \eps \cdot \frac{\by^j_i}{1-(j-1)\eps} \cdot \left(1- \eps \cdot \frac{1}{1-(\eps^{-1}-\eps^{-\frac{1}{2}}-1)\eps} \right)\right)\\
								={} & \sum_{i \in I} \v(i) \cdot \left( \eps \cdot \frac{\by^j_i}{1-(j-1)\eps} \cdot \left(1-  \frac{\eps}{\eps+\eps^{\frac{1}{2}}} \right)\right)\\
									={} & \frac{1}{1-(j-1)\eps} \cdot \left(\eps-  \frac{\eps^2}{\eps+\eps^{\frac{1}{2}}} \right) \cdot \sum_{i \in I} \v(i) \cdot \by^j_i\\
									={} &  \frac{1}{1-(j-1)\eps} \cdot \left(\eps-  \frac{1}{\eps^{-1}+\eps^{-\frac{3}{2}}} \right) \cdot \v(\bx^j)\\
										\geq{} &  \v(\bx^j) \cdot \left(\eps-\eps^{\frac{3}{2}} \right) \frac{1}{1-(j-1)\eps}. 
	\end{aligned}
\end{equation*} The first inequality holds  since   $j \leq \eps^{-1}-\eps^{-\frac{1}{2}}$ and since $\bx^j$ is a feasible solution for $\LP(S, m_j)$; thus, $\by^j \in [0,1]^I$. 
\end{proof}

The proof of \Cref{lem:Upper_Bound_V} follows from \Cref{lem:value_to_LP} and 
\upperboundv*
\begin{proof}
$$
\begin{aligned}
	\E[V]\, &=\, \E\left[\sum_{j=1}^{\eps^{-1}} \v(S_{j-1}\setminus S_j)\right] \\
	&=\, \E\left[~ \sum_{j=1}^{\eps^{-1}}\E\left[ \v(S_{j-1}\setminus S_j)\,\middle|\, \cF_{j-1}\right]~\right] \\
	&\geq \,  \E\left[~ \sum_{j=1}^{\eps^{-1}-\eps^{-\frac{1}{2}}} \E\left[\v(S_{j-1}\setminus S_j)\,\middle|\, \cF_{j-1}\right]~\right] \\
	&\geq \,   \E\left[~ \sum_{j=1}^{\eps^{-1}-\eps^{-\frac{1}{2}}} \v(\bx^j) \cdot \left(\eps-\eps^{\frac{3}{2}} \right) \frac{1}{1-(j-1)\eps}~\right] \\
	&=\,  \left(\eps-\eps^{\frac{3}{2}} \right) \sum_{j=1}^{\eps^{-1}-\eps^{-\frac{1}{2}}}  \frac{ \E[\bv(\bx^j)]}{ 1-(j-1)\eps}.
\end{aligned}
$$
The second equality follows from the tower property. Note that the range of $j$ changes in the first inequality. The second inequality is due to \Cref{lem:value_to_LP}.
\end{proof}

\subsection{The proof of \Cref{lem:value_in_terms_of_gamma}}
\label{sec:in_gamma}

While the proof of \Cref{lem:Upper_Bound_V} required us to lower bound the probability that items are selected by the sampled configurations, in order to prove \Cref{lem:value_in_terms_of_gamma} we show the opposite -- we  upper bound the probability that items are selected by the sampled configuration. Since non-selected items can be used in the solution of $\LP(S_{j-1},m_j)$  in later iterations, these upper bounds are useful to show  lower bounds on the value of $\LP(S_{j-1},m_j)$.
\begin{lemma}
	\label{lem:survival_prob}
	For every $i\in I$ and $j\in [\eps^{-1}]$ it holds that $$\Pr(i\in S_j~|~\cF_{j-1} ) \geq \one_{i\in S_{j-1}} \cdot \left( 1-\frac{\eps\cdot \by^j_i}{1-(j-1)\eps}\right) \geq \one_{i\in S_{j-1}} \cdot \frac{1-j\eps}{1-(j-1)\eps} .$$ 
	\end{lemma}
	
	\begin{proof} We use the following inequality. 
			\begin{equation}
			\label{eq:111112}
			\begin{aligned}
				\Pr(i\in S_j~|~\cF_{j-1} ) 
				={} & 	\Pr\left( i\in S_{j-1} \setminus \bigcup_{b=1}^{q} R^j_b ~\middle|~\cF_{j-1} \right) \\
					={} & 	\one_{i\in S_{j-1}}  \cdot \Pr\left(i \notin \bigcup_{b=1}^{q} R^j_b ~\middle|~\cF_{j-1}\right) \\
						={} & 	\one_{i\in S_{j-1}}  \cdot \left(1-\Pr\left(i \in \bigcup_{b=1}^{q} R^j_b ~\middle|~\cF_{j-1}\right) \right)\\
							\geq{} & 	\one_{i\in S_{j-1}}  \cdot \left(1-\sum_{b=1}^{q} \Pr\left(i \in  R^j_b ~\middle|~\cF_{j-1}\right) \right)\\
			\end{aligned}
		\end{equation} 
	The inequality uses the union bound. Therefore, by \eqref{eq:111112} and since the configurations are selected based on the distribution $\bx^j$ we have 
	\begin{equation*}
		\label{eq:1}
		\begin{aligned}
			\Pr\left( i\in S_j ~\middle|~\cF_{j-1} \right) \geq{} & \one_{i\in S_{j-1}}  \cdot \left(1-\frac{\eps \cdot m \cdot \sum_{C \in \cC(i)} \bx^j_C}{\|\bx^j\|}\right)\\
				={} & \one_{i\in S_{j-1}}  \cdot \left(1-\frac{\eps \cdot m \cdot \by^j_i}{m \cdot \left(1-(j-1)\cdot \eps \right)}\right)\\
				={} & \one_{i\in S_{j-1}}  \cdot \left(1-\frac{\eps \cdot \by^j_i}{ 1-(j-1)\cdot \eps }\right)\\
				\geq{} & \one_{i\in S_{j-1}}  \cdot \left(1-\frac{\eps}{ 1-(j-1)\cdot \eps }\right)\\
						={} &	\one_{i\in S_{j-1}} \cdot \frac{1-j\eps}{1-(j-1)\eps}. 
		\end{aligned}
	\end{equation*} %
The last inequality holds since the coverage of the feasible solution $\bx^j$ satisfies $\by^j \in [0,1]^I$. 
	\end{proof}

Recall the vectors $\bgam^j$ and $\blam^j$ has been define in \eqref{eq:blam_def} and \eqref{eq:bgam_def} (\Cref{sec:algorithm}). For readability purposes, we repeat these definitions.
	Define $\bgam^0 = \one_{S^*}$ (recall $S^*$ is the set of items in an optimal solution). For every $j\in [\eps^{-1}-1]$ define  	$\blam^{j}\in \mathbb{R}_{\geq 0} ^I$  by
	$$
	\blam^{j}_i =\frac{1-j\cdot \eps}{1-(j-1)\eps}\cdot \frac{1}{\Pr(i\in S_j~|~\cF_{j-1})} \cdot \bgam^{j-1}_i$$
	for all $i\in S_{j-1}$ and $\blam^{j}_i = 0$ for $i \notin S_{j-1}$. Also, for every $j\in [\eps^{-1}-1]$ define  	$\bgam^{j}\in \mathbb{R}_{\geq 0} ^I$  by $\bgam^j_i = \one_{i\in S_{j}} \cdot \blam^j_i$ for all $i\in I$	.  It holds  that $\blam^j$ is $\cF_{j-1}$-measurable random variable while $\bgam^j$ is $\cF_j$-measurable.
	
	\begin{lemma}
		\label{lem:gamma_expectation}
	For every $i\in I$ and $j\in [\eps^{-1}-1]$ it holds that 
	$\E\left[ \bgam^j_i ~|~\cF_{j-1}\right] = \frac{1-j\cdot \eps}{1-(j-1)\cdot \eps}\cdot \bgam^{j-1}_i$.
	\end{lemma}
	
	\begin{proof} By the definition of the vectors $\bgam^j$ and $\blam^j$ we have 
		\begin{equation}
			\label{eq:sB2}
			\begin{aligned}
				\E\left[ \bgam^j_i ~|~\cF_{j-1}\right] ={} & \E\left[ \one_{i\in S_{j}} \cdot \blam^j_i ~|~\cF_{j-1}\right] \\
				={} & \Pr \left( i \in S_j ~|~\cF_{j-1}\right) \cdot \blam^j_i\\
			={} & \begin{cases}
					\Pr \left( i \in S_j ~|~\cF_{j-1}\right)  \cdot \frac{1-j\cdot \eps}{1-(j-1)\eps}\cdot \frac{1}{\Pr(i\in S_j~|~\cF_{j-1})} \cdot \bgam^{j-1}_i & ~~~~~~~~~~\text{if } i \in S_{j-1}, \\
				\Pr \left( i \in S_j ~|~\cF_{j-1}\right) \cdot 0 & ~~~~~~~~~~ \text{if } i \notin S_{j-1}
			\end{cases}\\
			={} & \begin{cases}
		\frac{1-j\cdot \eps}{1-(j-1)\eps} \cdot \bgam^{j-1}_i & ~~~~~~~~~~\text{if } i \in S_{j-1}, \\
			\frac{1-j\cdot \eps}{1-(j-1)\eps} \cdot \bgam^{j-1}_i & ~~~~~~~~~~ \text{if } i \notin S_{j-1}
		\end{cases}\\
	={} & 	\frac{1-j\cdot \eps}{1-(j-1)\eps} \cdot \bgam^{j-1}_i
			\end{aligned}
		\end{equation} The fourth equality holds since if $ i\notin S_j$ then $\bgam^{j-1}_i = 0$ by definition. 
	\end{proof}

	\begin{lemma}
	\label{lem:ZO}
	For all $j\in \left[1,\ldots, \eps^{-1}-1\right]$ and $i \in I$ it holds that $\blam^j_i,\bgam^j_i \leq 1$.  
\end{lemma}

\begin{proof}
	We prove the claim by induction on $j$. Let $j = 1$. Then,%
	 \begin{equation}
		\label{eq:Bq1}
		\blam^{1}_i =\frac{1-j\cdot \eps}{1-(j-1)\eps}\cdot \frac{1}{\Pr(i\in S_j~|~\cF_{j-1})} \cdot \bgam^{j-1}_i \leq \frac{1-j\cdot \eps}{1-(j-1)\eps}\cdot \frac{1}{ \frac{1-j\eps}{1-(j-1)\eps}} \cdot \bgam^{j-1}_i = \bgam^{j-1}_i \leq 1.
	\end{equation} The first inequality follows from \Cref{lem:survival_prob}. The last inequality holds since by definition it holds that $\bgam^0 \in \{0,1\}^I \subseteq [0,1]$. Therefore, by \eqref{eq:Bq1} it follows that $\bgam^1_i = \one_{i\in S_{1}} \cdot \blam^1_i \leq 1$. Assume that for some $j\in \left[2,\ldots, \eps^{-1}-1\right]$ it holds that $\blam^{j-1}_i,\bgam^{j-1}_i \leq 1$.  Then, in case $i\in S_{j-1}$ it holds that
 \begin{equation}
	\label{eq:Bq2}
	\blam^{j}_i =\frac{1-j\cdot \eps}{1-(j-1)\eps}\cdot \frac{1}{\Pr(i\in S_j~|~\cF_{j-1})} \cdot \bgam^{j-1}_i \leq \frac{1-j\cdot \eps}{1-(j-1)\eps}\cdot \frac{1}{ \frac{1-j\eps}{1-(j-1)\eps}} \cdot \bgam^{j-1}_i = \bgam^{j-1}_i \leq 1.
\end{equation} The first inequality follows from \Cref{lem:survival_prob}. The last inequality holds since by the assumption of the induction. Also, in case $i\notin S_{j-1}$ then $\blam^j_1=0\leq 1$.  Thus, $\blam^j_1 \leq 1$ in both cases.
It  follows that $\bgam^j_i = \one_{i\in S_{j}} \cdot \blam^j_i \leq 1$. 
\end{proof}

We can use \Cref{lem:generalized_concentration} to show a concentration property of the vectors $\bgam^j$. 
	\begin{lemma}
		\label{lem:step_concentration_upperbound}
		Let $\bu\in \mathbb{R}^I_{\geq 0 }$, $j\in [\eps^{-1}]$ and $t>0$. Also, assume that $\tol(\bu)\neq 0$. Then,
		$$
		\Pr\left(\bu \cdot \bgam^j \geq \E[\bu \cdot\bgam^j \,|\,\cF_{j-1}] +t\cdot \tol (\bu)  \right)\leq \exp\left( -\frac{t^2}{2\cdot \eps \cdot m}\right)
		$$
	\end{lemma}

\begin{proof}
	Define $T = \sum_{i \in S_{j-1}} \bu \cdot \blam^j_i$; observe that $T$ is $\cF_{j-1}$-measurable. 
	We can express $\bu\cdot \bgam^j$ using $T$ as follows:
	\begin{equation}
		\label{eq:S1}
		\begin{aligned}
		\bu \cdot \bgam^j ={} &  \sum_{i \in I} 	\bu_i \cdot \bgam^j_i \\
		={} &  \sum_{i \in I}  \one_{i\in S_{j}} \cdot \bu_i \cdot \blam^j_i\\
			={} &  \sum_{i \in S_{j-1}}\bu_i \cdot \blam^j_i-\sum_{i \in S_{j-1} \setminus S_j} \bu_i \cdot \blam^j_i\\
		={} & T-\sum_{i \in S_{j-1} \setminus S_j} \bu_i \cdot \blam^j_i\\
				={} &  T-\sum_{i \in \bigcup_{b=1}^{q} R^j_b} \bu_i \cdot \blam^j_i.
		\end{aligned}
	\end{equation} 
Given a vector $\bbe \in [0,1]^I$ we define the following functions. First, define the function $h_{\bbe}:I \rightarrow \mathbb{R}_{\geq 0}$
by  $h_{\bbe}(i) = \bu_i \cdot \bbe_i$ for for all $i \in I$. Next, define the function $f_{\bbe} :\cC^{q} \rightarrow \mathbb{R}_{\geq 0}$ 
by 
$$f_{\bbe}(X) = \frac{h_{\bbe}\left( \bigcup^q_{b=1} C_b \right)}{\tol (\bu)}$$
for all $X = (C_1,\ldots,C_q) \in \cC^q$.

	Let $\bbe \in [0,1]^I$. We show next that $f_{\bbe}$ is a self-bounding function. Observe that \begin{equation}
		\label{eq:Be}
		\max\left\{ h_{\bbe}(C)~|~C\in\cC\right\} = \max\left\{ \sum_{i\in C} \bu_i \cdot \bbe_i ~|~C\in\cC\right\} \leq \max\left\{ \sum_{i\in C} \bu_i~|~C\in\cC\right\} = \tol(\bu).
	\end{equation} The inequality holds since $\bbe \in [0,1]^I$. Therefore, by \eqref{eq:Be} and by \Cref{lem:boundFunc} it holds that $f_{\bbe}$ is a self-bounding function. 
 Observe that $\blam^{j} \in [0,1]^I$ by \Cref{lem:ZO}; thus, $f_{\blam^j}$ and $h_{\blam^{j}}$ are well defined. Therefore, 
\begin{equation}
	\label{eq:St}
	\begin{aligned}
	\sum_{i \in \bigcup_{b=1}^{q} R^j_b} \bu_i \cdot \blam^j_i ={} & \sum_{i \in \bigcup_{b=1}^{q} R^j_b}  h_{\blam^j}(i) \\
	={} & h_{\blam^j} \left(  \bigcup_{b=1}^{q} R^j_b  \right) \\
	={} & \tol (\bu) \cdot \frac{h_{\blam^j} \left(  \bigcup_{b=1}^{q} R^j_b  \right) }{\tol (\bu)} \\
	={} &  \tol (\bu) \cdot f_{\blam^j}\left( R^j_1,\ldots,R^j_q\right).
	\end{aligned}
\end{equation} For short, let $\cR^j = (R^j_1,\ldots,R^j_q)$. By \eqref{eq:S1} and \eqref{eq:St} we have,
\begin{equation}
	\label{eq:S2}
		\bu \cdot \bgam^j =  T-\tol (\bu) \cdot f_{\blam^j}\left(\cR^j\right).
\end{equation} 
Hence, 
\begin{equation}
	\label{eq:B31}
	\begin{aligned}
	{} & \Pr\left(\bu \cdot \bgam^j \geq \E[\bu \cdot\bgam^j \,|\,\cF_{j-1}] +t\cdot \tol (\bu)  ~\bigg|~\cF_{j-1}\right)\\
	={} & 	\Pr\Bigg(T-\tol (\bu) \cdot f_{\blam^j}\left( \cR^j \right) \geq \E \left[ T-\tol (\bu) \cdot f_{\blam^j}\left(\cR^j\right) \,|\,\cF_{j-1}\right] +t\cdot \tol (\bu)  ~\Bigg|~\cF_{j-1}\Bigg)\\
		={} & 	\Pr\Bigg(\tol (\bu) \cdot f_{\blam^j}\left(\cR^j\right) \leq \E \left[\tol (\bu) \cdot f_{\blam^j}\left(\cR^j\right) \,|\,\cF_{j-1}\right]-t\cdot \tol (\bu)  ~\Bigg|~\cF_{j-1}\Bigg)\\
			={} & 	\Pr\Bigg(f_{\blam^j}\left(\cR^j\right) \leq \E \left[f_{\blam^j}\left(\cR^j\right) \,|\,\cF_{j-1}\right]-t~\Bigg|~\cF_{j-1}\Bigg)\\
			\leq{} & \exp \left( -\frac{t^2}{2 \cdot \E \left[f_{\blam^j}\left(\cR^j\right) \,|\,\cF_{j-1} \right]} \right).
	\end{aligned}
\end{equation} The first equality follows from \eqref{eq:S2}. The second equality holds since $T$ is $\cF_{j-1}$-measurable; thus, it holds that $\E \left[T~|~ \cF_{j-1}\right] = T$. The inequality holds by applying the concentration bound from \Cref{lem:generalized_concentration}. We now bound $f_{\blam^j}\left(\cR^j\right)$. \begin{equation}
\label{eq:Bf}
f_{\blam^j}\left(\cR^j\right) = \frac{h_{\blam^j} \left( \bigcup^q_{b=1} R^j_b \right)}{\tol(\bu)} \leq \frac{\sum^q_{b=1} h_{\blam^j} \left( R^j_b \right)}{\tol(\bu)} \leq \frac{\sum^q_{b=1} \max_{C \in \cC} h_{\blam^j} \left( C \right)}{\tol(\bu)} \leq \frac{\sum^q_{b=1} \tol(\bu)}{\tol(\bu)} = q. %
\end{equation} The last inequality follows by \eqref{eq:Be}. Therefore, 
\begin{equation*}
	\begin{aligned}
					 \Pr\left(\bu \cdot \bgam^j \geq \E[\bu \cdot\bgam^j \,|\,\cF_{j-1}] +t\cdot \tol (\bu)  ~\bigg|~\cF_{j-1}\right) \leq \exp \left( -\frac{t^2}{2 \cdot \E \left[f_{\blam^j}\left(\cR^j\right) \,|\,\cF_{j-1} \right]} \right)
					 \leq \exp \left( -\frac{t^2}{2 \cdot q} \right). 
	\end{aligned}
\end{equation*} Thus, we can get the same bound unconditional on $\cF_{j-1}$: 
\begin{equation*}
	\begin{aligned}
		\Pr\left(\bu \cdot \bgam^j \geq \E[\bu \cdot\bgam^j \,|\,\cF_{j-1}] +t\cdot \tol (\bu)\right) ={} & 	\E \left[\Pr\left(\bu \cdot \bgam^j \geq \E[\bu \cdot\bgam^j \,|\,\cF_{j-1}] +t\cdot \tol (\bu)~\big|~\cF_{j-1}\right) \right]  \\
		\leq{} & \E \left[\exp \left( -\frac{t^2}{2 \cdot q} \right) \right]\\
		={} & \exp \left( -\frac{t^2}{2 \cdot q} \right) 
	\end{aligned}
\end{equation*}
\end{proof}

	By inductive application  of \Cref{lem:step_concentration_upperbound} and  using \Cref{lem:gamma_expectation} we can show the following.
	\begin{lemma}
		\label{lem:concentration_vector}
	Let $\bu \in \mathbb{R}^I_{\geq 0 }$ and $t>0$. Then,
	$$
	\Pr\left(\forall j\in[\eps^{-1}]:~~ \bu \cdot \bgam^{j-1}  \leq \left( 1- (j-1) \cdot \eps\right) \cdot \bu \cdot \one_{S^*}  +\frac{t}{\eps} \cdot \tol(\bu) \right)  \geq 1-\frac{1}{\eps} \cdot \exp\left( -\frac{t^2}{q}\right)
	$$
	\end{lemma}
	
	\begin{proof}
		We use the following auxiliary claim. 
		\begin{claim}
			\label{claim:auxCV}
			If for all $ j\in[\eps^{-1}]$ it holds that $\bu \cdot \bgam^j < \E \left[ \bu \cdot \bgam^j ~|~\cF_{j-1} \right]+t \cdot \tol(\bu)$ Then for all $j\in[\eps^{-1}]$ \begin{equation}
				\label{eq:jj}
				\bu \cdot \bgam^{j-1}  \leq \left( 1- (j-1) \cdot \eps\right) \cdot \bu \cdot \one_{S^*}  + (j-1) \cdot t \cdot \tol(\bu). 
			\end{equation}
		\end{claim}
	\begin{claimproof}
		Assume that for all $ j\in[\eps^{-1}]$ it holds that $\bu \cdot \bgam^j < \E \left[ \bu \cdot \bgam^j ~|~\cF_{j-1} \right]+t \cdot \tol(\bu)$. We prove that \eqref{eq:jj} holds for all $j\in[\eps^{-1}]$ by induction on $j$. For the base case, let $j = 1$. Recall that $\bgam^0 = \one_{S^*}$; thus, $$	\bu \cdot \bgam^{j-1} =   \bu \cdot \one_{S^*} = \left( 1- (1-1) \cdot \eps\right) \cdot \bu \cdot \one_{S^*}  + (1-1) \cdot t \cdot \tol(\bu).$$
		For the step of the induction, assume that \eqref{eq:jj} holds for some $j\in[\eps^{-1}]$. Consider \eqref{eq:jj} for $j+1$:
		\begin{equation}
			\label{eq:Nt}
			\begin{aligned}
					\bu \cdot \bgam^{j} < {} & \E \left[ \bu \cdot \bgam^j ~|~\cF_{j-1} \right]+t \cdot \tol(\bu)\\
					={} &  \E \left[ \sum_{i \in I} \bu_i \cdot \bgam^j_i ~|~\cF_{j-1} \right]+t \cdot \tol(\bu)\\
						={} &  \sum_{i \in I} \E \left[  \bu_i \cdot \bgam^j_i ~|~\cF_{j-1} \right]+t \cdot \tol(\bu)\\
							={} &  \sum_{i \in I} \bu_i \cdot \E \left[   \bgam^j_i ~|~\cF_{j-1} \right]+t \cdot \tol(\bu)\\
					={} &   \sum_{i \in I} \frac{1-j\cdot \eps}{1-(j-1)\cdot \eps}\cdot \bu_i \cdot \bgam^{j-1}_i+t \cdot \tol(\bu)\\
			\end{aligned}
		\end{equation} The  inequality holds by the assumption at the beginning of the claim. The fourth equality follows by \Cref{lem:gamma_expectation}. therefore, by rewriting the last expression in \eqref{eq:Nt}: 
	\begin{equation*}
		\begin{aligned}
				\bu \cdot \bgam^{j} <{} &   \frac{1-j\cdot \eps}{1-(j-1)\cdot \eps}\cdot \bu \cdot \bgam^{j-1}+t \cdot \tol(\bu)\\
		\leq{} &  \frac{1-j\cdot \eps}{1-(j-1)\cdot \eps}\cdot \bigg(  \left( 1- (j-1) \cdot \eps\right) \cdot \bu \cdot \one_{S^*}  + (j-1) \cdot t \cdot \tol(\bu) \bigg)+t \cdot \tol(\bu)\\
		={} &     \left( 1-j\cdot \eps\right) \cdot \bu \cdot \one_{S^*}  + \left(\frac{1-j\cdot \eps}{1-(j-1)\cdot \eps}\cdot (j-1)+1\right) \cdot  t \cdot \tol(\bu)\\
		\leq{} &     \left( 1-j\cdot \eps\right) \cdot \bu \cdot \one_{S^*}  + j \cdot t \cdot \tol(\bu)\\
		\end{aligned}
	\end{equation*} The second inequality follows from the induction hypothesis.
		\end{claimproof}
	
	Therefore, by \Cref{claim:auxCV} we have 
	\begin{equation}
		\label{eq:Nt22}
		\begin{aligned}
		{} &	\Pr\left(\forall j\in[\eps^{-1}]:~~ \bu \cdot \bgam^{j-1}  \leq \left( 1- (j-1) \cdot \eps\right) \cdot \bu \cdot \one_{S^*}  +\frac{t}{\eps} \cdot \tol(\bu) \right) \\ 
		\geq{} &  \Pr\bigg(\forall j\in[\eps^{-1}]:~~ \bu \cdot \bgam^{j-1}  \leq \left( 1- (j-1) \cdot \eps\right) \cdot \bu \cdot \one_{S^*}  +(j-1) \cdot t\cdot \tol(\bu) \bigg) \\
				\geq{} &  \Pr\bigg(\forall j \in \left[ \eps^{-1} \right]:~~\bu \cdot \bgam^{j-1}  <  \E \left[ \bu \cdot \bgam^j ~|~\cF_{j-1} \right]+t \cdot \tol(\bu) \bigg) \\
					\geq{} &  1-\Pr\bigg(\exists j \in \left[ \eps^{-1} \right]:~~\bu \cdot \bgam^{j-1}  \geq \E \left[ \bu \cdot \bgam^j ~|~\cF_{j-1} \right]+t \cdot \tol(\bu) \bigg) \\
						\geq{} &  1- \sum^{\eps^{-1}}_{j = 1} \Pr\bigg(\bu \cdot \bgam^{j-1}  \geq \E \left[ \bu \cdot \bgam^j ~|~\cF_{j-1} \right]+t \cdot \tol(\bu) \bigg)
		\end{aligned}
	\end{equation} The second inequality holds by \Cref{claim:auxCV}. The last inequality follows from the union bound. Thus, by \eqref{eq:Nt22} and \Cref{lem:step_concentration_upperbound} it follows that:
	\begin{equation*}
	\label{eq:Nt2}
	\begin{aligned}
			\Pr\left(\forall j\in[\eps^{-1}]:~~ \bu \cdot \bgam^{j-1}  \leq \left( 1- (j-1) \cdot \eps\right) \cdot \bu \cdot \one_{S^*}  +\frac{t}{\eps} \cdot \tol(\bu) \right)
		\geq{} & 1-\sum^{\eps^{-1}}_{j = 1} \frac{1}{\eps} \cdot \exp\left( -\frac{t^2}{q}\right)\\
		={} & 1-\frac{1}{\eps} \cdot \exp\left( -\frac{t^2}{q}\right). 
	\end{aligned}
\end{equation*}
	\end{proof}
	By \Cref{lem:structure} there is an  $\eps^2$-linear structure $\cL$ of the optimal solution $S^*$ of $\cI$ such that $|\cL| \leq \exp(-\eps^{-8})$. 
	The following is an immediate consequence of \Cref{lem:concentration_vector}. 
	\begin{lemma}
		\label{lem:concentration_structure}
	For every $t>0$ it holds that 
	$$\Pr\left( \forall \bu \in \cL, \, j\in [\eps^{-1}]:~ \bu\cdot \bgam^{j-1} \leq \left(1-(j-1)\cdot \eps\right) \cdot \bu \cdot \one_{S^*} +\frac{t}{\eps} \cdot\ \tol(\bu)\right) \geq 1- \frac{|\cL|}{\eps} \cdot \exp\left( -\frac{t^2}{q}\right).$$
	\end{lemma}
	
	\begin{proof}
		\begin{equation}
			\label{eq:sB2}
			\begin{aligned}
		&{}	\Pr\left( \forall \bu \in \cL, \, j\in [\eps^{-1}]:~ \bu\cdot \bgam^{j-1} \leq \left(1-(j-1)\cdot \eps\right) \cdot \bu \cdot \one_{S^*} +\frac{t}{\eps} \cdot\ \tol(\bu)\right) \\
		={} & 1- \Pr\left( \exists \bu \in \cL, \, j\in [\eps^{-1}]:~ \bu\cdot \bgam^{j-1} > \left(1-(j-1)\cdot \eps\right) \cdot \bu \cdot \one_{S^*} +\frac{t}{\eps} \cdot\ \tol(\bu)\right) \\
			\geq{} & 1- \sum_{\bu \in \cL}\Pr\left(\exists j\in [\eps^{-1}]:~ \bu\cdot \bgam^{j-1} > \left(1-(j-1)\cdot \eps\right) \cdot \bu \cdot \one_{S^*} +\frac{t}{\eps} \cdot\ \tol(\bu)\right) \\
				\geq{} & 1- \sum_{\bu \in \cL} \frac{1}{\eps} \cdot \exp\left( -\frac{t^2}{q}\right) \\
				={} & 1- \frac{|\cL|}{\eps} \cdot \exp\left( -\frac{t^2}{q}\right). 
			\end{aligned}
		\end{equation} The first inequality holds by the union bound. The second inequality follows from \Cref{lem:concentration_vector}. 
	\end{proof}
	
	We use  \Cref{lem:concentration_structure}
	and the properties of the linear structure $\cL$
	to provide a lower bound for $\v(\bx^j)$ in terms of $\by^{j-1}$.

	\xingamma*
\begin{proof}
	Define $t=\eps^{10}\cdot \exp(-\eps^{-10})$. Also, define the event
	$$
	\Gamma \,=\, \left\{\forall \bu \in \cL, \, j\in [\eps^{-1}]:~ \bu\cdot \bgam^{j-1} \leq \left(1-(j-1)\cdot \eps\right) \cdot \bu \cdot \one_{S^*} +\frac{t}{\eps} \cdot\ \tol(\bu) \right\}.
	$$
	
	By \Cref{lem:concentration_structure} it holds that $\Pr(\Gamma)\geq 1-\frac{\abs{\cL}}{\eps} \cdot \exp\left( -\frac{t^2}{q}\right)$. It also holds that,	
	$$
	\begin{aligned}
		\frac{\abs{\cL}}{\eps} \cdot \exp\left( -\frac{t^2}{q}\right) \, &\leq \, \frac{\exp(\eps^{-8}) }{\eps  } \cdot \exp\left(- \frac{\eps^{20}\cdot \exp(-2\cdot \eps^{-10})\cdot m^2}{\eps \cdot m }\right)\\
		&=\,   \frac{\exp(\eps^{-8}) }{\eps  } \cdot \exp\left(- \eps^{19}\cdot \exp(-2\cdot \eps^{-10})\cdot m\right)\\
		&\leq \, \frac{\exp(\eps^{-8}) }{\eps  } \cdot \exp\left(- \eps^{19}\cdot \exp(-2\cdot \eps^{-10})\cdot \exp\left(\exp(\eps^{-30})\right)\right)\\
		&\leq \,\exp(-\eps^{-20}).
\end{aligned}
		$$
	The first inequality holds as $\abs{\cL}\leq \exp(\eps^{-8})$, the second inequality uses $m\geq \exp(\exp(\eps^{-30}))$ as $\cI$ is an $\eps$-simple instance (\Cref{def:simple}). 
	Thus $\Pr(\Gamma)\geq 1-\exp(-\eps^{-20})$.
	
	We  assume that $\Gamma$ occurs from this point on in the proof.
 Since $\cL$ is a linear structure, for every $j\in[\eps^{-1}]$ there is a fractional solution $\bz^j$ such that  $\cover(\bz^j) =\bgam^{j-1}$ and 
	\begin{equation}
		\label{eq:zsize}
	\begin{aligned}
	\|\bz^j\| \,&\leq \, (1-(j-1)\eps)\cdot m+20\eps^2\cdot m +\left(\frac{t}{\eps}+1\right)\cdot \exp(\eps^{-10})\\
	&=\, (1-(j-1)\eps)\cdot m+20\eps^2\cdot m +\left(\frac{\eps^{10}\cdot \exp(-\eps^{-10})}{\eps}+1\right)\cdot\exp(\eps^{-10})\\
	&\leq \, (1-(j-1)\eps)\cdot m+30\cdot \eps^2\cdot m, 
	\end{aligned}
	\end{equation}
	where the last inequality holds as $m\geq \exp(\exp(\eps^{-30}))$.
	For every $j\in [\eps^{-1}]$ define  $$\bbe^j = \bz^j \cdot \frac{1-(j-1)\eps}{1-(j-1)\eps+30\cdot \eps^2}.$$ 
	By \eqref{eq:zsize}, for every $j\in [\eps^{-1}$] it holds that $\|\bbe^j\| \leq (1-(j-1)\eps)\cdot m=m_j$ (recall $m_j$ is defined in \Cref{randomized_rounding:solve} of \Cref{alg:randomized_rounding}). Furthermore, 
	\begin{equation}
		\label{eq:beta_cover}\cover(\bbe^j)\,=\, \cover(\bz^j) \cdot \frac{1-(j-1)\eps}{1-(j-1)\eps+30\cdot \eps^2}\,=\,\bgam^{j-1} \cdot \frac{1-(j-1)\eps}{1-(j-1)\eps+30\cdot \eps^2}.
		\end{equation} 
		For every $j\in[\eps^{-1}]$, by \eqref{eq:beta_cover} it holds that $$\supp(\cover(\bbe^j))\,\subseteq \,\supp(\cover(\bz^j))\,\subseteq\, \supp(\bgam^{j-1})\,\subseteq\, S_{j-1},$$
		 and therefore $\supp(\bbe^j) \subseteq 2^{S_{j-1}}$.  Thus, $\bbe^j$  is a solution for $\LP(S_{j-1},m_j)$ for every $j\in [\eps^{-1}$]. Since $\bx^j$ is a $(1-\eps)$-approximate  solution  for $\LP(S_{j-1},m_j)$ we conclude that 
		\begin{equation*}
		 \begin{aligned}
		\v(\bx^j) \,&\geq (1-\eps)\cdot \v(\bbe^j) \\&
		= \,(1-\eps)\cdot \v\left( \bgam^{j-1} \cdot \frac{1-(j-1)\eps}{1-(j-1)\eps+30\cdot \eps^2}\right)  \\
		& =\,  (1-\eps)\cdot \frac{1-(j-1)\eps}{1-(j-1)\eps+30\cdot \eps^2}\cdot  \v(\bgam^{j-1})\\
		&=\,  (1-\eps)\cdot\left( 1- \frac{30\cdot \eps^2}{1-(j-1)\eps+30\cdot \eps^2}\right)\cdot  \v(\bgam^{j-1})\\
		&\geq \, (1-\eps)\cdot\left( 1- \frac{30\cdot \eps^2}{1-(j-1)\eps}\right)\cdot  \v(\bgam^{j-1}),
		\end{aligned}
		\end{equation*}
		for every $j\in [\eps^{-1}]$, 
	where the first equality follows from \eqref{eq:beta_cover}. Since we assume $\Gamma$ occurs, we have
	$$
	\Pr\left( \forall j\in [\eps^{-1}]:~~
	\v(\bx^j) \,\geq \, (1-\eps)\cdot\left( 1- \frac{30\cdot \eps^2}{1-(j-1)\eps}\right)\cdot  \v(\bgam^{j-1})\right) \,\geq \, \Pr(\Gamma)\,\geq 1-\exp(-\eps^{-20}).
	$$
\end{proof}

	\subsubsection{Proof of~\Cref{lem:gamma_value_lowerbound}}
	\label{sec:gamma_val}
	
	The proof of \Cref{lem:gamma_value_lowerbound}  is a simple consequence  \Cref{lem:gamma_expectation}.
	\gammaval*
	\begin{proof}
		We prove the following claim by induction over $j$. 
		\begin{claim}
			\label{claim:gamma_exp_uncond}
			For every $j\in [\eps^{-1}]$ and $i\in I$ is holds that $\E[\bgam^{j-1}_i ] \,=\, (1-(j-1)\cdot \eps)\cdot \bgam^0_i$.
		\end{claim}
		\begin{claimproof}~

			\noindent{\bf Base case:} for $j=0$ and every $i\in I$, as the vector $\bgam^0=\one_{S^*}$ is not random, we have $\E[\bgam^0_i] =\bgam^0_i$. 
			
			\noindent{\bf Induction Step:} Assume the induction hypothesis holds for $j-1$ and let $i\in I$. Then
			$$
			\begin{aligned}
				\E[\bgam^{j-1}_i]\, &=\, \E\left[\, \E\left[\bgam^{j-1}_i\,|\,\cF_{j-2}\right]\,\right] \, \\
				&= \,  \E\left[\,  \frac{1-(j-1)\cdot \eps}{1-(j-2)\cdot \eps} \cdot \bgam^{j-2}_i\,\right] \, \\
				&= \, \frac{1-(j-1)\cdot \eps}{1-(j-2)\cdot \eps} 
			\cdot \left(1-(j-2)\cdot \eps \right) \cdot \bgam^0_i \,\\
			&= \, \left(1-(j-1)\cdot \eps \right) \cdot \bgam^0_i.
			\end{aligned}
			$$
			The second equality follows from \Cref{lem:gamma_expectation}, and the third equality follows from the induction hypothesis. 
		\end{claimproof}
		By \Cref{claim:gamma_exp_uncond} we have
		$$
		\begin{aligned}
		\E[\v(\bgam^{j-1})] \, &=\, \E\left[ \sum_{i\in I} \v(i)\cdot \bgam^{j-1}_i\right] \, \\
		&= \, \sum_{i\in I} (1-(j-1)\cdot \eps)\cdot \bgam^0_i \cdot \v(i) \,\\
		&=\, (1-(j-1)\cdot \eps)\cdot \sum_{i\in S^*} \v(i) \,\\
		&=\, (1-(j-1)\cdot \eps)\cdot \v(S^* )\, \\
		&=\,(1-(j-1)\cdot \eps)\cdot \OPT.
		\end{aligned}
		$$
	\end{proof}

	\subsubsection{Approximation Ratio}
	\label{sec:combine}
	Recall that $V=\v\left( \bigcup_{j=1}^{\eps^{-1}} \bigcup_{ b=1}^{q} R^j_b\right)$ is the value of the solution returned by \Cref{alg:randomized_rounding}.
	The results in \Cref{lem:Upper_Bound_V,lem:value_in_terms_of_gamma,lem:gamma_value_lowerbound} are put together to derive the next lower bound on $V$.

	\approximation*
	\begin{proof}
	Define the event
	$$
	\begin{aligned}
	B\,&=\,\left\{ \forall j\in [\eps^{-1}]:~
	\v(\bx^j) \,\geq \, (1-\eps)\cdot\left( 1- \frac{30\cdot \eps^2}{1-(j-1)\eps}\right)\cdot  \v(\bgam^{j-1})\right\}.
	\end{aligned}
$$
	By \Cref{lem:value_in_terms_of_gamma,} it holds that  $\Pr(B)\geq 1-\exp(-\eps^{-20})$. Observe that $\one_B$ is a random variable whose value is $1$ when $B$ occurs, and $0$ otherwise.

	By \Cref{lem:Upper_Bound_V} it holds that 
	\begin{equation}
		\label{eq:approx_V1}
	\begin{aligned}
		\E\left[V\right]\,&\geq \, \E\left[ (\eps-\eps^{\frac{3}{2}})\cdot \sum_{j=1}^{\eps^{-1} - \eps^{-\frac{1}{2}}}  \frac{ \v(\bx^j)}{1-(j-1)\eps}\right]\\
		&\geq \, \E\left[\one_B\cdot  (\eps-\eps^{\frac{3}{2}})\cdot \sum_{j=1}^{\eps^{-1} - \eps^{-\frac{1}{2}}}  \frac{ \v(\bx^j)}{1-(j-1)\eps}\right]\\
		&\geq \, \E\left[\one_B\cdot  (\eps-\eps^{\frac{3}{2}})\cdot \sum_{j=1}^{\eps^{-1} - \eps^{-\frac{1}{2}}}  \frac{ (1-\eps)\left( 1-\frac{30\cdot \eps^2}{1-(j-1)\cdot \eps} \right)\cdot \v(\bgam^{j-1})}{1-(j-1)\eps}\right],\\
	\end{aligned}
	\end{equation}
	where the last inequality follows from the definition of $B$.  
	For every $j\in [\eps^{-1}-\eps^{-\frac{1}{2}}]$ it holds that 
	\begin{equation*}
		\label{eq:j_simplification}
	\frac{30\cdot \eps^{2}}{1-(j-1)\cdot \eps}\,\leq \, 	\frac{30\cdot \eps^{2}}{1-(\eps^{-1}-\eps^{-\frac{1}{2} }-1)\cdot \eps}\, = \,\frac{30\cdot \eps^{2}}{\eps^{\frac{1}{2} }- \eps} \, \leq 60\cdot \eps^{\frac{3}{2}}.
	\end{equation*}
	Plugging the above into \eqref{eq:approx_V1}, we get
	\begin{equation}
		\label{eq:approx_V2}
	\begin{aligned}
		\E\left[V\right]\,
		&\geq \, \E\left[\one_B\cdot  (\eps-\eps^{\frac{3}{2}})\cdot \sum_{j=1}^{\eps^{-1} - \eps^{-\frac{1}{2}}}  \frac{ (1-\eps)\left( 1-\frac{30\cdot \eps^2}{1-(j-1)\cdot \eps} \right)\cdot \v(\bgam^{j-1})}{1-(j-1)\eps}\right],\\
		&\geq \, \E\left[\one_B\cdot  (\eps-\eps^{\frac{3}{2}})\cdot \sum_{j=1}^{\eps^{-1} - \eps^{-\frac{1}{2}}}  \frac{ (1-\eps)\left( 1-60\cdot \eps^{\frac{3}{2}} \right)\cdot \v(\bgam^{j-1})}{1-(j-1)\eps}\right],\\
		&\geq \, \E\left[\one_B\cdot \eps  (1-100\cdot \eps^{\frac{1}{2}})\cdot \sum_{j=1}^{\eps^{-1} - \eps^{-\frac{1}{2}}}  \frac{  \v(\bgam^{j-1})}{1-(j-1)\eps}\right],\\
		&=\, \E\left[ \eps  (1-100\cdot \eps^{\frac{1}{2}})\cdot \sum_{j=1}^{\eps^{-1} - \eps^{-\frac{1}{2}}}  \frac{  \v(\bgam^{j-1})}{1-(j-1)\eps}\right] \\
		&~~~~~~~~-  \E\left[(1-\one_B)\cdot \eps  (1-100\cdot \eps^{\frac{1}{2}})\cdot \sum_{j=1}^{\eps^{-1} - \eps^{-\frac{1}{2}}}  \frac{  \v(\bgam^{j-1})}{1-(j-1)\eps}\right]
	\end{aligned}
	\end{equation}
	
	By \Cref{lem:gamma_value_lowerbound}, it holds that
	 \begin{equation}
	 	\label{eq:approx_V3}
	 	\begin{aligned}
	 &\,\E\left[ \eps  (1-100\cdot \eps^{\frac{1}{2}})\cdot \sum_{j=1}^{\eps^{-1} - \eps^{-\frac{1}{2}}}  \frac{  \v(\bgam^{j-1})}{1-(j-1)\eps}\right] \,\\
	 =&\,	 \eps  (1-100\cdot \eps^{\frac{1}{2}})\cdot \sum_{j=1}^{\eps^{-1} - \eps^{-\frac{1}{2}}}  \frac{  (1-(j-1)\cdot \eps)\cdot \OPT(\II))}{1-(j-1)\eps}\\
	 =& \,	 \eps  (1-100\cdot \eps^{\frac{1}{2}}) \left(\eps^{-1} - \eps^{-\frac{1}{2}}\right)\cdot \OPT(\II)\\
	 \geq & \left(1-200\cdot \eps^{\frac{1}{2}}\right) \cdot \OPT(\II).
	 \end{aligned}
	 \end{equation}
	
	By the definition of $\bgam^j$ is holds that $\supp(\bgam^j)\subseteq S^*$ and $\bgam^j\in [0,1]$; therefore, $\v(\bgam^j )\leq \v(S^*) = \OPT(\II)$. Hence, 
	$$
	\eps  (1-100\cdot \eps^{\frac{1}{2}})\cdot  \sum_{j=1}^{\eps^{-1} - \eps^{-\frac{1}{2}}}  \frac{  \v(\bgam^{j-1})}{1-(j-1)\eps}\, \leq \,
	\eps\cdot \sum_{j=1}^{\eps^{-1} - \eps^{-\frac{1}{2}}}  \frac{  \OPT(\II)}{\eps^{\frac{1}{2}}}
	\leq \, \eps^{-\frac{1}{2}}\cdot \OPT(\II)	,
	$$
	where the first inequlity holds as $(1-100\cdot \eps^{\frac{1}{2}})\leq 1$,
	 $(1-(j-1)\eps)\geq \eps^{\frac{1}{2}}$ for $1\leq j \leq \eps^{-1}-\eps^{-\frac{1}{2}}$ and $\v(\bgam^j )\leq \v(S^*) = \OPT(\II)$. Thus,
		\begin{equation}
		\label{eq:approx_V4}
		\begin{aligned}
		\E\left[(1-\one_B)\cdot \eps  (1-100\cdot \eps^{\frac{1}{2}})\cdot \sum_{j=1}^{\eps^{-1} - \eps^{-\frac{1}{2}}}  \frac{  \v(\bgam^{j-1})}{1-(j-1)\eps}\right] \,&\leq \, \left(1-\Pr(B)\right) \cdot \eps^{-\frac{1}{2}}\cdot \OPT(\II) \,\\ &\leq \, \exp(-\eps^{-20})\cdot \eps^{-\frac{1}{2}}\cdot \OPT .
	\end{aligned}
		\end{equation}

By \eqref{eq:approx_V2}, \eqref{eq:approx_V3} and \eqref{eq:approx_V4}, we get
	\begin{equation}
	\begin{aligned}
		\E\left[V\right]\,
		&\geq \, \E\left[ \eps  (1-100\cdot \eps^{\frac{1}{2}})\cdot \sum_{j=1}^{\eps^{-1} - \eps^{-\frac{1}{2}}}  \frac{  \v(\bgam^{j-1})}{1-(j-1)\eps}\right] \\
		&~~~~~~~~-  \E\left[(1-\one_B)\cdot \eps  (1-100\cdot \eps^{\frac{1}{2}})\cdot \sum_{j=1}^{\eps^{-1} - \eps^{-\frac{1}{2}}}  \frac{  \v(\bgam^{j-1})}{1-(j-1)\eps}\right]\\
		&\geq (1-200\cdot \eps^{\frac{1}{2}})\cdot \OPT(\II) - \exp(-\eps^{-20})\cdot \eps^{-\frac{1}{2}} \cdot \OPT(\II)\\
			&\geq (1-300 \cdot \eps^{\frac{1}{2}})\cdot \OPT(\II).
	\end{aligned}
\end{equation}

	\end{proof}

%% file: structure.tex
This section is devoted to the proof of \Cref{lem:structure}. 
Let $\cI = (I,\w,\v,m,k)$ be a CMK instance, $S\subseteq I$ be a subset of items which can be packed into $\ell$ bins and $0<\delta<0.1$ such that $\delta^{-1}\in \mathbb{N}$. We assume $\cI$, $\ell$ and $\delta$ are fixed throughout  the whole of  \Cref{sec:structure}. Our objective is to construct a linear structure (\Cref{def:structure}).

The construction is similar to a construction in \cite{KMS23} which deals with the more generic setting of two-dimensional (vector) bins, which is, in turn, inspired by a construction from \cite{BEK16}.
The construction in \cite{KMS23} required the items in $S$ to be packed into $\ell$ configurations with an additional property (``slack'').
Our construction leverages  the simplified structure of a single bin with a cardinality constraint to alleviate the need for this additional property.

We classify items into {\em large} and {\em small} using the concept of adjusted weights which takes into account both the weight of an item and the cardinality constraint.
Define the {\em adjusted weight} of an item $i\in I$ by $\tw(i) = w(i) + \frac{1}{k}$.  The following observations are immediate consequences of the definition of a configuration. 
\begin{obs}
	\label{obs:config_adjusted_weight}
	For every $C\in \cC$ it holds that $\tw(C)\leq 2$.
\end{obs}
\begin{obs}
	\label{obs:set_adjusted_weight}
	For every $T\subseteq I$, if $\tw(T)\leq 1$ then $T\in \cC$.
\end{obs}

The construction of the linear structure relies on a classification of the items into {\em large} and {\em small}  items based on the adjusted weight. We say an item $i\in I$ is {\em large} if $\tw(i) \geq \delta$, otherwise the item is {\em small}. Let $L=\{i\in S~|~\tw(i)\geq \delta\}$ be the set of large items in $S$. Note that if $k\leq \frac{1}{\delta}$ then all the items are large.  The next observation follows from \Cref{obs:config_adjusted_weight}.
\begin{obs}
	\label{lem:large_items_in_conf}
For every $C\in \cC$ it holds that $|C\cap L|\leq \frac{2}{\delta}$.
\end{obs}

The construction of the linear structure is done in three phases. The first phase deals with the large items and follows from a  standard application of the linear grouping technique of \cite{FL81}. The second phase generates a classification of the small items based on the linear grouping of the first phase. Together, the two phases are used to derive a variant of a linear structure which is insufficient for our application. The third phase refines the initial classification of the small items and leads to the final linear structure.

	We use the following two simple lemmas that construct  fractional solutions within the proofs. The first lemma is a fractional variant of First-Fit, and the second lemma only uses configurations of a single item. 
	 The proofs of the lemmas are given in \Cref{sec:defered_const}
\begin{restatable}{lemma}{fractionalfirstfit}
	\label{lem:fractional_first_fit}
	Let $\by \in \left([0,1]\cap \mathbb{Q}\right)^I$ such that  $\supp(\by)\subseteq S\setminus L$. Then there is a fractional solution $\bx$ such that $\cover(\bx) = \by$ and $\|\bx\| \leq 2\cdot \sum_{i\in I} \by_i\cdot \tw(i)  + 1$. 
\end{restatable}
\begin{restatable}{lemma}{itemperbin}
	\label{lem:item_per_bin}
	For every $\by\in [0,1]^I$ there is a fractional solution $\bx$ such that $\cover(\bx)=\by$ and~$\|\bx\|=\|\by\|$.
\end{restatable}

\subsection{Large Items}
\label{sec:grouping}
Recall $L$ is the set of large items in $S$ and let $L=\{i_1,\ldots, i_r\}$ where $r=|L|$ and assume the items in $L$ are sorted in non-decreasing order according to their adjusted weight. That is, $\tw(i_1)\geq \tw(i_2)\geq \ldots \geq \tw(i_r)$.  This also implies that $\w(i_1)\geq \w(i_1)\geq \ldots \geq \w(i_r)$. We apply standard linear grouping~\cite{FL81} on $L$.

We partition $L$ into $\delta^{-2}$ groups $G_1,\ldots, G_{\delta^{-2}}$. The first group consists of the first $s=\ceil{\delta^2\cdot  r}$ items in $L$, the next group consists of the subsequent $s$ items, and so forth. Formally, 
$$
\forall j\in [\delta^{-2}]:~~~~G_j = \{ i_p~|~ s(j-1)+1\leq p \leq s\cdot (j-1) + s,~p\leq r \}.
$$
Assuming $r$ is sufficiently large, it holds that $|G_j|=s$ for every $1\leq j\leq \delta^{-2}-1$, and $|G_{\delta^{-2}}|\leq s$.
The following observation holds independently of the value of $r$. 
\begin{obs} 
	\label{obs:groups_size}
$s\geq |G_1|\geq |G_2|\geq \ldots \geq |G_{\delta^{-2}}|$.
\end{obs}

Let $D_1,\ldots,D_{\ell}\in \cC$ be $\ell$ disjoint configurations such that $\bigcup_{b\in [\ell]} D_{b}= S$. Recall the existence of such configurations is an assumption of the \Cref{lem:structure}. It therefore holds that 
$$
r\,=\,\abs{L}\,=\,\sum_{b\in [\ell]  } \abs{D_b\cap L} \,\leq \, \sum_{b\in [\ell]} 2\cdot\delta^{-1} \, =\, 2\cdot \delta^{-1}\cdot \ell
$$
therefore 
\begin{equation}
	\label{eq:s_bound}
s\,=\,\ceil {\delta^2\cdot r} \, \leq \, \ceil{ \delta^2 \cdot 2 \cdot \delta^{-1} \cdot \ell } \, =\, \ceil{2\cdot \delta \cdot \ell} \,\leq 2\cdot \delta \cdot \ell +1.
\end{equation}

\subsection{Types and Small Items}

The {\em type} of a configuration $C\in \cC$ is the vector $\bp \in \mathbb{N}^{\delta^{-2}}$ defined by $\bp_j = |G_j \cap C|$ for all $j\in [\delta^{-2}]$.  That is, $\bp_j$ is the number of items  in $C$ from the $j$-th group $G_j$.  We also define $\|\bp \| = \sum_{j=1}^{\delta^{-2}} \bp_j$ for every $\bp \in \mathbb{R}_{\geq 0}^{\delta^{-2}}$. 
We use $\type(C)$ to denote the type of $C$. We also use $\type_j(C)$ to refer to the $j$-th entry of the vector $\type(C)$.
Since $G_1,\ldots, G_{\delta^{-2}}$ is a partition of $L$, the definition of types implies the following observation.
\begin{obs}
	\label{obs:type_card}
	For every $C\in \cC$ it holds that $\abs{C\cap L }= \| \bp \| $. 
\end{obs}

For every $j\in [\delta^{-2}]$ we define the {\em rounded weight} of $G_j$ by $\w_j = \min_{i\in G_j} \w(i)$ if $G_j\neq \emptyset$ and $\w_j = 0 $ if $G_j = \emptyset$.
The following observation holds since the weight of every item in $G_j$ is at least the weight of every item in $G_{j+1}$ .
\begin{obs}
\label{obs:shifting_argument}
For every $j\in \{2,\ldots, \delta^{-2}\}$ and $i\in G_j$ it holds that $\w(i)\leq \w_{j-1}$.
\end{obs}
 We use the rounded weights to associate a weight with types. 
Let $\cP = \{ \type(C)~|~C\in \cC\}$ be the set of all possible types of  configurations.  The {\em weight} of a type $\bp \in \cP$ is $\w(\bp) = \sum_{j=1}^{\delta^{-2}} \bp_j \cdot \w_j$.  The weight of a type is a lower bound for the weight of large items in a configuration of the same type. 
\begin{lemma}
	\label{lem:type_weight}
Let $C\in \cC$ and $\bp  = \type(C)$. Then $\w(C\cap L ) \geq \w(\bp)$. 
\end{lemma}
\begin{proof}
By the definition of types and rounded weights we have,
$$
\w(C\cap L)  \,=\, \sum_{j=1}^{\delta^{-2}} \w(C\cap G_j ) \,=\,  \sum_{j=1}^{\delta^{-2}} \sum_{~i\in C\cap G_j~} \w(i ) \,\geq\,   \sum_{j=1}^{\delta^{-2}} \sum_{~i\in C\cap G_j~} \w_j \,=\, \sum_{j=1}^{\delta^{-2}} \bp_j\cdot  \w_j 
\,=\,\w(\bp).
$$
\end{proof}
We can also use the types to upper bound the weight and cardinality of small items in a configuration.
\begin{lemma}
	\label{lem:small_items}
	Let $C\in \cC$ and $\bp = \type(C)$. Then $\w(C\setminus L ) \leq 1- \w(\bp) $ and $|C\setminus L| \leq k - \|\bp\|$. 
\end{lemma}
\begin{proof}
Since $C$ is a configuration it holds that $ \w(C\cap L) + \w(C\setminus L) = \w(C) \leq 1$. Therefore $$\w(C\setminus L )\,\leq\,1 - w(C\cap L)\, \leq \,1- \w(\bp ),$$ 
where the last inequality follows from \Cref{lem:type_weight}.

Similarly, since $C$ is a configuration is holds that $|C\cap L| + |C\setminus L| =|C|\leq k$. Thus, 
$$
|C\setminus L| \,\leq\, k -|C\cap L|  \,=\, k-\|\bp\|,
$$
where the  equality follows from \Cref{obs:type_card}.
\end{proof}

We  use a  simple combinatorial argument to attain a naive  bound on the size of $\cP$.
\begin{lemma}
	\label{lem:num_types}
	$|\cP|\leq \exp(\delta^{-3})$
\end{lemma}
\begin{proof}
	By \Cref{lem:large_items_in_conf}, for every $C\in \cC$ it holds that $|C\cap L|\leq \frac{2}{\delta}$.
	Let $\bp \in \cP$, then  there is $C\in \cC$ such that $\type(C)=\bp$.
	Thus, for every $j\in [\delta^{-2}]$ we have $\bp_j  = |C\cap G_j|\leq |C\cap L| \leq 2\cdot \delta^{-1}$. 
	Therefore $\cP\subseteq \{a\in \mathbb{N}_{\geq 0}~|~a\leq 2\cdot \delta^{-1}\}^{\delta^{-2}}$, and thus
	$$
	|\cP|\subseteq \left|\{a\in \mathbb{N}_{\geq 0}~|~a\leq 2\cdot \delta^{-1}\}^{\delta^{-2}} \right| \leq  \left(1+2\cdot \delta^{-1} \right)^{\delta^{-2}}  \leq \left(\delta^{-2} \right)^{\delta^{-2}}  = \exp\left(\delta^{-2} \cdot \ln \left(\delta^{-2} \right)\right)\leq \exp(\delta^{-3}). 
	$$
\end{proof}

We use the types to generate a classification of the small items in $S$ into groups. Recall $D_1,\ldots ,D_{\ell} \in \cC$ are $\ell$ disjoint configurations such that $\bigcup_{b=1}^{\ell} D_b = S$	. We partition $S\setminus L$,  the small items in $S$, into classes based on the type of the configuration among $D_1,\ldots, D_m$  to which an item belongs. Formally, for every $\bp \in \cP$ we define the class  $\cK_{\bp}\subseteq S\setminus L$ by
$$\cK_{\bp} =\left\{i\in S\setminus L~|~\exists b\in [\ell]:~\type(D_b) = \bp \textnormal{ and } i\in D_b\right\}.$$

Define $\eta:\cP\rightarrow \mathbb{N}$ by $\eta(\bp) = \left| \left\{ b\in [\ell ]~|~\type(D_b) = \bp\right\}\right|$. That is, $\eta(\bp)$ is the number of configurations among $D_1,\ldots, D_{\ell}$ of type $\bp$.
\begin{lemma}
	\label{lem:class_weight_and_card}
	For every $\bp \in \cP$ it holds that $\w(\cK_{\bp}) \leq \eta(\bp)\cdot (1-\w(\bp))$ and $\abs{\cK_{\bp}}\leq \eta(\bp)\cdot (k-\|\bp\|)$. 
\end{lemma}
\begin{proof}
Let $Q = \{b\in [\ell]~|~\type(D_b)=\bp\}$ be the set of indices  of configuration of type $\bp$.
 Then $\cK_{\bp} = \bigcup_{b\in Q} D_b\setminus L$ and $\abs{Q}=\eta(\bp)$. Therefore, by \Cref{lem:small_items} we have,
$$
\w(\cK_{\bp}) \,=\, \sum_{b\in Q} \w(D_b\setminus L) \,\leq \, \sum_{b\in Q} (1-\w(\bp )) \,=\,  \eta(\bp)\cdot (1-\w(\bp))
$$ 
and 
$$
\abs{\cK_{\bp}} \,=\, \sum_{b\in Q} \abs{D_b\setminus L} \,\leq \, \sum_{b\in Q} (k-\|\bp\|) \,=\,  \eta(\bp)\cdot (k-\|\bp\|)
$$ 
\end{proof}

The mapping $\eta$ also satisfies the following property.
\begin{lemma}
	\label{lem:eta_to_groups}
For every $j\in [\delta^{-2}]$ it holds that $\sum_{\bp \in \cP} \eta(\bp)\cdot \bp_j = \abs{G_j}$.
\end{lemma}
\begin{proof}
Since  every item in $S$ belongs to exactly one of the configurations $D_1,\ldots, D_{\ell}$  and $G_j\subseteq S$ we have,
$$
\abs{G_j} \,=\, \sum_{b=1}^{\ell} \abs{D_b\cap G_j} \,= \, \sum_{b=1}^{\ell} \type_j(D_b) \,=\,  \sum_{\bp \in \cP} \sum_{~b\in [\ell] \textnormal{ s.t. } \type(D_b) = \bp~} \bp_j \,=\, \sum_{\bp\in\cP} \eta(\bp)\cdot \bp_j.$$
\end{proof}

\subsection{Scaled selection  and a Weak Variant of a Linear Structure}

We use the classification of items in to groups and classes to show that if  $\by\in [0,1]^{I}$ selects an $\alpha$-fraction of every group $G_j$ by cardinality  and $\alpha$-fraction of each class $\cK_{\bp}$  by weight and cardinality, then there is a fractional solution $\bx$ such that $\cover(\bx)=\by$ and $\|\bx\|\approx \alpha\cdot \ell$. 
The following definition formalizes the notion of fractional selection of items from groups and classes.

\begin{definition}[$\alpha$-scaled]
	\label{def:compliance}
	Let $0\leq \alpha\leq 1$ and $\by \in [0,1]^{I}$. We say $\by$ is {\em $\alpha$-scaled}  if the following conditions hold.
	\begin{itemize}
		\item $\supp(\by) \subseteq S$. 
		\item For every $j\in [\delta^{-2}]$ it holds that $\sum_{i\in G_j} \by_i \leq \alpha \cdot \abs{G_j}$. 
		\item For every type $\bp \in \cP$ it holds that $\sum_{i\in \cK_{\bp} } \by_i \leq \alpha \cdot \abs{ \cK_{\bp}}$ and $\sum_{i\in \cK_{\bp} } \by_i \cdot \w(i)\leq \alpha \cdot \w\left(\cK_{\bp}\right)$.
	\end{itemize}
\end{definition}

The next lemma can be seen as a weak variant of a linear structure.
\begin{lemma}
	\label{lem:weak_structure}
	Let $\by\in \left([0,1]\cap \mathbb{Q}\right)^{I}$ such that $\by$ is $\alpha$-scaled. Then there a fractional solution $\bx$ such that $\cover(\bx)=\by$ and $\| \bx\| \leq \alpha \cdot (1+10\delta)\cdot \ell +\exp(\delta^{-4})$. 
\end{lemma}
\begin{proof}
	
	Let $N\in \mathbb{N}_{>0}$ such that $N\cdot \by_i	\in \mathbb{N}_{\geq 0}$ for every $i\in I$ and $\alpha \cdot N \in \mathbb{N}$. The construction of $\bx$ goes through two major steps. The first step generates $M=\alpha \cdot N \cdot \ell + N\abs{\cP}$ configurations $A_1,\ldots ,A_M$ such that every large item $i\in L$ appears in  $N\cdot \by_i$ configuration , with the exception of a few large items. 
  The second step  uses a linear program to add small items to $A_1,\ldots, A_M$, generating new configurations  $A^*_1,\ldots, A^*_M$  such that every small item $i\in S\setminus L$ appears in  $N\cdot \by_i$ configuration , with the exception of a few small items. The configurations $A^*_1,\ldots, A^*_M$ are used to generate the final fractional solution~$\bx$.

We begin with a  procedure which can be viewed as a variant of the {\em shifting}  step of linear grouping~\cite{FL81}. The procedure  conceptually shifts  between items the of the  groups $G_1,\ldots, G_{\eps^{-2}} $. 
Let $\xi:[M] \rightarrow \cP$ be an arbitrary function such that $\abs{\xi^{-1}(\bp) } = \alpha \cdot N \cdot \eta(\bp)+N$ for every $\bp \in \cP$. Such a function exists since $\sum_{\bp \in \cP} \eta(\bp)  ~=~\ell$.   That is, the  function maps $\alpha \cdot N \cdot \eta(\bp)+N$ indices to the type $\bp$.  The procedure maintains $M$ sets (which are in fact configurations)  $A_1,\ldots, A_{M} $ initialized with the empty set and a set $R$ of items whose allocation failed. The configuration $A_b$ serves as a placeholder for up to $\xi_{j-1}(b)$ items from the group $G_j$. 
The procedure iterates over the items  $i\in L\setminus G_1$  and attempts to allocate $i$  to $N\cdot \by_i$ of the configurations $A_1,\ldots, A_{M}$ while maintaining the constraint $|A_b\cap G_j| \leq \xi_{j-1}(b)$ for every $b\in [M]$ and $j\in \{2,\ldots, \delta^{-2}\}$.  If the allocation failed, the item is added to the set $R$.

More formally, the procedure is: 
\begin{itemize}
	\item 
 Initialize $A_1,\ldots, A_{M}\leftarrow \emptyset$ and $R\leftarrow \emptyset$. 
 \item For every $j=2,\ldots, \delta^{-2}$ do:
 \begin{itemize}
 \item 
 For every item $i\in G_j$ do:
 \begin{itemize}
 	\item Find a set of indices  $B\subseteq [M]$ 
 	such that $\abs{B} =   N \cdot \by_i$ and $\abs{A_b \cap G_j} <\xi_{j-1}(b) $ for every $b\in B$.
 	\item If such a set $B$ was found, update $A_b \leftarrow A_b\cup \{i\}$ for every $b\in B$. 
 	\item If such a set $B$ does not exist, update $R\leftarrow R\cup \{i\}$.
 	\end{itemize}
\end{itemize}
\end{itemize}
Let $A_1,\ldots, A_{M} $ and $R$ be the values of the variables at the end of the procedure.  
The next observation is an immediate consequence of the construction of $A_1,\ldots, A_M$. 
\begin{observation}
	\label{obs:large_items}
Every item $i\in L\setminus\left( G_1 \cup R\right)$ is contained in exactly $N\cdot \by_i$  of the sets $A_1,\ldots, A_{M}$. That is, $\abs{\{b\in [M]~|~i\in A_b\}}=N\cdot \by_i$. Furthermore $A_b\cap R= A_b\cap G_1=\emptyset$ for every $b\in [M]$.
\end{observation}

The next claim is used to bound the number of items in $R$. 
\begin{claim}
	\label{claim:RcapGj}
	For every $j\in \{2,\ldots, \delta^{-2}\}$ it holds that $\abs{R\cap G_j } \leq 2\cdot \delta^{-1}$.
\end{claim}
\begin{claimproof}
	In case $R\cap G_j =\emptyset$ the claim trivially holds. We henceforth assume $R\cap G_j \neq \emptyset$.
	Let $i^*\in R\cap G_j$ such that $\by_{i^*}$ is minimal.  Also, let $B=\left\{ b\in [M]~|~\abs{A_b \cap G_j } <\xi_{j-1}(b) \right\}$, the set of configurations which have remaining placeholders for items in $G_j$. As $i^*\in R$ it must hold  that $ \abs{B} \leq N\cdot \by_{i^*}$. 	
	Since~$\by$ is $\alpha$-scaled we have,
	\begin{equation}
	\label{eq:shifting_first}
	\begin{aligned}
	\alpha \cdot N\cdot  \abs{G_j} \,&\geq\, \sum_{i\in G_j} N \cdot \by_i \\
	&= \, \sum_{b\in [M]} \abs{A_b  \cap  G_j}  +\sum_{i\in G_j \cap R} N\cdot \by_i
	\\
	&=\, \sum_{ b\in [M] \setminus B}  \abs{A_b\cap G_j}+ \sum_{b\in B} \abs{A_b\cap G_{j}} +  \sum_{i\in G_j \cap R} N\cdot \by_i\\
	&\geq \,  \sum_{ b\in [M] \setminus B}  \xi_{j-1}(b) + \sum_{b\in B} \left( \xi_{j-1}(b) - 2\cdot \delta^{-1}\right) +  \sum_{i\in G_j \cap R} N\cdot \by_i\\
	&\geq \, \sum_{b\in [M] } \xi_{j-1}(b) -2\cdot \abs{B} \cdot \delta^{-1} + N\cdot \abs{R\cap G_j}\cdot \by_{i^*} .
	\end{aligned}
	\end{equation}
	The first equality holds as each item $i\in G_j\setminus R$ is contained in exactly $N\cdot\by_i$ of the sets~$A_1,\ldots, A_{M}$ and items in $R$ are not contained in any of the sets $A_1,\ldots, A_{M}$.
	The second inequality  follows from the definition of $B$ and since $\abs{A_b\cap G_j} \leq \xi_{j-1} (b)\leq 2\cdot \delta^{-1}$ for every $b\in[M]$. The last inequality holds since $\by_i\geq \by_{i^*}$ for every $i\in G_j\cap R $ by the selection of $i^*$.  
	
	Furthermore,
	\begin{equation}
		\label{eq:shifting_second}
	 \sum_{b\in [\alpha \cdot N\cdot \ell ] } \xi_{j-1}(b)  \,=\, \sum_{\bp \in \cP} \abs{\xi^{-1} (\bp )} \cdot \bp_{j-1}  \,\geq\,\sum_{\bp \in \cP} \alpha \cdot N \cdot \eta(\bp ) \cdot \bp_{j-1}  \,=\, \alpha \cdot N \cdot \abs{G_{j-1} }\,\geq \, \alpha \cdot N \cdot \abs{G_j},
	\end{equation}
	where the first inequality follows from the definition of $\xi$,  the last equality follows from \Cref{lem:eta_to_groups} and the last inequality follows from \Cref{obs:groups_size}. 
	
	By \eqref{eq:shifting_first},  \eqref{eq:shifting_second}  and $\abs{B}\leq N\cdot \by_{i^*}$ we have
	$$
	\begin{aligned}
	\alpha \cdot N \cdot \abs{G_j} \,&\geq\,
	 \sum_{b\in [\alpha \cdot N\cdot \ell ] } \xi_{j-1}(b) -2\cdot \abs{B} \cdot \delta^{-1} + N\cdot \abs{R\cap G_j}\cdot \by_{i^*} \\ &\geq\, 
	  \alpha \cdot N \cdot \abs{G_j} -2\cdot N \cdot \by^*_{i}\cdot \delta^{-1} + N\cdot \abs{R\cap G_j} \cdot \by_{i^*}.
	  \end{aligned}
 	$$ 
 	By rearranging the terms we have, $\abs{R \cap G_j } \leq 2\cdot \delta^{-1}$.
\end{claimproof}

Since $R\subseteq \bigcup_{j=2}^{\delta^{-2}} {R\cap G_j}$, \Cref{claim:RcapGj} implies that \begin{equation}
	\label{eq:Rsize}
	|R|\leq 2\cdot \delta^{-3} \leq \delta^{-4}.
	\end{equation}
The  items in $R$ will be handled later using \Cref{lem:item_per_bin}.

The next claim will be later used to show that the sets $A_1,\ldots, A_M$ have enough vacant capacity to add the small items in $\supp(\by)$.
\begin{claim}
	\label{claim:A_prop}
For every $b\in [M]$ it holds that $\abs{A_b}\leq \| \xi(b)\|$ and $\w(A_b)\leq \w\left(\xi(b)\right)$.	
\end{claim}
\begin{claimproof}
Since $A_b\subseteq L\setminus G_1$ it holds that
$$
\abs{A_b} \,=\,\sum_{j=2}^{\delta^{-2}} \abs{A_b\cap G_j} \,\leq \, \sum_{j=2}^{\delta^{-2}} \xi_{j-1}(b)\, \leq \, \sum_{j=1}^{\delta^{-2}} \xi_{j}(b)\, = \,\|\xi(b)\|.
$$
Similarly,
$$
\w\left(A_b\right) \,=\,\sum_{j=2}^{\delta^{-2}} \sum_{i\in A_b\cap G_j} \w(i)  \,\leq \, \sum_{j=2}^{\delta^{-2}} \sum_{\,i\in A_b\cap G_j\,}  \w_{j-1}\, \leq \,
\sum_{j=2}^{\delta^{-2}} \xi_{j-1}(b)\cdot\w_{j-1}
\, \leq\, 
\sum_{j=1}^{\delta^{-2}} \xi_{j}(b)\cdot \w_{j}
\,=\,
\w\left(\xi(b)\right),
$$
where the first inequality follows from \Cref{obs:shifting_argument} and the second inequality holds as $\abs{A_b\cap G_j}\leq \xi_{j-1}(b)$.
\end{claimproof}

The next step in the construction of $\bx$ is  to add every item $i\in S\setminus L$  to $N\cdot \by_i$ of the sets $A_1,\ldots,A_{M}$ while ensuring the resulting sets are configurations. 
The assignment utilizes a linear program defined by the polytope of a partition matroid and two linear inequalities per set $A_b$. 
 We rely on a known property \cite{grandoni2010approximation}  that basic solutions for such linear programs have few fractional entries.

The ground set of the matroid is $E=S\setminus L\times [M]$. The element $(i,b)\in E$ represents an assignment of $i\in S\setminus L$ to the set $A_b$.
The {\em partition} of $i\in S\setminus L$ is $T_i=\left\{ (i,b)~|~b\in [M]\right\}$, and its elements  can be  intuitively viewed  as $M$ copies of $i$.
 The independent sets of the matroid contain  at most $N\cdot \by_i$ copies of $i$. That is, 
   $$\cM = \left\{ Z\subseteq E~\middle|~ \forall i\in S\setminus L:~\abs{Z\cap T_i} \leq N\cdot \by_i \right\}.$$
   It is well known that $(E,\cM)$ is a matroid, and specifically a 
   {\em partition matroid} (we refer the reader to, e.g.,  \cite{schrijver2003combinatorial} for a formal definition of matroids). 
   
  Let $P_{\cM}$ be the matroid  polytope of $(E,\cM)$. That is, $P_{\cM}$ is the convex hull of the characteristic vectors of $\cM$. It can be easily observed that 
  \begin{equation}
  	\label{eq:partition_poly}
  P_{\cM} =  \left\{ \bz\ \in[0,1]^{ E}~\middle|~\forall i\in S\setminus L:~~~
  	\sum_{b\in [M]} \bz_{i,b} = N\cdot \by_i  
  \right\}.
  \end{equation}

We consider a linear program whose  feasible region is the intersection of $P_{\cM}$ with $2\cdot M$ linear inequalites.
\begin{equation}
	\label{eq:assignLP}
\ALP~~~~~~~
\begin{aligned}
	&\textnormal{ max }~~&&\sum_{(i,b)\in E} \bz_{i,b} \\
	&\textnormal{ s.t. } &&  \bz\in P_{\cM}\\
	&&&\sum_{i\in S\setminus L} \bz_{i,b}+\abs{A_b} \leq k && \forall b\in [M]\\
	&&&\sum_{i\in S\setminus L} \bz_{i,b}\cdot \w(i)+\w(A_b)  \leq 1 && \forall b\in [M]
\end{aligned}
\end{equation}
The value of $\bz_{i,b}$ is viewed as a fractional assignment of $i\in S\setminus L$ to the set $A_b$. The constraint  $\bz\in P_{\cM}$ ensures each element in $E$ is assigned at most once, and  that the total assignment of an item $i\in I$ is at most   $N\cdot \by_i$.
 The  additional constraints of  $\ALP$ require that the total number of items in $A_b$ together with the fractionally assigned items is at most $k$, and that the total weight of items in $A_b$ and the fractionally assigned items is at most $1$.
 
 In \cite{grandoni2010approximation} the authors proved that in a  basic solution of a linear program defined by a matroid polytope and additional $\beta$ linear inequalities, the sum of fractional entries  (that is, entries in the interval $(0,1)$) is at most $\beta$. 
  As $\ALP$ is defined by a matorid polytope and $2\cdot M$ additional inequalities we attain the following claim as a consequence of \cite{grandoni2010approximation}.
 \begin{claim}
 	\label{claim:grandoni}
 	Let $\bz$ be a basic optimum solution for $\ALP$. Then $\sum_{(i,b)\in E \textnormal{ s.t. } 0<\bz_{i,b}<1} \bz_{i,b} \leq 2\cdot M$.
 \end{claim}

 The next claim uses the properties of $\by$ as an $\alpha$-scaled vector -- $\sum_{i\in \cK_{\bp} } \by_i \leq \alpha \cdot \abs{\cK_{\bp}}$ and $\sum_{i\in \cK_{\bp}} \by_i \cdot \w(i) \leq \alpha\cdot  \w(\cK_{\bp})$ for every $\bp\in \cP$ -- to lower bound the optimum of $\ALP$.
	\begin{claim}
		\label{claim:ALP}
 	The value of an optimal  solution for $\ALP$ is $N \cdot\sum_{i\in S\setminus L} \by_i$. Furthermore, if $\bz$ is an optimal solution for $\ALP$ then $\sum_{b\in [M] } \bz_{i,b} = N\cdot \by_i$ for every $i\in S\setminus L$.
	\end{claim}
	\begin{claimproof}		
	We  lower bound the value of the optimum by constructing a solution for $\ALP$.
	Define $\bz \in \mathbb{R}^{E}$ by 
		\begin{equation}
			\label{eq:nonempty_zdef}
		\bz_{i,b} \,=\,\begin{cases} 
			\frac{\by_i}{\alpha \cdot \eta(\xi(b)) +1} & i\in \cK_{\xi(b)} \\
			0 & i\notin \cK_{\xi(b)} 
			\end{cases}
		\end{equation}
		for every $(i,b)\in E$. In particular, $\bz_{i,b} =0$ if $i$ does not belong to the class $\cK_{\bp}$ where $\bp$ it the type associated by $\xi$ with the bin $b$.
	It trivially holds that $\bz_{i,b}\geq 0$ for every $(i,b)\in E$. Furthermore, for every $(i,b)\in E$, if $i\in \cK_{\xi(b)}$ then 
	$$
	\bz_{i,b}\, =\, \frac{\by_i}{\alpha \cdot \eta(\xi(b))+1} \,\leq\, \frac{\by_i}{1} \,\leq\, 1,$$
	and if $i\notin \cK_{\xi(b)}$ then $\bz_{i,b}=0\leq 1$. Therefore $\bz \in [0,1]^{E}$.
	
	Let  $i\in S\setminus L$ and let  $\bp\in \cP$ be the unique type  such that $i\in \cK_{\bp}$.  Then it holds that 
	\begin{equation}
		\label{eq:bz_partition}
	\sum_{b\in [M]} \bz_{i,b} \,=\, \sum_{b\in \xi^{-1}(\bp) } \frac{\by_i}{\alpha \cdot \eta(\bp) +1}\,=\, 
	 \frac{\abs{\xi^{-1}(\bp) } \cdot\by_i}{\alpha \cdot \eta(\bp) +1}  \, = \,  \frac{(\alpha \cdot N \cdot \eta(\bp)+N) \cdot\by_i}{\alpha \cdot \eta(\bp)+1} \,=\,N\cdot \by_i,
	\end{equation}
	where the first equality follows from \eqref{eq:nonempty_zdef}, and the third holds by the definition of $\xi$. Following \eqref{eq:partition_poly} we showed that $\bz\in \cP_{\cM}$.
	
	For every $b\in [M]$ it holds that 
	\begin{equation}
		\label{eq:linear_first}
		\begin{aligned}
	\sum_{i\in S\setminus L} \bz_{i,b}\,&=\, \sum_{i\in \cK_{\xi(b)} } \frac{\by_i}{\alpha \cdot \eta(\xi(b)) + 1 }\\
	&\leq\, \frac{\alpha \cdot \abs{\cK_{\xi(b)}}}{\alpha \cdot \eta(\xi(b)) + 1}\\
	&\leq \, \frac{\alpha \cdot \eta(\xi(b)) \cdot (k-\|\xi(b)\|) }{\alpha \cdot \eta(\xi(b)) +1} \\&
	\leq\, k-\|\xi(b)\| \\
	& \leq\, k-\abs{A_b},
	\end{aligned}
	\end{equation}
	where the first inequality holds since $\by$ is $\alpha$-scaled, the second inequality follows from \Cref{lem:class_weight_and_card}, and the last inequality follows from \Cref{claim:A_prop}. Similarly, for every $b\in [M]$ it holds that 
	\begin{equation}
		\label{eq:linear_second}
	\begin{aligned}
	\sum_{i\in S\setminus L } \bz_{i,b} \cdot \w(i)&=\, \sum_{i \in \cK_{\xi(b)}} \frac{\by_i}{\alpha \cdot \eta(\xi(b))+1} \cdot \w(i)\\
	&\leq \,  \frac{\alpha \cdot \w(\cK_{\xi(b)})} { \alpha\cdot \eta(\xi(b)) +1} \\
	& \leq \, \frac{ \alpha \cdot \eta(\xi(b)) \left(1-\w(\xi(b))\right) }{\alpha \cdot \eta(\xi(b))+1}\\
	&\leq \, 1-\w(\xi(b)) \\
	&\leq \, 1-\w(A_b).
	\end{aligned}
	\end{equation}
	The first inequality follows from $\sum_{i\in \cK_{\xi(b)}} \by_i \cdot \w(i) \leq \alpha \cdot \w\left( \cK_{\xi(b)}\right)$, as $\by$ is $\alpha$-scaled. 
	The second inequality follows from \Cref{lem:class_weight_and_card} and the last inequality follows from \Cref{claim:A_prop}. By~\eqref{eq:linear_first},~\eqref{eq:linear_second}   and since $\bz\in P_{\cM}$ it holds that $\bz$ is a feasible solution for $\ALP$ \eqref{eq:assignLP}. 
	
	The value of $\bz$ as a solution for $\ALP$ is 
	\begin{equation}
		\label{eq:z_value}
	\sum_{ (i,b) \in E} \bz_{i,b} \, = \, \sum_{i\in S\setminus L } \sum_{b\in [M]} \bz_{i,b}\, = \, \sum_{i\in S\setminus L } N\cdot \by_i \, = \,N\cdot \sum_{i\in S\setminus L }  \by_i ,
\end{equation}
	where the second equality follows from \eqref{eq:bz_partition}.

	Let $\bz'$  be an optimal solution for  $\ALP$ then
	\begin{equation}
		\label{eq:ALP_opt}
	N\cdot \sum_{i\in S\setminus L }  \by_i \,\leq\,\sum_{ (i,b) \in E} \bz'_{i,b} \, = \, \sum_{i\in S\setminus L } \sum_{b\in [M]} \bz'_{i,b}\, \leq \, \sum_{i\in S\setminus L } N\cdot \by_i \, = \,N\cdot \sum_{i\in S\setminus L }  \by_i .
	\end{equation}
	The first inequality holds since $\bz$ is a solution, while $\bz'$ is an optimal solution, and due to \eqref{eq:z_value}. The second inequality follows from 
	$\sum_{b\in [M]}\bz'_{i,b}  \leq N\cdot \by_i$ as $\bz'\in P_{\cM}$ (see \eqref{eq:partition_poly}). Thus, 
	$\sum_{(i,b)\in E} \bz'_{i,b}=N\cdot \sum_{i\in S\setminus L }  \by_i $. That is, the value of an optimal solution for $\ALP$ is $N\cdot \sum_{i\in S\setminus L }  \by_i$. Furthermore, if there is $i^*\in S\setminus L$ such that $\sum_{b\in [M]} \bz_{i^*,b}< N\cdot \by_{i^*}$ then 
	$$
	\sum_{ (i,b) \in E} \bz'_{i,b} \, = \, \sum_{i\in S\setminus L } \sum_{b\in [M]} \bz'_{i,b}\, < \, \sum_{i\in S\setminus L } N\cdot \by_i \, = \,N\cdot \sum_{i\in S\setminus L }  \by_i 
	$$
	contradicting \eqref{eq:ALP_opt}. Thus $\sum_{b\in [M]} \bz'_{i,b}= N\cdot \by_i$ for every $i\in S\setminus L$.
	
	\end{claimproof}

Let $\bz^*$ be a basic optimal solution for $\ALP$. Define 
\begin{equation}
\label{eq:Fdef}
F= \{ (i,b)\in E~|~0<\bz_{i,b}<1\}
\end{equation} be the set of  fractional entries of $\bz^*$. By \Cref{claim:grandoni} it holds that 
\begin{equation}
	\label{eq:Fbound}
	\sum_{(i,b)\in F } \bz^*_{i,b}\leq 2\cdot M.
	\end{equation}
Also, for every $i\in S\setminus L$ define 
\begin{equation}
	\label{eq:betadef}
	\beta_i ~=~\frac{1}{N}\cdot\sum_{b\in [M]} \floor{\bz^*_{i,b}}
\end{equation}
to be the total integral assignment of $i$ by $\bz^*$.

For every $b\in [M]$ define 
\begin{equation}
	\label{eq:Astar_def}
A^*_b = A_b\cup \left\{i\in S\setminus L~\middle|~\bz^*_{i,b}=1\right\}.
\end{equation}

That is, in $A^*_1,\ldots, A^*_M$ we add to the set  $A_b$ the items fully assigned to  the $b$-th configuration by~$\bz^*$.
\begin{claim}
For every $b\in [M]$ it holds that $A^*_b\in \cC$.
\end{claim}
\begin{claimproof}
For every $b\in [M]$ it holds that,
$$
\abs{A^*_b} \,=\, \abs{A_b} + \abs{\{i\in S\setminus L ~|~\bz^*_{i,b}=1\}}\,\leq \abs{A_b} + \sum_{i\in S\setminus L } \bz^*_{i,b} \,\leq \, k,
$$
where the second inequality holds as $\bz^*$ is a feasible solution for $\ALP$.  By a similar argument, for every $b\in [M]$ we have
$$
\w(A^*_b) \,=\,\w(A^*_b)  +\w\left( \{i\in S\setminus L ~|~\bz^*_{i,b}=1\}\right)  \, \leq \, \w(A^*_b) +\sum_{i\in S\setminus L}\bz^*_{i,b} \cdot \w(i) \,\leq\, 1.
$$
Therefore, $A^*_b\in \cC$ for every $b\in [M]$.
\end{claimproof}

We use the configurations $A^*_1,\ldots, A^*_M$ to define an intermediary fractional solution $\bx^*\in \mathbb{R}^{\cC}_{\geq 0}$ by $\bx^*_C = \frac{1}{N}\cdot \abs{ \left\{ b\in [M]~|~A^*_b = C\right\}}$ for every $C\in \cC$. 
It holds that $\|\bx\| = \frac{M}{N}$. Furthermore,
for every $i\in I$ it holds that 
\begin{equation}
	\label{eq:xstar_cover}
	\cover_i(\bx^*) \, = \, \sum_{C\in \cC(i)} \bx^*_C \, = \,\sum_{C\in \cC(i) } \frac{1}{N} \cdot \abs{ \left\{ b\in [M]~|~A^*_b = C\right\}} \, = \, \frac{1}{N}\cdot  \abs{\{b\in [M]~|~i\in A^*_b\}}.
	\end{equation}
\begin{claim}
	\label{claim:xstar_cover}
	For every $i\in L\setminus( G_1 \cup R)$   it holds that $\cover_i(\bx^*) = \by_i$, for every $i\in G_1\cup R$ it holds that $\cover_i(\bx^*)=0$, for every 
	$i\in S\setminus L$ it holds that $\cover_i(\bx^*) = \beta_i$, and for every $i\in I\setminus S$ it holds that $\cover_i(\bx^*) = 0$.
\end{claim}
\begin{claimproof}
	For every item  $i\in L\setminus( G_1\cup R)$  it  holds that 
	$$
	\cover_i(\bx^*) = \frac{1}{N}\cdot  \abs{\{b\in [M]~|~i\in A^*_b\}} =\,\frac{1}{N}\cdot  \abs{\{b\in [M]~|~i\in A_b\}}
	\,=
	\,\frac{1}{N}\cdot N\cdot \by_i \\
	\,=\, \by_i,
	$$
	where the first equality follows from \eqref{eq:xstar_cover}, 
	 the second equality holds since $A^*_b \cap L  = A_b$, and the third equality follows from \Cref{obs:large_items}.
	 By the same arguments, for every $i\in G_1\cup R$ it holds that 
	$$
	\cover_i(\bx^*) = \frac{1}{N}\cdot  \abs{\{b\in [M]~|~i\in A^*_b\}} =\,\frac{1}{N}\cdot  \abs{\{b\in [M]~|~i\in A_b\}}
	\,=
	\,0.
	$$
	
	For every $i\in S\setminus L$ it holds  that
	$$
	\begin{aligned}
	\cover_i(\bx^*) &=\,\frac{1}{N} \cdot \abs{\{b\in [M]~|~i\in A^*_b\}} \\
	&= \,\frac{1}{N} \cdot \abs{\{b\in [M]~|~\bz^*_{i,b}=1\}} 
	\\
	&=\, \frac{1}{N}\sum_{b\in [M]} \floor{\bz^*_{i,b}}\\
	&=\, \beta_i
	\end{aligned}
	$$
	where the first equality is by \eqref{eq:xstar_cover}, the second equality follows from \eqref{eq:Astar_def} and $i\notin A_b$ as $A_b\subseteq L$.
	
	Finally, $A^*_b\subseteq S$ for every $b\in [M]$,  thus by \eqref{eq:xstar_cover} it holds that $\cover_i(\bx^*)=0=\by_i$ for every $i\in I\setminus S$.
\end{claimproof}

Define $\bq \in \mathbb{R}^{I}$ by $\bq_i = \by_i - \beta_i$ for $i\in S\setminus L$ and $\bq_i=0$ for $i\in I\setminus  (S\setminus L)$. Clearly, $\supp(\bq)\subseteq S\setminus L$. Furthermore, for every $i\in S\setminus L$ it holds that $\bq_i \leq \by_i \leq 1$ and 
$$
\bq_i \,=\, \by_i -\beta_i \,=\, \by_i -\frac{1}{N}\sum_{b\in [M]}\floor{\bz^*_{i,b} } \geq \by_i -\frac{1}{N}\sum_{b\in [M]}\bz^*_{i,b} \,=\, \by_i-\by_i \,=\,0,
	$$
	where the third equality follows from \Cref{claim:ALP}. 
Thus $\bq\in [0,1]^I$. Furthermore, since $\by\in \mathbb{Q}^I$ and $\beta_i\in \mathbb{Q}$ for every $i\in I$ by \eqref{eq:betadef} it follows that $\bq\in \mathbb{Q}^I$. 
 Hence, by  \Cref{lem:fractional_first_fit} there is fractional solution $\bef$ such that $\cover(\bef)=\bq$ and $\|\bef\| \,\leq \,2\cdot \sum_{i\in I} \bq_i \cdot \tw(i) +1 $.
 We therefore have,
\begin{equation}
	\label{eq:fsize}
	\begin{aligned}
\|\bef\| \,&\leq \, 2\cdot \sum_{i\in I} \bq_i \cdot \tw(i) +1\\ 
&=\, 2\cdot \sum_{i\in S\setminus L} \left(\by_i - \beta_i \right) \cdot \tw(i) +1\\
&=\, 2\cdot \sum_{i\in S\setminus L} \left(\frac{1}{N}\cdot \sum_{b\in [M]} \bz^*_{i,b} - \frac{1}{N}\cdot \sum_{b\in [M]} \floor{\bz^*_{i,b}} \right) \cdot \tw(i) +1\\
&=\, \frac{2}{N} \sum_{(i,b)\in F} \bz^*_{i,b} \cdot \tw(i) +1\\
&\leq\, \frac{2}{N} \sum_{(i,b)\in F} \bz^*_{i,b} \cdot \delta +1\\
&\leq\, \frac{2}{N}\cdot \delta \cdot 2\cdot M +1\\
&=\, 4\cdot \delta\cdot \frac{M}{N}+1.
\end{aligned}
\end{equation}
The second equality follows from \Cref{claim:ALP} and \eqref{eq:betadef},  the third equality follows from the definition of $F$ in \eqref{eq:Fdef}, the second inequality holds as $\tw(i)\leq \delta$ for items in $S\setminus L$, and the third inequality follows from~\eqref{eq:Fbound}.

Finally, by \Cref{lem:item_per_bin} there is a fractional solution $\bd$ such that $\cover(\bd) = \by\wedge \one_{G_1\cup R}$ and $\|\bd\| =\| \by \wedge \one_{G_1\cup R}\|$. 
By \eqref{eq:Rsize} it holds that $|R|\leq \delta^{-4}$, and by \Cref{obs:groups_size} and \eqref{eq:s_bound} it holds that $\abs{G_1} \leq s\leq 2\cdot \delta \cdot \ell +1$. Therefore,
\begin{equation}
	\label{eq:dsize}
	\|\bd\| \,\leq \,\| \by \wedge \one_{G_1\cup R}\| \,\leq\,
	\sum_{i\in G_1} \by_i + \sum_{i\in R} \by_i \, \leq \, \alpha \cdot |G_1| +|R| \, \leq \, 
	 2\cdot\alpha \cdot  \delta \cdot \ell +1+ \delta^{-4},
	\end{equation}
	where the third inequality used the fact that $\by$ is $\alpha$-scaled.

Define $\bx = \bx^* + \bef + \bd$.  By the definition of $\bx^*$, \eqref{eq:fsize} and \eqref{eq:dsize} it holds that 
\begin{equation}
	\label{eq:xnorm}
\begin{aligned}
\|\bx\|\,&=\, \|\bx^*\| + \|\bef\|+\|\bd\| \\
&\leq \, \frac{M}{N} + 4\cdot \delta \cdot \frac{M}{N} +1 + 2\cdot\alpha \cdot  \delta \cdot \ell +1+\delta^{-4} \\
&= \, \alpha \ell + \abs{\cP} +4\cdot \delta \left(\alpha \ell +\abs{P} \right)  +2 +2\cdot \alpha\cdot \delta \cdot \ell +\delta^{-4}\\
&\leq\, \alpha \cdot (1+10\delta)\cdot \ell +\exp(\delta^{-4}).
\end{aligned}
\end{equation}
The first inequality follows from $\|\bx^*\| \leq \frac{M}{N}$, \eqref{eq:fsize} and \eqref{eq:dsize}. The second equality substitutes $M$ with $\alpha \cdot N\cdot \ell +N\cdot \abs{\cP}$, and the last inequality follows from \Cref{lem:num_types}.

To show that $\cover_i(\bx)  = \by_i$ for every $i\in I$ we consider the following cases.
\begin{itemize}
	\item
For every $i\in L\setminus (G_1\cup R)$ it holds that  $\cover_i(\bx^*)=\by_i$ (\Cref{claim:xstar_cover}) and $\cover_i(\bef)=\cover_i(\bd)=0$, thus 
$
\cover_i(\bx)\, = \, \by_i$. 
\item For every $i\in G_1\cup R$ it holds that $\cover_i(\bx) = 0$ (\Cref{claim:xstar_cover}), $\cover_i(\bd) = \by_i$ and $\cover_i(\bef) =0$, thus $
\cover_i(\bx)\, = \, \by_i$. 
\item For every  $i \in S\setminus L$ it holds that $\cover_i(\bx^*) =\beta_i$, $\cover_i(\bef) = \bq_i= \by- \beta_i $ and $\cover_i(\bd)=0$. Therefore $\cover_i (\bx) =\beta_i+\by_i -\beta_i =\by_i$.
\item For every $i\in I\setminus S$ it holds that $\cover_i(\bx^*)=\cover_i(\bef)=\cover_i(\bd)=0$. Thus $\cover_i(\bx)=0=\by_i$ since $\supp(\by_i) \subseteq S$. 
\end{itemize}
Overall, we showed that $\cover(\bx)=\by$, together with \eqref{eq:xnorm} this completes the proof of the lemma.
\end{proof}

\input{refine}

 \subsection{Deferred Proofs of the Simple Constructions}
 
 \label{sec:defered_const}

We are left to prove \Cref{lem:fractional_first_fit} and \Cref{lem:item_per_bin} which were stated without a proof.

\fractionalfirstfit*
\begin{proof}	
	Let  $N\in \mathbb{N}_{>0}$ be a number such that 
	$N\cdot \by_i\in \mathbb{N}_{\geq 0}$ for every $i\in I$. 
	Consider the following procedure which maintains a collection of configurations $A_1,\ldots, A_s$, initialized with $s=0$.  The procedure iterates over the items in $i\in\supp(\by)$ and attempts to add $i$ to $N\cdot \by_i$ configurations among $A_1,\ldots, A_s$. If $i$ cannot be added to sufficiently many configurations then $s$ is increased and  a minimal number of extra configurations are added (initialized by $\emptyset$) so that $i$ can be added to $N\cdot \by_i$ configurations. More concretely, consider the following procedure. 
	
	\begin{enumerate}
		\item Initialize an empty collection $\cR \leftarrow ()$ of configurations. 
		\item For every $i\in\supp(\by)$:
		\begin{enumerate}
			\item If there are configurations $A_1,\ldots, A_{N \cdot \by_i}$ in $\cR$ such that $A_b \cup \{i\} \in \cC$ for all $b \in \left[ N \cdot \by_i \right]$, then remove $A_1,\ldots, A_{N \cdot \by_i}$ from $\cR$ and add $A_1 \cup \{i\},\ldots, A_{N \cdot \by_i} \cup \{i\}$ to $\cR$.  
			\item Otherwise, let $A_1,\ldots, A_{c}$ be all configurations in $\cR$ such that $A_b \cup \{i\} \in \cC$ for all $b \in \left[ c \right]$. Then, remove $A_1,\ldots, A_{c}$ from $\cR$, add $A_1 \cup \{i\},\ldots, A_{c} \cup \{i\}$ to $\cR$, and add extra $N 
		\cdot \by_i -c$ configurations $A'_1,\ldots, A'_{N 
			\cdot \by_i -c}$ to $\cR$ such that $A'_b = \{i\}$ for all $b \in \left[  N 
		\cdot \by_i -c\right]$.  
		\end{enumerate}
	\end{enumerate}
	
	Let $A_1,\ldots, A_s$ be the configurations in $\cR$ at the end of the procedure. Also, let $i$ be the  item in the last iteration in which $s$, the number of configurations in $\cR$, has been increased. 
	Let $B=\{b\in [s]~|~i\in A_b\}$ be the set of indices of  configurations which contain $i$. Since the procedure adds $i$ to $N\cdot \by_i$ configurations it follows that $\abs{B} = N\cdot \by_i\leq N$.  For every $b\in [s]\setminus B$ it holds that $A_b\cup \{i\} \notin \cC$ (otherwise $i$ would have been added to $A_b$ and  the increase in $s$ would have been smaller). In particular, by \Cref{obs:set_adjusted_weight}, we have $\tw(A_b)\geq 1-\delta \geq \frac{1}{2}$  for every~$b\in [s]\setminus B$, as $i\in S\setminus L$. Thus, \begin{equation}
		\label{eq:frac_ff_first}
		\sum_{b\in [s]} \tw(A_b)\,\geq\, \sum_{b\in [s]\setminus B} \tw(A_b) \,\geq\, \sum_{b\in [s]\setminus B}  \frac{1}{2} \,=\,\frac{1}{2}\cdot\left(s-|B|\right)\,\geq \frac{1}{2}\cdot (s-N).
	\end{equation}
	On the other hand, since each item $i\in I$ appears in exactly $N\cdot \by_i$ configurations among $A_1,\ldots, A_s$, we also have
	\begin{equation}
		\label{eq:frac_ff_second}
		\sum_{b\in [s]} \tw(A_b)\, =\, \sum_{i\in I} N\cdot \by_i\cdot \tw(i)\,=\, N\cdot \sum_{i\in I}  \tw(i) \cdot \by_i.
	\end{equation}
	By \eqref{eq:frac_ff_first} and \eqref{eq:frac_ff_second} we have \begin{equation}
		\label{eq:frac_ff_bound_on_s}
		\frac{s}{N} \leq 2\cdot \sum_{i\in I}\tw(i)\cdot \by_i +1.\end{equation}
	
	Define a fractional solution $\bx\in \mathbb{R}_{\geq 0}^{\cC}$ by $\bx_{C} = \frac{1}{N}\cdot \abs{\{b\in[s]~|~A_b=C\} }$ for every $C\in\cC$. That is, $\bx_C$ is the number of times $C$ appears in $A_1,\ldots, A_s$ divided by $N$.  It trivially holds that 
	$$
	\|\bx\|  \,=\, \sum_{C\in \cC}\bx_C \,=\, \sum_{C\in \cC} \frac{ \abs{\{b\in[s]~|~A_b=C\} }}{N} \,=\, \frac{s}{N} \,\leq \, 2\cdot \sum_{i\in I} \tw(i)\cdot \by_i+1,$$
	where the last inequality follows from \eqref{eq:frac_ff_bound_on_s}.
	Also, for every $i\in I$ it holds that 
	$$
	\cover_i(\bx) \,=\, \sum_{C\in \cC(i)} \bx_C = \, \sum_{C\in \cC(i)}  \frac{ \abs{\{b\in[s]~|~A_b=C\} }}{N}\,=\, \frac{\abs{\{b\in[s]~|~i\in A_b\}  }}{N} \,=\,  \frac{N\cdot \by_i}{N} \,=\, \by_i,
	$$ 
	where the fourth equality holds as $N\cdot \by_i$ configurations among $A_1,\ldots, A_s$ contain $i$.
	
\end{proof}

\itemperbin*
\begin{proof}
	Define $\bx\in \mathbb{R}^{\cC}_{\geq 0}$ by $\bx_{\{i\}}=\by_i$ for every $i\in I$  and $\bx_C= 0$ for every $C\in \cC$ such that $\abs{C}\neq 1$. By the construction, 
	for every $i\in I$ it holds that
	$	\cover_i(\bx) =\by_i 	$
	and $\|\bx\| = \sum_{i \in I} \bx_{\{i\}} = \|\by\|$.
\end{proof}

%% file: refine.tex
\subsection{Refinement for the Small Items and the Final Structure}

To complete the construction of the liner  structure, we need  to further partition the classes of small items into sub-classes.

We say a type $\bp \in \cP$ is {\em degenerate} if $\tw\left(\cK_{\bp}\right) \leq  \ell \cdot \frac{\delta^3}{\abs{\cP}}$, otherwise it is {\em non-degenerate}. Let $\cD\subseteq \cP$ be the set of all degenerate types. 
We will define a set of vectors for every group $G_j$  and class $\cK_{\bp}$. The final structure $\cL$ is the union of those sets of vectors. 
For every degenerate type $\bp\in \cD $ we define the set of vectors to be $\cL_{\bp}=\emptyset$. 

For every non-degenerate type $\bp \in \cP\setminus \cD$ we further partition the items in $\cK_{\bp}$ into subclasses. Let $r_{\bp}=\abs{\cK_{\bp}}$
 and  $\cK_{\bp}=\{i_{\bp,1}\ldots, i_{\bp,r_{\bp}}\}$ where $\w(i_{\bp,1})\geq \ldots \geq \w(i_{\bp,r_{\bp}})$. That is, the items are sorted in a non-increasing order by weights. 
Define $h_{\bp,0}=0$ and  
\begin{equation}
	\label{eq:subclass_def}
\forall j\in [\delta^{-2}]:~~~~
h_{\bp,j} = \min\left\{ s\in [r_{\bp}]~\middle|~\sum_{s' =1}^{s} \tw(i_{\bp,s'}) \,\geq\, j\cdot \delta^2\cdot  \tw\left( \cK_{\bp}\right)\right\}.
\end{equation}
Also,  define  $H_{\bp, j} =\left\{i_{\bp,s}~\middle|~ h_{\bp,j-1}<s\leq h_{\bp,j}\right\}$, the $j$-th subclass of $\cK_{\bp}$. The subclasses partition $\cK_{\bp}$ into subclasses of roughly equal adjusted weight. 
\begin{lemma}
	\label{lem:adjusted_subclass_weight}
For every $\bp\in \cP\setminus \cD$ and $j\in [\delta^{-2}]$ it holds that 
$$
\frac{\delta^{6}}{\abs{\cP}} \cdot \ell \,\leq\,\delta^2\cdot \tw\left(\cK_{\bp}\right)-\delta\,\leq\,  \tw\left(H_{\bp,j}\right)\,\leq\, \delta^2\cdot \tw\left(\cK_{\bp}\right)+\delta 
$$
\end{lemma}
\begin{proof}
For every $j'\in[\delta^{-2}]$, by \eqref{eq:subclass_def}, it holds that 
\begin{equation}
	\label{eq:subclass_excluding}
	\sum_{s=1}^{h_{\bp,j'}-1}\tw(i_{\bp,s})\,<\,j'\cdot \delta^{2} \cdot \tw\left(\cK_{\bp}\right). 
\end{equation}
Therefore,
\begin{equation}
	\label{eq:subclass_weight_upper}
\begin{aligned}
\tw\left(H_{\bp,j}\right)\,&=\, \sum_{s=h_{\bp,j-1}+1}^{h_{\bp,j}} \tw(i_{\bp,s})\\
& =\, \sum_{s=1}^{h_{\bp,j}-1} \tw(i_{\bp,s}) + \tw(i_{\bp, h_{\bp,j}}) -\sum_{s=1}^{h_{\bp,j-1}} \tw(i_{\bp,s})\\
&\leq \, j\cdot \delta^{2} \cdot \tw\left(\cK_{\bp}\right)  + \delta -(j-1)\cdot \delta^{2} \cdot\tw\left(\cK_{\bp}\right) \\
&= \delta^2\cdot \tw\left(\cK_{\bp}\right) + \delta,
\end{aligned}
\end{equation}
where the inequality follows from \eqref{eq:subclass_def}, \eqref{eq:subclass_excluding} and since $\tw(i_{\bp, h_{\bp,j}})\leq \delta$ as $i_{\bp, h_{\bp,j}}$ is  a small item. 

To establish a lower bound on $\tw(H_{\bp,j})$ we consider two cases.
\begin{itemize}
	\item
In case $j=1$ then 
$$\tw(H_{\bp,j}) \,=\, \sum_{s=1}^{h_{\bp,j}} \tw(i_{\bp,s}) \, \geq \,1\cdot \delta^{2}\cdot \tw\left(\cK_{\bp}\right)\, \geq \, \delta^{2}\cdot \tw\left(\cK_{\bp}\right)  - \delta$$
by \eqref{eq:subclass_def}.
\item 
In case $j>1$ it holds that 
\begin{equation*}
	\begin{aligned}
		\tw\left(H_{\bp,j}\right)\,&=\, \sum_{s=h_{\bp,j-1}+1}^{h_{\bp,j}} \tw(i_{\bp,s})\\
		& =\, \sum_{s=1}^{h_{\bp,j}} \tw(i_{\bp,s})  -\sum_{s=1}^{h_{\bp,j-1}-1} \tw(i_{\bp,s})-\tw\left( i_{\bp, h_{\bp,j-1} }\right)\\
		&\geq \, j\cdot \delta^{2} \cdot \tw\left(\cK_{\bp}\right)-(j-1)\cdot \delta^{2} \cdot\tw\left(\cK_{\bp}\right) -\delta \\
		&=\, \delta^2\cdot \tw\left(\cK_{\bp}\right) - \delta,
	\end{aligned}
\end{equation*}
where the inequality follows from \eqref{eq:subclass_def}, \eqref{eq:subclass_excluding} and since $\tw(i_{\bp, h_{\bp,j-1}})\leq \delta$ as $i_{\bp, h_{\bp,j-1}}$ is s a small item. 
\end{itemize}
In both cases,
\begin{equation}
	\label{eq:subclass_weight_lower}
\tw(H_{\bp,j})\,\geq \,\delta^2 \cdot \tw\left(\cK_{\bp}\right) -\delta
\, \geq \,  \delta^{2} \cdot \frac{\delta^3}{\abs{\cP}} \cdot \ell -\delta
\, \geq \,  \frac{\delta^5}{\abs{\cP}} \cdot \ell - \frac{\delta^6}{\abs{\cP}} \cdot \ell
 \, \geq \, \frac{\delta^6}{\abs{\cP}} \cdot \ell ,
\end{equation}
the second inequality holds as $\bp$ is non-degenerate, and the third inequality follows from $\ell \geq  \frac{\abs{\cP}}{\delta^5}$.  The lemma  follows from \eqref{eq:subclass_weight_upper} and \eqref{eq:subclass_weight_lower}.
\end{proof}

Define $\bw \in \mathbb{R}^I_{\geq 0}$ by $\bw_i =\w(i)$ for all $i\in I$.
Also, for any two vector $\bgam,\blam \in \mathbb{R}^{I}$ we use $\bgam \wedge \blam$ to denote their element-wise minimum. That is $\bgam \wedge \blam=\bmu$  where $\bmu_{i} =\min(\bgam_i,\blam_i)$ for all $i\in I$.
 The set of vectors associated with the class $\bp \in \cP\setminus \cD$ is $$\cL_{\bp} =\left\{\one_{H_{\bp,j} } ~|~j\in [\delta^{-2}] \right\}\,\cup\,\left\{\bw\wedge\one_{H_{\bp,j} } ~|~j\in [\delta^{-2}] \right\}.$$
That is, each subclass $H_{\bp,j}$ contributes two vectors. One that limits its cardinality and another which limits its weight.
 Also, for every $\bp\in \cP\setminus \cD$ and $j\in [\delta^{-2}]$ define $\w_{\bp, j} = \min_{i\in H_{\bp,j} } \w(i)$ to be the minimal weight of an item in the subclass $H_{\bp,j}$. By the definition of the subclasses it holds that $\w_{\bp,1} \geq \w_{\bp,2}\geq \ldots\geq \w_{\bp,\delta^{-2}}$ and $\w(i)\leq \w_{\bp,j-1}$ for every $j\in \{2,\ldots, \delta^{-2}\}$ and $i\in H_{\bp,j}$.
We observe some basic  properties of the vectors in $\cL_{\bp}$.
\begin{lemma}
	\label{lem:weighted_tol}
Let $\bp\in \cP\setminus \cD$ and $j\in \{2,\ldots, \delta^{-2}\}$. Then $\tol\left( \bw \wedge \one_{H_{\bp,j}} \right) \leq \w(H_{\bp,j-1})\cdot 2\cdot  \frac{\abs{\cP}\cdot\delta^{-6}}{\ell}$.
\end{lemma}
\begin{proof}
	By the definition of $\tol$ it holds that 
	$$
	\tol\left( \bw \wedge \one_{H_{\bp,j}} \right) \, = \, \max\left\{ \w(C\cap H_{\bp,j})  ~|~C\in \cC\right\}\,\leq \, 1,
	$$
	and 
	$$
	\tol\left( \bw \wedge \one_{H_{\bp,j}} \right) \, = \, \max\left\{ \w(C\cap H_{\bp,j})  ~|~C\in \cC\right\}\,\leq \, k \cdot \w_{\bp,j-1},
	$$
	where the last inequality holds since $C\cap H_{\bp, j}$ is a set of cardinality at most $k$, and each item in it has weight at most $\w_{\bp,j-1}$. Therefore,
	\begin{equation}
		\label{eq:weighted_tol}
	 \tol\left( \bw \wedge \one_{H_{\bp,j}} \right) \, \leq\,\min \left\{ 1, k\cdot \w_{\bp, j-1}\right\}\,\leq \,2\cdot \left(\frac{1}{k\cdot \w_{\bp,j-1}}+1 \right)^{-1},
	\end{equation}
	where the last inequality holds since $\min\{1,x\} \leq 2\cdot \left( x^{-1}+1\right)^{-1}$ for every $x\geq 0$, where, with a slight abuse of notation, we assume $\frac{1}{0} = \infty$ and $\frac{1}{\infty}=0$ in case $x=0$.
	
	We can also use $\tw(H_{\bp,j-1})$ to lower bound $\w(H_{\bp,j-1})$. It holds that 
	$$
	\begin{aligned}
	\tw(H_{\bp,j-1}) &=\, \frac{\abs{H_{\bp,j-1}}}{k} + \w(H_{\bp,j-1}) \\
	& \leq \, \frac{1}{k}\cdot \frac{\w(H_{\bp, j-1})}{\w_{\bp,j-1}} +\w(H_{\bp,j-1}) \\
	&=\,  \w(H_{\bp,j-1})\cdot \left( \frac{1}{k\cdot \w_{\bp,j-1}}+1\right) \\
	& \leq\,  \w(H_{\bp,j-1}) \cdot \frac{2}{\tol\left( \bw\wedge \one_{H_{\bp,h}}\right)}
	\end{aligned}
	$$
	where the first inequality holds as each item in $\tw(H_{\bp,j-1})$ has weight at least $\w_{\bp,j-1}$ and the last inequality follows from \eqref{eq:weighted_tol}. Therefore, 
	$$
	\tol\left(\bw \wedge \one_{H_{\bp,j}}\right)  \, \leq\, \frac{2\cdot \w(H_{\bp, j-1})}{\tw(H_{\bp,j-1})}\, \leq   \w(H_{\bp,j-1})\cdot2\cdot \frac{\abs{\cP}\cdot\delta^{-6}}{\ell},
	$$
	where the last inequality follows from \Cref{lem:adjusted_subclass_weight}. 
\end{proof}

Similarly to  \Cref{lem:weighted_tol}, we can obtain an upper bound on $\tol\left( \one_{H_{\bp, j}} \right)$.
\begin{lemma}
		\label{lem:cardinality_tol}
	Let $\bp\in \cP\setminus \cD$ and $j\in \{1,\ldots, \delta^{-2}-1\}$. Then $\tol\left( \one_{H_{\bp,j}} \right) \leq \abs{H_{\bp,j+1}}\cdot 2\cdot  \frac{\abs{\cP}\cdot\delta^{-6}}{\ell}$.
\end{lemma}
\begin{proof}
It holds that 
$$
\tol(\one_{H_{\bp,j}}) \,=\, \max\left\{\abs{C\cap H_{\bp,j}}~|~C\in \cC\right\}\, \leq k,
$$
and 
$$
\tol(\one_{H_{\bp,j}}) \,=\, \max\left\{\abs{C\cap H_{\bp,j}}~|~C\in \cC\right\}\, \leq \frac{1}{\w_{\bp,j}},
$$
where the last inequality holds since $\w(C\cap H_{\bp,j})\leq 1$ and each item in $C\cap H_{\bp,j}$ is of weight at least $\w_{\bp,j}$.
Therefore,
\begin{equation}
\label{eq:cardinality_tol}
\tol(\one_{H_{\bp,j}})	\,\leq \,
 \min\left\{ k, \frac{1}{\w_{\bp,j}}\right\}  \,=\, k \cdot \min\left\{ 1, \,\frac{1}{k\cdot \w_{\bp,j}}\right\} \, \leq \, k\cdot 2 \cdot \left( 1+ k\cdot \w_{\bp,j}\right)^{-1},
\end{equation}
where the inequality follows from $\min\{1,x\}\leq 2\cdot \left(1+x^{-1}\right)^{-1}$ for $x\geq 0$ as in the proof of \Cref{lem:weighted_tol}.

It also holds that 
$$
\begin{aligned}
\tw(H_{\bp,j+1}) \,&=\, \frac{\abs{H_{\bp,j+1}}}{k} + \w(H_{\bp, j+1}) \\
& \leq \,  \frac{\abs{H_{\bp,j+1}}}{k} + \abs{H_{\bp,j+1}}  \cdot \w_{\bp, j} \\
& =\, 
 \frac{\abs{H_{\bp,j+1}}}{k }\cdot \left( 1+k\cdot \w_{\bp,j} \right) \\
 &\leq \, \frac{\abs{H_{\bp,j+1}}}{k }\cdot  \frac{2\cdot k }{\tol(\one_{H_{\bp,j}})},
 \end{aligned}
$$
where the first inequality holds since the weight of each item in $H_{\bp,j+1}$ is at most $\w_{\bp,j}$ and the last inequality follows from \eqref{eq:cardinality_tol}. By rearranging the terms and \Cref{lem:adjusted_subclass_weight} we have,
$$
\tol(\one_{H_{\bp,j}})\, \leq \, \frac{2\cdot \abs{H_{\bp,j+1}}}{\tw(H_{\bp, j+1})} \, \leq \,\abs{H_{\bp,j+1}}\cdot2\cdot \frac{\abs{\cP}\cdot\delta^{-6}}{\ell}.
$$
\end{proof}
\Cref{lem:weighted_tol,lem:cardinality_tol} are used in the proof of the following lemma.
\begin{lemma}
	\label{lem:subclass_sums}
	Let $\bp\in \cP\setminus \cD$ , $\by \in [0,1]^{I}$, $t>0$ and $0<\alpha \leq 1$  such that $\by \cdot \bu \leq  \alpha \cdot \one_{S}\cdot \bu +t\cdot \tol(\bu)$ for every $\bu \in\cL_{\bp}$. Then 
	$$\sum_{j=2}^{\delta^{-2}} \sum_{\,i\in H_{\bp, j }\,} \by_i \cdot \w(i) \,\leq\, \w\left( \cK_{\bp}\right)\cdot \left( \alpha + 2\cdot \delta^{-6}\cdot \abs{\cP} \cdot \frac{t}{\ell} \right)
	$$
	and 
	$$\sum_{j=1}^{\delta^{-2}-1} \sum_{\,i\in H_{\bp, j }\,} \by_i \,\leq\, \abs{ \cK_{\bp}}\cdot \left( \alpha + 2\cdot \delta^{-6}\cdot \abs{\cP} \cdot \frac{t}{\ell} \right).
	$$
\end{lemma}
\begin{proof}
	Consider the following equation:
	$$
	\begin{aligned}
		\sum_{j=2}^{\delta^{-2}} \sum_{\,i\in H_{\bp, j }\,} \by_i \cdot \w(i) \,&=\, \sum_{j=2}^{\delta^{-2}}  \left(\bw \wedge \one_{H_{\bp,j}}\right)  \cdot \by \\
		&\leq \, 
		\sum_{j=2}^{\delta^{-2}}   \left( \alpha \cdot \left(\bw \wedge \one_{H_{\bp,j}}\right)  \cdot \one_{S} + t \cdot \tol\left(\bw \wedge \one_{H_{\bp,j}}\right)   \right) \\
		&\leq \, 		\sum_{j=2}^{\delta^{-2}}   \left( \alpha \cdot \w\left(H_{\bp,j} \right)+ t \cdot \w(H_{\bp,j-1})\cdot2\cdot \frac{\abs{\cP}\cdot\delta^{-6}}{\ell} \right)  \\
		&\leq \, \sum_{j=1}^{\delta^{-2}}   \left( \alpha \cdot \w\left(H_{\bp,j} \right)+ t \cdot \w(H_{\bp,j})\cdot2\cdot  \frac{\abs{\cP}\cdot\delta^{-6}}{\ell} \right)\\
		&=\, \left( \alpha +  \frac{t}{\ell}\cdot2 \cdot  \abs{\cP}\cdot\delta^{-6}\right) \cdot \sum_{j=1}^{\delta^{-2}} \w(H_{\bp,j})\\
		&= \,\left( \alpha +  \frac{t}{\ell}\cdot2 \cdot  \abs{\cP}\cdot\delta^{-6}\right) \cdot  \w\left( \cK_{\bp}\right).
			\end{aligned}
$$
The first inequality holds as $\bw\wedge \one_{H_{\bp,j}}\in \cL_{\bp}$ and by the conditions of the lemma. The second inequality follows from \Cref{lem:weighted_tol}. Observe that after the third inequality the sum begins with $j=1$ and not $j=2$ as before.  

Similarly, 
$$
\begin{aligned}
\sum_{j=1}^{\delta^{-2}-1} \sum_{i\in H_{\bp ,j} } \by_i  \,&= \, \sum_{j=1}^{\delta^{-2}-1} \one_{H_{\bp,j} } \cdot \by\\
& \leq \, \sum_{j=1}^{\delta^{-2}-1} \left( \alpha \cdot \one_{H_{\bp,j}} \cdot \one_{S} + t\cdot \tol(\one_{H_{\bp,j}})\right)\\
&\leq \, \sum_{j=1}^{\delta^{-2}-1} \left( \alpha \cdot \abs{{H_{\bp,j}} } + t\cdot \abs{H_{\bp,j+1}}\cdot2\cdot \frac{\abs{\cP}\cdot\delta^{-6}}{\ell} \right)\\ 
&\leq \, \sum_{j=1}^{\delta^{-2}} \left(\alpha \cdot \abs{{H_{\bp,j}} } + t\cdot \abs{H_{\bp,j}}\cdot2\cdot \frac{\abs{\cP}\cdot\delta^{-6}}{\ell} \right)\\
& =\,\left( \alpha + \frac{t}{\ell}\cdot 2 \cdot \delta^{-6}\cdot \abs{\cP}\right) \cdot \sum_{j=1}^{\delta^{-2}} \abs{H_{\bp, j}}\\
&=\, \left( \alpha + \frac{t}{\ell}\cdot 2 \cdot \delta^{-6}\cdot \abs{\cP}\right) \cdot  \abs{\cK_{\bp}} .
\end{aligned}
$$
The first inequality holds as $\one_{H_{\bp,j}}\in \cL_{\bp}$. The second inequality  follows from \Cref{lem:cardinality_tol}. The third inequality changes the range of $j$. 
\end{proof}

To complete the construction of the linear structure, for every $j\in [\delta^{-2}]$ define the set of vectors associated with the group $G_j$ by $\cL_j = \{\one_{G_j}\}$ (recall the groups are defined in \Cref{sec:grouping}). The linear structure is $$\cL= \left(\bigcup_{\bp \in \cP} \cL_{\bp}  \right) \cup \left(\bigcup_{j=1}^{\delta^{-2}} \cL_j \right).$$
Since $\abs{\cL_{\bp}}\leq 2\cdot \delta^{-2}$ for every $\bp \in \cP$ and $\abs{\cL_j}\leq 1$  for every $j\in [\delta^{-2}]$, it holds that 
$$
\abs{\cL} \,\leq \,  2\cdot \abs{\cP }\cdot \delta^{-2}+\delta^{-2} \, \leq 2\cdot \delta^{-2}\cdot \exp\left( \delta^{-3}\right) +\delta^{-2} \, \leq \exp\left(\delta^{-4}\right),
$$
where the second inequality follows from \Cref{lem:num_types}.

The next lemma completes the proof of \Cref{lem:structure}.
\begin{lemma}
Let $\by \in \left( [0,1]^I \cap \mathbb{Q}\right)^{I}$, $t>0$ and $0\leq\alpha \leq 1$  such $\supp(\by)\subseteq S$ and $\by \cdot \bu \leq \alpha \cdot \one_{S} \cdot \bu  + t\cdot \tol(\bu)$ for every $\bu\in \cL$. Then there is a factional solution $\bx$ such that $\cover(\bx)=\by$ and $$\|\bx\|\leq \alpha \cdot \ell + 20 \cdot \delta\cdot \ell + t\cdot \exp(\delta^{-5})+\exp(\delta^{-5})$$
\end{lemma}
\begin{proof}
To prove the lemma we generate several fractional solutions whose cover is $\by\wedge \one_{Z}$ for several disjoint sets $Z\subseteq S$. The sum of those solutions forms the final solution $\bx$. 

Define $D= \bigcup_{\bp \in \cD} \cK_{\bp}$ to be the set of items in classes of degenerate types. It holds that $\supp\left(\by \wedge \one_{D}\right)\subseteq S\setminus L$, therefore by \Cref{lem:fractional_first_fit} there is a fractional solution $\bd$ such that $\cover(\bd)=\by \wedge \one_{D}$ and $\|{\bd}\| \leq 2\cdot  \sum_{i\in I} \left(\by \wedge \one_D\right)_i\cdot \tw(i)+1$. It holds that 
$$
 \sum_{i\in I} \left(\by \wedge \one_D\right)_i\cdot \tw(i)\,= \, \sum_{i\in D} \by_i\cdot \tw(i)\,\leq  \,\sum_{i\in D}  \tw(i)\,
=\,\sum_{\bp\in \cD}\tw(\cK_{\bp}) \, \leq \, \abs{\cD}\cdot \frac{ \delta^3}{\abs{\cP}}\cdot \ell \,\leq\, \delta^{3}\cdot \ell,
$$
where the second inequality holds as $\tw(\cK_{\bp})\leq \frac{\delta^{3}}{\abs{\cP}} \cdot \ell$ when $\bp$ is degenerate, and the last inequality holds since $\cD\subseteq \cP$. Therefore 
\begin{equation}
	\label{eq:bd_size}
	\|{\bd}\| \,\leq\,2\cdot  \sum_{i\in I} \left(\by \wedge \one_D\right)_i\cdot \tw(i)+1\,\leq \, 2\cdot \delta^{3}\cdot \ell +1.
\end{equation}

Define $Q= \bigcup_{\bp\in \cP\setminus \cD}\left( H_{\bp,1}\cup H_{\bp,\delta^{-2}}\right)$, the set of items in the first and last sub-classes of non-degenerate types. It holds that $\supp(\by\cap \one_{Q})\subseteq S\setminus L $, therefore by \Cref{lem:fractional_first_fit} there is a fractional solution $\bq$ such that $\cover(\bq)=\by\wedge \one_{Q}$ and  $\|\bq\| \leq 2\cdot \sum_{i\in Q} \left(\by\wedge \one_{Q}\right)_i \cdot \tw(i)+1$. By \Cref{lem:adjusted_subclass_weight} we have,
$$
	\begin{aligned}
 \sum_{i\in Q} \left(\by\wedge \one_{Q}\right)_i \cdot \tw(i) \,&\leq\, \sum_{\bp\in \cP\setminus \cD} \left(\tw(H_{\bp,1} ) + \tw(H_{\bp,\delta^{-2}})\right) \\
 &\leq \, \sum_{\bp\in \cP\setminus \cD} \left( 2\cdot \delta^{2} \cdot \tw(\cK_{\bp})+ 2\cdot\delta \right) \\
 &\leq\, 2\cdot \delta^{2}\cdot \tw(S)+2\cdot \delta \cdot \abs{\cP}\\
 &\leq\, 4\cdot \delta^{2}\cdot \ell +\abs{\cP}.
 \end{aligned}
$$
The third inequality holds since the classes $\cK_{\bp}$ are disjoint substs of $S$, and the last inequality holds as $S$ can be packed into $\ell$ bins and due to  \Cref{obs:config_adjusted_weight}.  Therefore,
\begin{equation}
	\label{eq:bq_size}
	\|\bq\|\,\leq \, 2\cdot \sum_{i\in Q} \left(\by\wedge \one_{Q}\right)_i \cdot \tw(i)+1\,\leq\,8\cdot \delta^{2}\cdot \ell +2\cdot \abs{\cP}+1.
\end{equation}

Next, for $j\in [\delta^{-2}]$ we say that the group $G_j$ is {\em unsaturated}  if $\abs{G_j}\leq \delta^{4}\cdot \ell$. Let $\cU=\{j\in [\delta^{-2}]~|~\textnormal{$G_j$ is unsaturated}\}$ be the set of indices of unsaturated groups, and let $U=\bigcup_{j\in \cU} G_j$ be the set of items in unsaturated groups.  By \Cref{lem:item_per_bin} there is a fractional solution $\bmu$ such that $\cover(\bmu) =\by \wedge \one_{U}$ and
\begin{equation}
	\label{eq:bmu_size}
\|\bmu\| \,=\,\sum_{i\in I} \left(\by \wedge \one_{U} \right)_i \, \leq \, \abs{U}\, \leq \, \abs{\cU}\cdot \delta^{4}\cdot \ell \,\leq \,\delta^{2}\cdot \ell.
\end{equation}
The second inequality holds since there are at most $\delta^{4}\cdot \ell$ items in each unsaturated group $G_j$ where $j\in \cU$, and the last inequality holds since $\abs{\cU}\leq \delta^{-2}$.

Finally, define $\by^* = \by - \by \wedge\one_{D} -\by\wedge \one_{Q} - \by \wedge \one_{U}$.  We will show that  $\by^*$ is $\left( \alpha + 2\cdot \delta^{-6}\cdot \abs{\cP} \cdot \frac{t}{\ell}\ \right)$-scaled (\Cref{def:compliance}), and then apply \Cref{lem:weak_structure}.
\begin{itemize}
	\item
It holds that $\supp(\by^*)\subseteq \supp(\by)\subseteq S$. 
\item For every $j\in [\delta^{-2}]$, if $j$ is unsaturated ($j\in \cU$) then $\sum_{i\in G_j} \by^*_i = 0$. Otherwise $j\notin \cU$ and therefore,
$$
\begin{aligned}
\sum_{i\in G_j } \by^*_i \,&=\,\sum_{i\in G_j} \by_i \\
& =\, \by \cdot \one_{G_j}\\&
 \leq \, \alpha \cdot \abs{G_j} +t\cdot \tol(\one_{G_j}) \\
&\leq \, \abs{G_j}\cdot \left( \alpha +t\cdot \frac{2\cdot \delta^{-1}}{\abs{G_j}}\right) \\
& \leq \,
\abs{G_j}\cdot \left( \alpha + t\cdot \frac{2\cdot \delta^{-1}}{\delta^{4}\cdot \ell} \right)\\
& \leq \,
\abs{G_j}\cdot \left( \alpha + 2\cdot \delta^{-6}\cdot \abs{\cP} \cdot \frac{t}{\ell} \right).\\
\end{aligned}
$$
The first inequality holds since $\one_{G_j} \in \cL$, the second  inequality holds since $\tol(\one_{G_j})\leq 2\cdot\delta^{-1}$ as a consequence of \Cref{lem:large_items_in_conf}, the third inequality follows from $\abs{G_j}>\delta^{4}\cdot \ell$.
\item 
For every $\bp \in \cP$, in case $\bp$ is degenerate ($\bp\in\cD$)  then $\by^*_i =0$  for every $i\in \cK_{\bp}\subseteq D$. Thus, $\sum_{i\in \cK_{\bp}} \by^*_i = \sum_{i \in \cK_{\bp}} \by^*_i \cdot \w(i) = 0$. In case $\bp$ is non degenerate, it holds that 
$$
\sum_{i\in \cK_{\bp}} \by^*_i \,=\, \sum_{j=2}^{\delta^{-2}-1} \sum_{i\in H_{\bp, j}}\by_i \,\leq \,\abs{ \cK_{\bp}}\cdot \left( \alpha + 2\cdot \delta^{-6}\cdot \abs{\cP} \cdot \frac{t}{\ell}\ \right),
$$
where the equality holds since $H_{\bp,1},H_{\bp,\delta^{-2}}\subseteq Q$, and the inequality follows from \Cref{lem:subclass_sums} as $\cL_{\bp}\subseteq \cL$. By the same arguments, 
$$
\sum_{i\in \cK_{\bp}} \by^*_i \cdot \w(i)\,=\, \sum_{j=2}^{\delta^{-2}-1} \sum_{i\in H_{\bp, j}}\by_i \cdot \w(i)
\,\leq \,\w\left( \cK_{\bp}\right)\cdot \left( \alpha + 2\cdot \delta^{-6}\cdot \abs{\cP} \cdot\frac{t}{\ell} \right).
$$
\end{itemize}
Overall, we showed that $\by^*$ is $\left( \alpha + 2\cdot \delta^{-6}\cdot \abs{\cP} \cdot \frac{t}{\ell}\ \right)$-scaled. Hence, by \Cref{lem:weak_structure} there is a fractional solution $\bx^*$ such that $\cover(\bx^*)=\by^*$ and 
\begin{equation}
	\label{eq:xstar_size}
\|\bx^*\| \leq \left( \alpha + 2\cdot \delta^{-6}\cdot \abs{\cP} \cdot \frac{t}{\ell}\ \right)\cdot (1+10\delta) \cdot \ell+\exp(\delta^{-4}).
\end{equation}

Define, $\bx = \bx^*+\bd+\bq+\bmu$. Then,
$$
\cover(\bx)\,=\,\cover(\bx^*)+\cover(\bd)+\cover(\bq)+\cover(\bmu )\,=\, \by^* +\by\wedge \one_{D}+\by\wedge \one_{Q}+\by\wedge \one_{U}\,=\,\by, 
$$
where the last equality follows from the definition of $\by^*$. Finally, 
$$
\begin{aligned}
\|\bx\| \,&=\,\|\bx^*\| +\|\bd\|+\|\bq\|+\|\bmu\|\\
&\leq  \left( \alpha + 2\cdot \delta^{-6}\cdot \abs{\cP} \cdot \frac{t}{\ell}\ \right)\cdot (1+10\delta)\cdot \ell +\exp(\delta^{-4}) + 2 \cdot \delta^3\cdot \ell +1 +8\cdot \delta^2\cdot \ell +2 \cdot \abs{\cP} +1+ \delta^{2}\cdot \ell \\
&\leq \alpha \cdot \ell + 20 \cdot \delta \cdot \ell + t\cdot 4\cdot \delta^{-6}\cdot \abs{\cP}+\exp(\delta^{-4})+2\abs{\cP}+2\\
&\leq \alpha \cdot \ell + 20 \cdot \delta \cdot \ell + t\cdot \delta^{-7}\cdot \exp\left( \delta^{-3}\right)+\exp(\delta^{-4})+2 \cdot \exp\left( \delta^{-3}\right)+2\\
&\leq \alpha \cdot \ell + 20\delta + t\cdot \exp(\delta^{-2}) \cdot \exp(\delta^{-3})+\exp(\delta^{-5})\\
&= \alpha \cdot \ell + 20\delta + t\cdot \exp(\delta^{-5})+\exp(\delta^{-5}),
\end{aligned}
$$
where the first inequality follows from \eqref{eq:xstar_size}, \eqref{eq:bd_size}, \eqref{eq:bq_size}, \eqref{eq:bmu_size} and the third inequality follows from \Cref{lem:num_types}. 
\end{proof}

%% file: reduction.tex
\section{Reduction to $\eps$-Simple Instances}
\label{sec:reduction}
In this section we give the proof of \Cref{lem:reduction}. The proof bears some similarity to the proof of Lemma 2.3 in~\cite{CKS23}. We consider two cases, depending on the number of bins in the given CMK instance. First, for the case where the number of bins is relatively small, we prove the following.

\begin{lemma}
	\label{thm:fewBins}
	There is an algorithm $\mathcal{G}$ that given a \textnormal{CMK} instance $\cI= (I,\w,\v,m,k)$ and an error parameter $\eps>0$, returns in time $\left(\frac{m}{\eps}\right)^{ O \left( m^2 \cdot {\eps}^{-2}\right) }  \cdot |\cI|^{O(1)}$ a $(1-\eps)$-approximation for $\cI$. 
\end{lemma}

The proof of \Cref{thm:fewBins} is given in \Cref{sec:fewBins}. For the second case, where the number of bins is large, we use \Cref{alg:randomized_rounding} presented earlier in the paper whose properties are formally given in \Cref{thm:simple}. Combined, we can give an EPTAS for CMK. %

\noindent{\bf Proof of \Cref{lem:reduction}:} Recall Algorithms $\cG$ and $\cB$ in \Cref{thm:fewBins} and \Cref{thm:simple}, respectively. Define algorithm $\cD$ that, given a CMK instance  $\cI= (I,\w,\v,m,k)$ and an error parameter $\eps \in (0,0.1)$, proceeds as follows.
\begin{enumerate}
	\item [(i)] If $m\leq  \exp\left(\exp\left({\eps}^{-160}\right)\right)+\left(\eps^{-1}\right)^3$, return the solution computed  by Algorithm $\mathcal{G}$ on $\cI$ and $\eps$. %
	\item [(ii)] If $m>   \exp\left(\exp\left({\eps}^{-160}\right)\right)+\left(\eps^{-1}\right)^3$, return the solution computed by Algorithm $\mathcal{B}$ on $\cI$ and $\eps$.%
\end{enumerate} Observe that by \Cref{thm:fewBins} and \Cref{thm:simple}, the running time of $\cD$ is bounded by $f(\frac{1}{\eps}) \cdot |\cI|^{O(1)}$ for some computable function $f:\mathbb{N} \rightarrow \mathbb{N}$ (a concrete bound on $f$ can be obtained by the statement of \Cref{thm:fewBins} and \Cref{thm:simple}, and the upper bound on $m$ described in case $(i)$ of  algorithm $\cD$). Moreover, by \Cref{thm:fewBins} and \Cref{thm:simple}, $\cD$ returns a solution for $\cI$ of value at least $(1-\eps) \cdot \OPT(\cI)$ with probability at least $1/2$. Thus, $\cD$ is a randomized EPTAS for CMK. \qed

%% file: fewBins.tex
	\subsection{An EPTAS for a Constant Number of Bins (Proof of \Cref{thm:fewBins})}
\label{sec:fewBins}
In this section we give an EPTAS for CMK under the restriction that $m$, the number of bins, is bounded by some function of the given error parameter $\eps$, which gives the proof of \Cref{thm:fewBins}. Our algorithm relies on enumerating the number of items of relatively high value that are assigned to each of the bins; then, extra low value items are added using a {\em linear program} (LP). To properly define the high value (or, {\em valuable}) items, we first need, as an auxiliary algorithm, a constant factor approximation algorithm for CMK; this can be easily obtained combining an FPTAS for {\em knapsack with a cardinality constraint} \cite{li2022faster,caprara2000approximation} along with the local search algorithm of \cite{fleischer2011tight}, where these algorithms are applied with the error parameters, e.g., $\eps' = \frac{1}{4}$ and $\eps'' = \frac{5}{28}$, respectively. This yields the next result. 
	\begin{lemma}
	\label{lem:constApprox}
	There is a (deterministic) polynomial-time algorithm \textnormal{\textsf{Local-Search}} that is a $\frac{1}{4}$-approximation for \textnormal{CMK}.  
\end{lemma} We now define the valuable items w.r.t. the constant factor approximation given by \textsf{Local-Search}. Specifically, given a CMK instance $\cI= (I,\w,\v,m,k)$ and an error parameter $\eps>0$, let $$d(\cI,\eps) = \frac{\eps \cdot \v\left(\textsf{Local-Search}(\cI) \right)}{m}.$$ We say that an item $i \in I$ is a {\em valuable} item in $\cI$ if $\v(i) \geq d(\cI,\eps)$, and let $H(\cI,\eps) = \{i \in I~|~\v(i) \geq d(\cI,\eps)\}$ be the set of valuable items in $\cI$; when clear from the context, we simply use $H = H(\cI,\eps)$. By \Cref{lem:constApprox} it holds that $\textsf{Local-Search}(\cI)$ is a solution of $\cI$; thus, a valuable item $i \in I$ satisfies $\v(i) \geq \frac{\eps \cdot \OPT(\cI)}{m}$, implying that any solution for $\cI$ may contain only a small number of valuable items, motivating our enumeration-based algorithm. As a first step, we partition the valuable items into {\em value groups} as follows. Let $t_{\eps,m} = \left\{1,\ldots,\ceil{ \log_{1-\eps} \left( \frac{4 \cdot \eps}{m}\right)} \right\}$ for simplicity; for all $ r \in t_{\eps,m}$ define the $r$-th value group of $\cI$ and $\eps$ as \begin{equation}
	\label{eq:VG}
	G_r(\cI,\eps) = \left\{i \in H~\bigg|~ \frac{\v(i)}{4 \cdot \v(\textsf{Local-Search}(\cI))} \in \left((1-\eps)^r,(1-\eps)^{r-1}\right)\right\}. 
\end{equation}

Our algorithm enumerates over all vectors $U = (C_1,\ldots,C_m) \in H^m$, containing only subsets of valuable items, and tries to enrich each vector with non-valuable items via an LP. Our LP is defined as follows. For a vector $T$ and an entry $i$ in the vector, we use $T_i$ for the $i$-th entry of the vector $T$. Let $\cI= (I,\w,\v,m,k)$ be a CMK instance, let $\eps>0$ be an error parameter, and let $U \in H^m$. Define
\begin{equation}
	\label{eq:LP}
	\begin{aligned}
		&\textnormal{\textsf{LP}}(\cI,\eps,U):~~~~~~~~ \textnormal{ max }~~&& \sum_{b \in [m]~} \sum_{i \in I \setminus H} \bx_{i,b} \cdot \v(i) \\
		&&&\textnormal{ s.t. } && %
		\\
		&&& \w(U_b) + \sum_{i \in I \setminus H} \bx_{i,b} \cdot \w(i) \leq 1~~~~~~\forall b \in [m]\\
			&&& |U_b| + \sum_{i \in I \setminus H} \bx_{i,b} \leq k~~~~~~~~~~~~~~~~~\forall b \in [m]\\
		&&& \sum_{b \in [m]} \bx_{i,b} \leq 1~~~~~~~~~~~~~~~~~~~~~~~~~~\forall i \in I \setminus H\\
		&&& \bx_{i,b} \in [0,1]~~~~~~~~~~~~~~~~~~~~~~~~~~~~~\forall i \in I \setminus H,~\forall b \in [m].\\
	\end{aligned}
\end{equation}

In the next result, we show that the number of {\em fractional} entries in a basic feasible solution of \eqref{eq:LP} is relatively small.  
	\begin{lemma}
	\label{lem:BFS}
For any \textnormal{CMK} instance $\cI= (I,\w,\v,m,k)$, an error parameter $\eps>0$, a vector $U \in H^m$, and a basic feasible solution $\bx$ of $\textnormal{\textsf{LP}}(\cI,\eps,U)$ it holds that $$\left| \left\{ (i,b) \in \left(I \setminus H \right) \times [m] ~|~ \bx_{i,b} \in (0,1) \right\} \right| \leq 4 \cdot m.$$
\end{lemma} 

\begin{proof}
	We first show that $\textnormal{\textsf{LP}}(\cI,\eps,U)$ lies on a face of a {\em matroid polytope} (for more details on matroids and matroids polytopes we refer the reader to, e.g., \cite{schrijver2003combinatorial}). For all $i \in I \setminus H$ define $F_i = \{(i,b)~|~b \in [m]\}$; in addition, Let $E = (I \setminus H) \times [m]$ and define $$\cF = \{S \subseteq E~|~\forall i \in I \setminus H: ~|S \cap F_i| \leq 1\}.$$   
	It can be easily seen that $(E,\cF)$ is a {\em partition matroid}, which is in particular a matroid (see, e.g., \cite{lawler2001combinatorial}). The matroid polytope of  $(E,\cF)$ can be described as the set of all vectors $\by \in [0,1]^{E}$ satisfying \begin{equation}
		\label{eq:MP}
		\begin{aligned}
			&&& \sum_{b \in [m]} \by_{i,b} \leq 1~~~~~~~~\forall i \in I \setminus H\\
		\end{aligned}
	\end{equation} Note that the constraints of $\textnormal{\textsf{LP}}(\cI,\eps,U)$ are obtained from \eqref{eq:MP} by adding $2 \cdot m$ linear constraints; thus, $\bx$ (the basic feasible solution of $\textnormal{\textsf{LP}}(\cI,\eps,U)$) lies on a face of the matroid polytope \eqref{eq:MP}; the dimension of the face is at most $2 \cdot m$. Therefore, by Theorem 3 in \cite{grandoni2010approximation} there are at most $2 \cdot (2 \cdot m) = 4 \cdot m$ non-integral entries in $\bx$ and the proof follows. 
\end{proof}

Using the above, we can describe our algorithm. Let $\cI= (I,\w,\v,m,k)$ be a CMK instance and let $\eps'>0$ be the error parameter; we will apply the following scheme with a refined parameter $\eps = \min \left\{\frac{1}{2},\frac{\eps'}{4}\right\}$ to reach the desired approximation guarantee of $(1-\eps')$. %
The algorithm iterates over all vectors $N = (n_{r,b}~|~r \in t_{\eps,m}, b \in [m]) \in \left\{0,1,\ldots,4 \cdot \floor{\frac{m}{\eps}}\right\}^{t_{\eps,m} \times [m]}$; such a vector describes a {\em guess} for the number of valuable items, for each value group, taken by an optimal solution for $\cI$. It suffices to enumerate over vectors whose entries are bounded by $4 \cdot \floor{\frac{m}{\eps}}$ as the value of our guessed solution must be bounded by $\OPT(\cI)$ and we consider valuable items exclusively at this point.  

For each such guess $N \in \left\{0,1,\ldots,4 \cdot \floor{\frac{m}{\eps}}\right\}^{t_{\eps,m} \times [m]}$, value group index $r \in t_{\eps,m}$, and bin index $b \in [m]$, we will the take the $N_{r,b}$ (the corresponding entry for $r,b$ in $N$) items in $G_r$ with minimum weight to our solution, or all items in $G_r$ if $|G_r|< N_{r,b}$. Formally, for $r \in t_{\eps,m}$ let $i^r_1,\ldots, i^r_{|G_r|}$ be the items in $G_r$ according to a non decreasing order w.r.t. the weight function $\w$; w.l.o.g., we assume that tie breaks are settled according to some arbitrary predefined order of the items. Finally, for $r \in t_{\eps,m}$ and $n \in \mathbb{N}$ define the {\em first} items in $G_r$ according to the above order as: 
\begin{equation}
	\label{eq:first}
	\begin{aligned}
		\textsf{first}_{\cI}(n,r) ={} & \emptyset~~~~~~~~~~~~~~~~~~~~~~~~~~~~~~~~~~~~~~~~~~~~~~~~~~~~~n = 0\\
			\textsf{first}_{\cI}(n,r) ={} & \emptyset~~~~~~~~~~~~~~~~~~~~~~~~~~~~~~~~~~~~~~~~~~~~~~~~~~~~~G_r = \emptyset\\
				\textsf{first}_{\cI}(n,r) ={} &  \{i^r_1,\ldots,i^r_{q}\}, \text{where } q = \min \{n,|G_r|\}~~~~~~~~\text{else}.\\
	\end{aligned}
\end{equation} For each $N \in \left\{0,1,\ldots,4 \cdot \floor{\frac{m}{\eps}}\right\}^{t_{\eps,m} \times [m]}$ we define a corresponding vector of subsets of items (not necessarily a solution), constructed as follows. For each bin index $b \in [m]$, the $i$-th entry contains the union of the first $N_{r,b}$ items in $G_r$ over all value group indices $r \in t_{\eps,m}$. More concretely, for $N \in \left\{0,1,\ldots,4 \cdot \floor{\frac{m}{\eps}}\right\}^{t_{\eps,m} \times [m]}$ define \begin{equation}
\label{eq:U}
U^N = \left( \bigcup_{r \in t_{\eps,m}}	\textsf{first}_{\cI}\left(N_{r,b},r\right)~\bigg|~b \in [m] \right). 
\end{equation} Our algorithm takes the vector $U^N$ for some guess $N \in \left\{0,1,\ldots,4 \cdot \floor{\frac{m}{\eps}}\right\}^{t_{\eps,m} \times [m]}$, and tries to find a basic feasible solution $\bx^N$ for  $\textnormal{\textsf{LP}}(\cI,\eps,U^N)$ if exists; intuitively, the positive entries in $\bx^N$ represent additional items that can be added to the {\em initial} vector $U^N$. From here, items of positive-integral entries in $\bx^N$ are added to the initial vector $U^N$.  The algorithm returns the {\bf solution} with maximum value encountered through all iterations. The pseudocode of the algorithm is given in \Cref{alg:EPTAS}.  

 \begin{algorithm}[h]
	\caption{$\textsf{ConstantBins}(\cI,\eps')$}
	\label{alg:EPTAS}
	\SetKwInOut{Input}{input}
	\SetKwInOut{Output}{output}

	\Input{ $\cI= (I,\w,\v,m,k)$: a CMK instance, $\eps'>0$: an error parameter.}
	
	\Output{A solution $S$ for $\cI$ of value $\v(S) \geq (1-\eps') \cdot \OPT(\cI)$.}%

Define $\eps = \min \left\{\frac{1}{2},\frac{\eps'}{5}\right\}$.

Initialize an empty solution $S = \left(\emptyset~|~b \in [m]\right)$.

\For{$N \in \left\{0,1,\ldots,4 \cdot \floor{\frac{m}{\eps}}\right\}^{t_{\eps,m} \times [m]}$}{

Find a basic optimal feasible solution $\bx^N$ for $\textnormal{\textsf{LP}}\left(\cI,\eps,U^N\right)$ if exists; otherwise define $\bx^N = 0^{(I \setminus H) \times [m]}$.  

Define $S^N = \left( {U^N}_b \cup \bigcup_{i \in I \setminus H \text{ s.t. } \bx^N_{i,b} = 1} \{i\}~\bigg|~b \in [m] \right)$.

\If{$S^N$ is a solution for $\cI$ and $\v\left(S^N\right) > \v(S)$}{

Update $S \leftarrow S^N$. 

}
}

Return $S$. 

\end{algorithm}

 For the running time analysis of \Cref{alg:EPTAS}, observe that the running time is dominated by the number of vectors iterated by the algorithm, along with the $|\cI|^{O(1)}$ time computation of each iteration.  Hence, in the next result we focus on bounding the number of iterations. 
	\begin{lemma}
	\label{lem:time}
	For every \textnormal{CMK} instance $\cI= (I,\w,\v,m,k)$ and error parameter $\eps'>0$ \textnormal{\Cref{alg:EPTAS}} on input $\cI,\eps'$ runs in time $\left(\frac{m}{\eps'}\right)^{ O \left( m^2 \cdot {\eps'}^{-2}\right) }  \cdot |\cI|^{O(1)}$
\end{lemma}

\begin{proof}
	We first prove an auxiliary inequality. Let $\eps = \min \left\{\frac{1}{2},\frac{\eps'}{5}\right\}$. as defined by the algorithm. 
	\begin{equation}
		\label{eq:AI}
		\log_{1-\eps} \left( \frac{4 \cdot \eps}{m}\right) \leq \frac{\ln \left(\frac{m}{4 \cdot \eps} \right)}{-\ln(1-\eps)} \leq 4 \cdot m \cdot \eps^{-2}. 
	\end{equation} The second inequality holds since $x< -\ln(1-x)$ for all $x>-1, x \neq 0$, and since $\ln(y) < y$ for all $y>0$. 
	Observe that the number of vectors $N \in \left\{0,1,\ldots,4 \cdot \floor{\frac{m}{\eps}}\right\}^{t_{\eps,m} \times [m]}$  is equal to the number of entries multiplied by the number of possibilities for each entry; that is, \begin{equation}
		\label{eq:time}
		\begin{aligned}
			\left|  \left\{ N \in \left\{0,1,\ldots,4 \cdot \floor{\frac{m}{\eps}}\right\}^{t_{\eps,m} \times [m]} \right\} \right| ={} & {4 \cdot \floor{\frac{m}{\eps}}}^{ |t_{\eps}| \cdot m}\\
			\leq{} & {4 \cdot \left( \frac{m}{\eps}+1\right)}^{\ceil{ \log_{1-\eps} \left( \frac{4 \cdot \eps}{m}\right)} \cdot m}  \\ 
				\leq{} & \left( 4 \cdot \left(\frac{m}{\eps}+1\right)\right)^{ \left(\log_{1-\eps} \left( \frac{4 \cdot \eps}{m}\right)+1 \right) \cdot m} \\ 
					\leq{} & \left( 4 \cdot \left(\frac{m}{\eps}+1\right)\right)^{ \left(4 \cdot m \cdot \eps^{-2}+1 \right) \cdot m} \\ 
						\leq{} & \left(8 \cdot \frac{m}{\eps}\right)^{ 5 \cdot m^2 \cdot \eps^{-2}} \\ 
							={} & \left(\frac{m}{\eps}\right)^{ O \left( m^2 \cdot \eps^{-2}\right) } \\ 
								={} & \left(\frac{m}{\eps'}\right)^{ O \left( m^2 \cdot {\eps'}^{-2}\right) } \\ 
		\end{aligned}
	\end{equation} The second inequality follows from \eqref{eq:AI}. Thus, by \eqref{eq:time} the number of iterations of the {\bf for} loop of \Cref{alg:EPTAS} is bounded by $\left(\frac{m}{\eps'}\right)^{ O \left( m^2 \cdot {\eps'}^{-2}\right)}$. Moreover, for each $N \in \left\{0,1,\ldots,4 \cdot \floor{\frac{m}{\eps}}\right\}^{t_{\eps,m} \times [m]}$, computing a basic feasible solution $\bx^N$ for $\textnormal{\textsf{LP}}(\cI,\eps,U^N)$, or determining that there is no such solution can be computed in time $|\cI|$. Therefore, as the remaining steps in the loop can be trivially computed in time $|\cI|^{O(1)}$, the running time of \Cref{alg:EPTAS} can be bounded by  $\left(\frac{m}{\eps'}\right)^{ O \left( m^2 \cdot {\eps'}^{-2}\right) }  \cdot |\cI|^{O(1)}$
\end{proof}

In the next result we show that the approximation guarantee of \Cref{alg:EPTAS} is at least $(1-\eps')$, where $\eps'$ is the given error parameter. In high level, in the proof we consider an optimal solution $Q = (Q_1,\ldots,Q_m)$ for a CMK instance $\cI= (I,\w,\v,m,k)$. Let $\eps'>0$ be the error parameter and let $\eps$ be the refined error parameter defined in the course of the algorithm. In one of the iterations, the algorithm enumerates over a vector $N^* \in  \left\{0,1,\ldots,4 \cdot \floor{\frac{m}{\eps}}\right\}^{t_{\eps,m} \times [m]}$, such that for all $r \in t_{\eps,m}$ and $b \in [m]$ it holds that $N^*_{r,b}$ is equal to the number of items from value group $G_r$ in $Q_b$. Therefore, in this iteration we can show that $S^{N^*}$ (and as a result also the returned solution $S$) is almost as valuable as $Q$, since for  valuable items, we only {\em lose} the value from the (small) difference between any two items from the same value class; additionally, for non-valuable items we only lose the value from non-integral entries in the solution $\bx^{N^*}$ for the LP, where the number of non-integral entries is quite small by \Cref{lem:BFS}.  

\begin{lemma}
	\label{lem:approx}
	For every \textnormal{CMK} instance $\cI= (I,\w,\v,m,k)$ and error parameter $\eps'>0$, \textnormal{\Cref{alg:EPTAS}} on input $\cI,\eps'$ returns a solution $S$ for $\cI$ such that $\v(S) \geq (1-\eps') \cdot \OPT(\cI)$.  
\end{lemma}

\begin{proof}
	Observe that in \Cref{alg:EPTAS} we first initialize a feasible solution $S$. Any update of $S$ within the {\bf for} loop is also to a solution for $\cI$. Thus, we conclude that the algorithm indeed returns a solution for $\cI$. To show the desired approximation guarantee, let $Q = (Q_1,\ldots,Q_m) \in \cC^m$ be a solution for $\cI$ of value $\v(Q) = \OPT(\cI)$. Since the value of $Q$ is defined as $\v(Q) = \v\left(\bigcup_{b\in [m]} Q_b \right)$, we may assume w.l.o.g. that for all $b,b' \in [m], b \neq b'$ it holds that $Q_{b} \cap Q_{b'} = \emptyset$ (which does not reduce the total value of $Q$ w.r.t. any other optimal solution for $\cI$).  Therefore, 
	consider the following elementary auxiliary claim. 
	\begin{claim}
		\label{claim:fewH}
		For all $r \in t_{\eps,m}$ and $b \in [m]$ it holds that $0 \leq |Q_b \cap G_r| \leq 4 \cdot \floor{\frac{m}{\eps}}$. 
	\end{claim}
\begin{claimproof}
	Assume towards a contradiction that there are $r \in t_{\eps,m}$ and $b \in [m]$ such that $|Q_b \cap G_r| > 4 \cdot \floor{\frac{m}{\eps}}$. Then, $|Q_b \cap G_r| > 4 \cdot \frac{m}{\eps} $ and it follows that 
	\begin{equation}
		\label{eq:Cont}
		\v(Q) \geq \v(Q_b) \geq \v(Q_b \cap G_r) \geq |Q_b \cap G_r| \cdot \min_{i \in Q_b \cap G_r} \v(i) > 4 \cdot \frac{m}{\eps} \cdot \frac{\eps \cdot \v\left(\textsf{Local-Search}(\cI) \right)}{m} \geq \OPT(\cI). 
	\end{equation} The fourth inequality holds since $G_r \subseteq H$. The last inequality follows from \Cref{lem:constApprox}. By \eqref{eq:Cont} it holds that $Q$ is a solution for $\cI$ with value strictly larger than $\OPT(\cI)$, which is a contradiction. 
\end{claimproof}

Define a vector $N^* = (N^*_1,\ldots,N^*_m) \in \mathbb{N}^{t_{\eps,m} \times [m]}$ such that for all $r \in t_{\eps,m}$ and $b \in [m]$ we define $N^*_{r,b} = |Q_b \cap G_r|$. Clearly, by \Cref{claim:fewH} it holds that $N^* \in \left\{0,1,\ldots,4 \cdot \floor{\frac{m}{\eps}}\right\}^{t_{\eps,m} \times [m]}$; thus, there is an iteration of the {\bf for} loop in which the algorithm considers the vector $N = N^*$. Recall the vector $U^{N^*}$ defined in \eqref{eq:U}. Define $\by \in \{0,1\}^{(I \setminus H) \times [m]}$ such that for all $i \in I \setminus H$ and $b \in [m]$ define $\by_{i,b} = 1$ if $i \in Q_b$ and define $\by_{i,b} = 0$ if $i \notin Q_b$.             
	\begin{claim}
\label{claim:sol}
$\by$ is a feasible solution for $\textnormal{\textsf{LP}}\left(\cI,\eps,U^{N^*}\right)$. 
\end{claim}
\begin{claimproof}
	We show that $\by$ satisfies the constraints of $\textnormal{\textsf{LP}}\left(\cI,\eps,U^{N^*}\right)$.  By \eqref{eq:first}, for all $b \in [m]$ it holds that \begin{equation}
		\label{eq:W}
		\w\left(U^{N^*}_b\right) = \sum_{r \in t_{\eps,m}}	\w \left( \textsf{first}_{\cI}\left(N^*_{r,b},r\right) \right) \leq \sum_{r \in t_{\eps,m}}	\w \left( Q_b \cap G_r \right) = \w \left( Q_b \cap H\right).  
	\end{equation} Therefore, for all $b \in [m]$ it holds that 
\begin{equation}
	\label{eq:C1}
	\w\left(U^{N^*}_b\right) + \sum_{i \in I \setminus H} \by_{i,b} \cdot \w(i) \leq \w \left( Q_b \cap H\right)+\sum_{i \in Q_b \setminus H} \w(i)  = \w(Q_b) \leq 1. 
\end{equation} The first inequality follows from \eqref{eq:W}. The last inequality holds since $Q$ is a solution for $\cI$. Additionally, for all $b \in [m]$:	\begin{equation}
\label{eq:C2}
\left|U^{N^*}_b\right| + \sum_{i \in I \setminus H} \by_{i,b} = \left|\bigcup_{r \in t_{\eps,m}} N^*_{r,b}\right|+|Q_b \setminus H| = |Q_b \cap H|+|Q_b \setminus H| = |Q_b| \leq k. 
\end{equation} 
Furthermore, for all $ i \in I \setminus H$:
\begin{equation}
	\label{eq:C3}
	\sum_{b \in [m]} \by_{i,b} \leq 1.
\end{equation} The inequality holds since we assume that $b,b' \in [m], b \neq b'$ it holds that $Q_{b} \cap Q_{b'} = \emptyset$. 
Finally, by the definition of $\by$ it holds that $\by_{i,b} \in [0,1]$ for all $i \in I \setminus H$ and $b \in [m]$. Thus, by \eqref{eq:C1},\eqref{eq:C2},\eqref{eq:C3}, and \eqref{eq:LP} the proof follows. 
\end{claimproof}

By \Cref{claim:sol}, the algorithm is able to compute a basic feasible optimal solution $\bx^{N^*}$ in the iteration for $N = N^*$. %
Consider the vector $S^{N^*}$
as defined in the iteration where $N = N^*$. We use the next auxiliary claim. 
	\begin{claim}
	\label{claim:sol2}
	$S^{N^*}$ is a solution for $\cI$. 
\end{claim}
\begin{claimproof}
In the following, we rely on the fact that $\bx^{N^*}$ is a feasible solution for $\textnormal{\textsf{LP}}\left(\cI,\eps,U^{N^*}\right)$. First, for all $b \in [m]$:
\begin{equation}
	\label{eq:D1}
	 \w(S^{N^*}_b) = \w(U^{N^*}_b) + \sum_{i \in I \setminus H \text{ s.t. } \bx^{N^*}_{i,b} = 1} \w(i) \leq \w(U_b) + \sum_{i \in I \setminus H \text{ s.t. } \bx^{N^*}_{i,b} = 1} \bx^{N^*}_{i,b} \cdot \w(i) \leq 1. 
\end{equation} In addition, for all $b \in [m]$: 
\begin{equation}
	\label{eq:D2}
	\left|S^{N^*}_b\right| = \left|U^{N^*}_b\right| + \sum_{i \in I \setminus H \text{ s.t. } \bx^{N^*}_{i,b} = 1} 1 \leq \w(U_b) + \sum_{i \in I \setminus H} \bx^{N^*}_{i,b} \leq k.  
\end{equation}
By \eqref{eq:D1} and \eqref{eq:D2} it follows that $S^{N^*}$ is a solution for $\cI$. 
\end{claimproof} 

By \Cref{claim:sol2} our algorithm returns a solution $S$ with value at least as large as the value of $S^{N^*}$; thus, to conclude we give a lower bound on $\v\left(S^{N^*}\right)$. \begin{equation}
\label{eq:sB}
\begin{aligned}
\v\left(S^{N^*}\right) ={} & \v\left(U^{N^*}\right)+ \sum_{b \in [m]} \sum_{i \in I \setminus H \text{ s.t. } \bx^{N^*}_{i,b} = 1} \v(i) \\
={} & \v\left(U^{N^*}\right)+ \sum_{b \in [m]~} \sum_{i \in I \setminus H} \bx^{N^*}_{i,b} \cdot \v(i)- \sum_{b \in [m]~} \sum_{i \in I \setminus H \text{ s.t. } \bx^{N^*}_{i,b} \in (0,1)}  \bx^{N^*}_{i,b} \cdot \v(i) \\
\geq{} & \v\left(U^{N^*}\right)+ \sum_{b \in [m]~}\sum_{i \in I \setminus H} \bx^{N^*}_{i,b} \cdot \v(i)- \sum_{b \in [m]~}\sum_{i \in I \setminus H \text{ s.t. } \bx^{N^*}_{i,b} \in (0,1)} \v(i) \\
\geq{} & \v\left(U^{N^*}\right)+ \sum_{b \in [m]~}\sum_{i \in I \setminus H} \by^{N^*}_{i,b} \cdot \v(i)- 4 \cdot m \cdot \frac{\eps \cdot \v\left(\textsf{Local-Search}(\cI) \right)}{m}\\
\geq{} & \v\left(U^{N^*}\right)+ \v\left(\bigcup_{b\in [m]} (Q_b \setminus H)\right)- 4 \cdot \eps \cdot \OPT(\cI)\\
\end{aligned}
\end{equation} The second inequality holds since $\bx^{N^*}$ is an optimal solution for $\textnormal{\textsf{LP}}\left(\cI,\eps,U^{N^*}\right)$, since the number of non-integral entries in $\bx^{N^*}$ is at most $4 \cdot m$ by \Cref{lem:BFS}, and finally, since the values for items $i \in I \setminus H$ is bounded by $\frac{\eps \cdot \v\left(\textsf{Local-Search}(\cI) \right)}{m}$. The last inequality follows from the fact that $\v\left(\textsf{Local-Search}(\cI)\right) \leq \OPT(\cI)$ using \Cref{lem:constApprox}. Thus, by \eqref{eq:sB}
\begin{equation}
	\label{eq:sB2}
	\begin{aligned}
		\v\left(S^{N^*}\right) \geq{} & \v\left(U^{N^*}\right)+\v\left(\bigcup_{b\in [m]} (Q_b \setminus H)\right)- 4 \cdot \eps \cdot \OPT(\cI)\\
	\geq{} & \sum_{b \in [m]~}\sum_{r \in t_{\eps,m}} \v\left(\textsf{first}_{\cI}\left(N^*_{r,b},r\right)\right)+\v\left(\bigcup_{b\in [m]} (Q_b \setminus H)\right)- 4 \cdot \eps \cdot \OPT(\cI)\\
		\geq{} & \sum_{b \in [m]~}\sum_{r \in t_{\eps,m}} (1-\eps) \cdot N^*_{r,b}   \cdot  \v(\textsf{Local-Search}(\cI)) \cdot \left(1-\eps\right)^{r-1} +\\
		{} & \v\left(\bigcup_{b\in [m]} (Q_b \setminus H)\right)-4 \cdot \eps \cdot \OPT(\cI).
	\end{aligned}
\end{equation} The second and third inequalities hold by \eqref{eq:U} and \eqref{eq:first}, respectively. Thus, by \eqref{eq:sB2} and \eqref{eq:VG} we have 
\begin{equation}
	\label{eq:sB3}
	\begin{aligned}
		\v\left(S^{N^*}\right) \geq{} & (1-\eps) \cdot \sum_{b \in [m]~}\sum_{r \in t_{\eps,m}} \v(G_r \cap Q_b)+\v\left(\bigcup_{b\in [m]} (Q_b \setminus H)\right)-4 \cdot \eps \cdot \OPT(\cI)\\
		\geq{} & (1-\eps) \cdot \v\left(\bigcup_{b\in [m]} (Q_b \cap H)\right)+\v\left(\bigcup_{b\in [m]} (Q_b \setminus H)\right)-4 \cdot \eps \cdot \OPT(\cI)\\
		\geq{} & (1-\eps) \cdot \v(Q)-4 \cdot \eps \cdot \OPT(\cI)\\
		={} & (1-5 \cdot \eps) \cdot \OPT(\cI)\\
		\geq{} & (1 - \eps') \cdot \OPT(\cI).\\
	\end{aligned}
\end{equation}
 Overall, by \eqref{eq:sB3} the proof follows. 
\end{proof} We can now give the proof of \Cref{thm:fewBins}. 

\noindent {\bf Proof of \Cref{thm:fewBins}:} The statement of the lemma follows from \Cref{lem:time} and \Cref{lem:approx}.
\qed

%% file: app_greedy.tex
CMK naturally arises in the assignment of virtual machines (VMs) to hosts in the cloud computing environment~\cite{camati2014solving}. Consider a data center consisting of $m$ hosts (physical machines). Each host has a set of processors sharing a limited memory (bin capacity)~\cite{rajan2011cloud,othman2013survey,varghese2018next}. The data center administrator has a queue of client requests to assign {\em virtual machines} (VMs) to the hosts. Each VM is associated with a memory demand (weight) and requires the use of a processor; the VM also has a profit (value) that is gained from its execution. As each processor can be assigned to one VM at a time, this induces a bound on the number of VMs that can be assigned 
simultaneously to each host (cardinality constraint). The objective is to assign a subset of the VMs to the hosts subject to the memory and processor availability constraints, such that the total revenue from assigned VMs is maximized. 
Another application of CMK comes from 
frequency assignment for radio networks~\cite{song2008multiple}.

%% file: main.bbl
\begin{thebibliography}{10}
\providecommand{\url}[1]{\texttt{#1}}
\providecommand{\urlprefix}{URL }
\providecommand{\doi}[1]{https://doi.org/#1}

\bibitem{BEK16}
Bansal, N., Eli{\'{a}}s, M., Khan, A.: Improved approximation for vector bin
  packing. In: Proc. SODA. pp. 1561--1579. SIAM (2016)

\bibitem{BK14}
Bansal, N., Khan, A.: Improved approximation algorithm for two-dimensional bin
  packing. In: Proceedings of the twenty-fifth annual ACM-SIAM symposium on
  discrete algorithms. pp. 13--25. SIAM (2014)

\bibitem{BCS09}
Boucheron, S., Lugosi, G., Massart, P.: A sharp concentration inequality with
  applications. Random Structures \& Algorithms  \textbf{16}(3),  277--292
  (2000)

\bibitem{camati2014solving}
Camati, R.S., Calsavara, A., Lima~Jr, L.: Solving the virtual machine placement
  problem as a multiple multidimensional knapsack problem. ICN 2014
  \textbf{264} (2014)

\bibitem{caprara2000approximation}
Caprara, A., Kellerer, H., Pferschy, U., Pisinger, D.: Approximation algorithms
  for knapsack problems with cardinality constraints. Eur. J. Oper. Res.
  \textbf{123}(2),  333--345 (2000)

\bibitem{chekuri2005polynomial}
Chekuri, C., Khanna, S.: A polynomial time approximation scheme for the
  multiple knapsack problem. SIAM Journal on Computing  \textbf{35}(3),
  713--728 (2005)

\bibitem{CT97}
Chow, Y.S., Teicher, H.: Probability theory: independence, interchangeability,
  martingales. Springer Science \& Business Media (1997)

\bibitem{CKS23}
Cohen, T., Kulik, A., Shachnai, H.: Improved approximation for two-dimensional
  vector multiple knapsack. In: 34th International Symposium on Algorithms and
  Computation, {ISAAC}. pp. 20:1--20:17 (2023)

\bibitem{EL10}
Epstein, L., Levin, A.: Afptas results for common variants of bin packing: A
  new method for handling the small items. SIAM Journal on Optimization
  \textbf{20}(6),  3121--3145 (2010)

\bibitem{fleischer2011tight}
Fleischer, L., Goemans, M.X., Mirrokni, V.S., Sviridenko, M.: Tight
  approximation algorithms for maximum separable assignment problems. Math.
  Oper. Res.  \textbf{36}(3),  416--431 (2011)

\bibitem{grandoni2010approximation}
Grandoni, F., Zenklusen, R.: Approximation schemes for multi-budgeted
  independence systems. In: Proc. ESA. pp. 536--548 (2010)

\bibitem{jansen2010parameterized}
Jansen, K.: Parameterized approximation scheme for the multiple knapsack
  problem. SIAM Journal on Computing  \textbf{39}(4),  1392--1412 (2010)

\bibitem{jansen2012fast}
Jansen, K.: A fast approximation scheme for the multiple knapsack problem. In:
  Proc. SOFSEM. pp. 313--324 (2012)

\bibitem{karmarkar1982efficient}
Karmarkar, N., Karp, R.M.: An efficient approximation scheme for the
  one-dimensional bin-packing problem. In: 23rd Annual Symposium on Foundations
  of Computer Science (sfcs 1982). pp. 312--320. IEEE (1982)

\bibitem{kellerer1999polynomial}
Kellerer, H.: A polynomial time approximation scheme for the multiple knapsack
  problem. In: RANDOM-APPROX. pp. 51--62. Springer (1999)

\bibitem{KMS23}
Kulik, A., Mnich, M., Shachnai, H.: Improved approximations for vector bin
  packing via iterative randomized rounding. In: 64th {IEEE} Annual Symposium
  on Foundations of Computer Science, {FOCS}. pp. 1366--1376 (2023)

\bibitem{FL81}
Fernandez~de La~Vega, W., Lueker, G.S.: Bin packing can be solved within 1+
  $\varepsilon$ in linear time. Combinatorica  \textbf{1}(4),  349--355 (1981)

\bibitem{lawler2001combinatorial}
Lawler, E.L.: Combinatorial optimization: networks and matroids. Courier
  Corporation (2001)

\bibitem{li2022faster}
Li, W., Lee, J., Shroff, N.: A faster {FPTAS} for knapsack problem with
  cardinality constraint. Discrete Applied Mathematics  \textbf{315},  71--85
  (2022)

\bibitem{othman2013survey}
Othman, M., Madani, S.A., Khan, S.U., et~al.: A survey of mobile cloud
  computing application models. IEEE communications surveys \& tutorials
  \textbf{16}(1),  393--413 (2013)

\bibitem{rajan2011cloud}
Rajan, S., Jairath, A.: Cloud computing: The fifth generation of computing. In:
  Proc. CSNT. pp. 665--667. IEEE (2011)

\bibitem{schrijver2003combinatorial}
Schrijver, A., et~al.: Combinatorial optimization: polyhedra and efficiency.
  No.~2, Springer (2003)

\bibitem{song2008multiple}
Song, Y., Zhang, C., Fang, Y.: Multiple multidimensional knapsack problem and
  its applications in cognitive radio networks. In: Proc. MILCOM. pp.~1--7.
  IEEE (2008)

\bibitem{varghese2018next}
Varghese, B., Buyya, R.: Next generation cloud computing: New trends and
  research directions. Future Generation Computer Systems  \textbf{79},
  849--861 (2018)

\bibitem{vazirani2001approximation}
Vazirani, V.V.: Approximation algorithms, vol.~1. Springer (2001)

\end{thebibliography}
